\documentclass{mscs}
\usepackage[backend=bibtex,maxbibnames=9,style=alphabetic,natbib=true]{biblatex}
\usepackage{amsmath, amssymb, amsfonts, stmaryrd}
\usepackage{mathpartir}
\usepackage{bussproofs}
\usepackage{xparse}
\usepackage{hyperref}
\usepackage[usenames, dvipsnames]{xcolor}
\usepackage{tikz-cd}
\usepackage[disable]{todonotes}

\DeclareSymbolFont{sfoperators}{OT1}{cmss}{m}{n}
\DeclareSymbolFontAlphabet{\mathsf}{sfoperators}

\makeatletter
\def\operator@font{\mathgroup\symsfoperators}
\makeatother

\newcommand{\gdtt}{\ensuremath{\mathsf{GDTT}}}

\newcommand{\propeq}{=}
\newcommand{\judgeq}{\equiv}
\newcommand{\judgeqty}{\equiv}

\renewcommand{\gets}{\shortleftarrow}
\newcommand{\subst}[2]{[#1/#2]}
\newcommand{\esubst}[2]{\hrt{#1 \gets #2}} 

\DeclareDocumentCommand{\laterbare}{o}
{\IfNoValueTF{#1}
  {\mathord{\triangleright}}
  {\mathord{\triangleright}^{#1}}}

\DeclareDocumentCommand{\later}{ m o m }{
  \IfNoValueTF{#2}
  {\mathord{\triangleright}#3}
  {\mathord{\triangleright}#2 . #3}}
\newcommand{\laterclock}[1]{\mathord{\triangleright}_#1}

\newcommand{\hrt}[1]{\left[ #1 \right]}
\newcommand{\laterold}{\mathord\triangleright}
\newcommand{\latercode}[1]{\code{\triangleright}^{#1}}
\newcommand{\latercodebare}{\code{\triangleright}}
\newcommand{\U}[1]{\mathcal{U}_{#1}}

\newcommand{\El}{\operatorname{El}}

\newcommand{\depprod}[3]{\ensuremath{{\Pi{\left(#1 : #2\right)}.#3}}}

\newcommand{\alwaystype}[2]{\ensuremath{\forall{#1}\ifstrempty{#2}{}{.#2}}}

\newcommand{\alwaysterm}[2]{\ensuremath{\Lambda{#1}\ifstrempty{#2}{}{.#2}}}
\newcommand{\alwaysapp}[2]{\ensuremath{#1\!\left[#2\right]}} 

\newcommand{\code}[1]{\widehat{#1}}
\newcommand{\codeop}[1]{\mathop{\widehat{#1}}}

\newcommand{\purebare}{\operatorname{next}}
\DeclareDocumentCommand{\pure}{ m o m }{
  \IfNoValueTF{#2}
  {\operatorname{next}#3}
  {\operatorname{next} #2 . #3}}

\DeclareDocumentCommand{\app}{o}
{\IfNoValueTF{#1}
  {\ensuremath{\mathbin{\circledast}}}
  {\ensuremath{\mathbin{\ensuremath{\text{\textcircled{\scalebox{0.7}{\ensuremath{#1}}}}}}}}}
\newcommand{\prev}[2]{\operatorname{prev} #1 . #2}

\newcommand{\pair}[2]{\left\langle #1, #2 \right\rangle}
\DeclareDocumentCommand{\fixcombinator}{o}
{\IfNoValueTF{#1}
  {\operatorname{fix}}
  {\operatorname{fix}^{#1}}
}

\newcommand{\inl}{\operatorname{inl}}
\newcommand{\inr}{\operatorname{inr}}

\DeclareDocumentCommand{\force}{o}
{
  \IfNoValueTF{#1}
  {\ensuremath{\operatorname{force}}}
  {\ensuremath{\operatorname{force}\left(#1\right)}}
}

\newcommand{\wfctx}[2]{\ensuremath{#2 \vdash_{#1}}}
\newcommand{\wftype}[3]{\ensuremath{#2 \vdash_{#1} #3 \, \operatorname{type}}}
\newcommand{\hastype}[4]{\ensuremath{#2 \vdash_{#1} #3 : #4}}

\newcommand{\Sl}{\ensuremath{\mathcal{S}}}

\newtheorem{theorem}{Theorem}[section] 
\usepackage{chngcntr}
\newtheorem{definition}[theorem]{Definition} 

\newtheorem{proposition}[theorem]{Proposition}
\newtheorem{lemma}[theorem]{Lemma}

\newtheorem{corollary}[theorem]{Corollary}

\newcommand{\gstream}[1]{\ensuremath{\operatorname{Str}_g^{#1}}}

\newcommand{\den}[1]{\ensuremath{\left\llbracket #1 \right\rrbracket}}

\newenvironment{diagram}{\begin{tikzcd}[row sep=1.5cm,column sep=1.5cm]}{\end{tikzcd}}

\newcommand{\eqdef}{\overset{\mathrm{def}}{=\joinrel=}}

\newcommand{\denglob}[1]{\llbracket #1 \rrbracket^{\glob}}

\newcommand{\jud}[3]{#1 \vdash #2 : #3}

\newcommand{\pcffix}[2]{\operatorname{Y}_{#1} \ #2 }

\newcommand{\pcfzero}[0]{\operatorname{zero} }

\newcommand{\pcfsucc}[1]{\operatorname{succ} \ #1 }

\newcommand{\pcfifzbare}{\operatorname{ifz}}
\newcommand{\pcfifz}[3]{\operatorname{ifz} \ #1 \ #2 \ #3 }

\newcommand{\pcfnattype}[0]{\textbf{\textup{nat}}}

\newcommand{\nowbare}{\eta}

\newcommand{\now}[1]{\nowbare}
\newcommand{\tickbare}{\theta}

\newcommand{\tick}[2]{\tickbare_{#2}}

\newcommand{\delaybare}{\delta}
\newcommand{\delayglob}[1]{\delta^{\glob}_{#1}}

\DeclareDocumentCommand{\delay}{m m o}
{\IfNoValueTF{#3}
  {\delaybare_{#2}}
  {\delaybare_{#2}^{#3}}
}

\newcommand{\ctx}{\operatorname{Ctx}}

\newcommand{\zeromany}[2]{#1 \to^0_* #2}

\newcommand{\ctxeq}{\ \approx_{\texttt{CTX}} \ }

\newcommand{\wb}[1]{\ \approx_{#1} \ }
\newcommand{\wbisim}[2]{\ \approx_{#2} \ }
\newcommand{\wbglob}[1]{\ \approx^{\glob}_{#1} \ }
\newcommand{\laterWBisim}[2]{\ \later{#1}{\approx_{#2}}\ }

\newcommand{\R}[1]{\ \mathcal{R}_{#1}\ }
\newcommand{\Rbare}{\mathcal{R}}
\newcommand{\laterR}[1]{\ \laterbare\mathcal{R}_{#1}\ }

\newcommand{\ctxhastype}[5]{#1 : (#2, #3) \to (#4, #5)}

\newcommand{\FPCOTerms}{\ensuremath{\texttt{\textup{OTerm}}_{\texttt{\tiny\textup{FPC}}}}}
\newcommand{\FPCTerms}{\ensuremath{\texttt{\textup{Term}}_{\texttt{\tiny\textup{FPC}}}}}

\newcommand{\FPCValues}{\ensuremath{\texttt{\textup{Value}}_{ \texttt{\tiny \textup{FPC} }}}}
\newcommand{\FPCTypes}{\ensuremath{\texttt{\textup{Type}}_{ \texttt{\tiny \textup{FPC} }}}}
\newcommand{\fpcfst}[1]{\textup{\texttt{fst}}\;#1}
\newcommand{\fpcsnd}[1]{\textup{\texttt{snd}}\;#1}
\newcommand{\fpcinl}[1]{\textup{\texttt{inl}}\;#1}
\newcommand{\fpcinr}[1]{\textup{\texttt{inr}}\;#1}
\newcommand{\fpccasebare}{\textup{\texttt{case}}}
\newcommand{\fpccaseopen}[5]{
        \textup{\texttt{case}}\;#1\;
        \textup{\texttt{of}}\;#2.#3;#4.#5}
\newcommand{\fpccase}[3]{
        \fpccaseopen{#1}{x_1}{#2}{x_2}{#3}}
\newcommand{\fpcfold}[1]{\fpcfoldbare\;#1}
\newcommand{\fpcunfold}[1]{\fpcunfoldbare\;#1}
\newcommand{\fpcfoldbare}{\textup{\texttt{fold}}}
\newcommand{\fpcunfoldbare}{\textup{\texttt{unfold}}}
\newcommand{\fpcpair}[2]{\langle #1, #2 \rangle}
\newcommand{\fpcunit}{\langle \rangle}
\newcommand{\fpcunittype}{1}
\newcommand{\foldedtype}{\mu \alpha.\tau}
\newcommand{\unfoldedtype}{\tau[\mu \alpha.\tau/\alpha]}

\newcommand{\sfold}{{\textsf{fold}}}

\newcommand{\sunit}{*}

\newcommand{\sinl}{\inl}
\newcommand{\sinr}{\inl}
\newcommand{\glob}{\mathrm{gl}}
\newcommand{\Lglob}{L^{\glob}}
\newcommand{\liftrelglob}{\; L^{\glob} R \;}

\newcommand{\denglobal}[1]{ \den{#1}^{\glob}}

\newcommand{\bigstep}{\Downarrow}
\newcommand{\many}[3]{#1 \to^{#2}_{*} #3}
\newcommand{\manyk}[3]{#1 \Rightarrow^{#2} #3}
\newcommand{\tozero}[2]{#1 \to^0 #2}
\newcommand{\toone}[2]{#1 \to^1 #2}

\newcommand{\liftrel}{\;LR\;}
\newcommand{\laterliftrel}{\;\laterbare LR\;}

\newcommand{\gttfst}{\pi_1}
\newcommand{\gttsnd}{\pi_2}

\newcommand{\gttpair}[2]{\langle #1,#2 \rangle}

\newcommand{\co}{\colon}
\newcommand{\ld}{.}

\newcommand{\N}{\mathbb{N}}

\newcommand{\gnl}{\\[2ex]} 

\newcommand{\NN}{\ensuremath{\mathbb{N}}}
\newcommand{\copair}[2]{[#1,#2]}
\newcommand{\exec}{\ensuremath{\operatorname{exec}}}
\newcommand{\runstep}{\ensuremath{\operatorname{runstep}}}

\newcommand{\rasmus}[1]{\todo[inline,backgroundcolor=yellow!30]{#1}}

\newenvironment{proofof}[1]
{\begin{proof}[Proof of {#1}]}
{\end{proof}}

\usepackage[autostyle=true]{csquotes} 

\title{Denotational semantics of recursive types in synthetic guarded domain theory}
\author[R. E. M{\o}gelberg and M.Paviotti]
    {R\ls A\ls S\ls M\ls U\ls S\ns E.\ns M\ls {\O}\ls G\ls E\ls L\ls
      B\ls E\ls R\ls G$^1$ \thanks{This research was supported by The Danish Council for Independent Research for the Natural Sciences (FNU), Grant no. 4002-00442.} \ns and \ns M\ls A\ls R\ls C\ls O\ns P\ls A\ls V\ls I\ls O\ls T\ls T\ls I$^2$\thanks{Marco Paviotti was funded in part by EPSRC grant EP/M017176/1.}\\
      $^1$ IT University of Copenhagen,   Copenhagen, Denmark.
      \addressbreak $^2$ University of Kent, Canterbury, United Kingdom.}
\date{\today}

\begin{document}

\label{firstpage}
\maketitle

\begin{abstract}
Just like any other branch of mathematics, denotational semantics 
of programming languages should be formalised in type theory, but
adapting traditional domain theoretic semantics, as originally formulated
in classical set theory to type theory has proven challenging. This paper
is part of a project on formulating denotational semantics in type theories
with guarded recursion. This should have the benefit of not only giving 
simpler semantics and proofs of properties such as adequacy, but also 
hopefully in the future to scale to languages with advanced features, such
as general references, outside the reach of traditional domain theoretic 
techniques. 


%
%
%

Working in \emph{Guarded Dependent Type Theory} (\gdtt), we develop
denotational semantics for FPC, the simply typed lambda calculus
extended with recursive types, modelling the recursive types of 
FPC using the guarded recursive types of $\gdtt$. We prove soundness
and computational adequacy of the model in $\gdtt$ using a logical relation
between syntax and semantics constructed also using guarded recursive
types. The denotational semantics is intensional in the sense that it counts
the number of unfold-fold reductions needed to compute the value of a term,
but we construct a relation relating the denotations of extensionally equal 
terms, i.e., pairs of terms that compute the same value in a different number
of steps. Finally we show how the denotational semantics of terms can be 
executed inside type theory and prove that executing the 
denotation of a boolean term computes the same value as the operational 
semantics of FPC. 
%
%
%
\end{abstract}

\tableofcontents

\section{Introduction}
Recent years have seen great advances in formalisation of mathematics
in type theory, in particular with the development of homotopy type
theory~\cite{hottbook}.  Such formalisations are an important step
towards machine assisted verification of mathematical proofs. Rather
than adapting classical set theory-based mathematics to type theory,
new synthetic approaches sometimes offer simpler and clearer
presentations in type theory. As an example of the synthetic approach,
consider synthetic homotopy theory~\cite{hottbook}, which formalises
homotopy theory in type theory, not by formalising a topological space
as a type with structure, but rather by thinking of types as
topological spaces directly. Particular spaces such as the circle can then 
be constructed as types using higher inductive types. 
Synthetic homotopy theory can be formally
related to classical homotopy theory via the simplicial sets
interpretation of homotopy type theory~\cite{Simplicial:model}, interpreting 
types essentially as topological spaces. 

Just like any other branch of mathematics, domain theory and
denotational semantics for programming languages with recursion should
be formalised in type theory and, as was the case of homotopy theory,
synthetic approaches can provide clearer and more abstract proofs.  In
the case of domain theory, the synthetic approach means treating types
as domains, rather than constructing domains internally in type theory
as types with an order relation. The result of this should be a
considerable simplification of denotational semantics when expressed
in type theory.  For example, function types of a higher-order object
language can be modelled simply as the function types of type theory,
rather than as some type of Scott continuous maps.  To model
recursion, some form of fixed point construction must be added to type
theory, but, as is well known, an unrestricted fixed point combinator
makes the logical reading of type theory inconsistent.


\subsection{Synthetic guarded domain theory}
\label{sec:fpc:intro:sgdt}

%
In this paper we follow the approach of guarded
recursion~\cite{Nak00}, which introduces a new type constructor
$\laterbare$, pronounced ``later''. Elements of $\laterbare A$ are to
be thought of as elements of type $A$ available only one time step
from now, and the introduction form $\purebare\co A \to \laterbare A$
makes anything available now, also available later.  The fixed point
operator has type
\[
  \fixcombinator\co (\laterbare A \to A) \to A
\]
and maps an $f$ to a fixed point of $f\circ \purebare$. Guarded
recursion also assumes solutions to all guarded recursive type
equations, i.e., equations where all occurences of the type variable
are under a $\laterbare$, as for example in the equation
\begin{equation} \label{eq:lift:example} LA \cong A + \laterbare LA
\end{equation}
used to define the lifting monad $L$ below, but guarded recursive
equations can also have negative or even non-functorial occurences.

One application of guarded recursion is for programming with
coinductive types.  This requires a notion of clocks used to index
delays. For example, if $\kappa$ is a clock and $A$ is a type then
$\laterclock\kappa A$ is a type.  If $\kappa$ is a clock variable not
free in $A$ and $LA \cong A + \laterclock\kappa LA$, then $\kappa$ can
be universally quantified in $LA$ to give the type
$\forall\kappa . LA$ which can be shown to be a coinductive solution
to $\forall\kappa . LA \cong A + \forall\kappa . LA$. Almost
everything we do in this paper uses a single implicit clock variable
and all uses of $\laterbare$ should be thought of as indexed by this
clock. More details can be found in
Section~\ref{sec:guarded-recursion-intro}.


Recent work has shown how guarded recursion can be used to construct
syntactic models and operational reasoning principles for (also
combinations of) advanced programming language features including
general references, recursive types, countable non-determinism and
concurrency~\cite{BMSS12,BBM14,SB14}. These models often require
solving recursive domain equations which are beyond the reach of
domain theoretic methods. When viewing these syntactic models through
the topos of trees model of guarded recursion~\cite{BMSS12} one
recovers step-indexing~\cite{App01}, a technique for sidestepping
recursive domain equations by indexing the interpretation of types by
numbers, counting the number of unfoldings of the equation.  Thus
guarded recursion can be more accurately described as synthetic
step-indexing. Indeed, guarded recursion provides a type system for
constructing step-indexed models, in which the type equations
sidestepped by step-indexing can be solved using guarded recursive
types.


This work is part of a programme of developing \emph{denotational
  semantics} using guarded recursion with the expectation that this
will not only be simpler to formalise in type theory than the
classical domain theoretic semantics, but also generalise to languages
with advanced features for which step-indexing has been used for
operational reasoning.  This programme was initiated in previous
work~\cite{PMB15}, in which a model of PCF (simply typed lambda
calculus with fixed points) was developed in Guarded Dependent Type
Theory (\gdtt)~\cite{BGCMB16} an extensional type theory with guarded
recursive types and terms.  By aligning the fixpoint unfoldings of PCF
with the steps of the metalanguage (represented by $\laterbare$), we
proved a computational adequacy result for the model inside type
theory.  Guarded recursive types were used both in the denotational
semantics (to define a lifting monad) and in the proof of
computational adequacy.  Likewise, the fixed point operator
$\fixcombinator$ of $\gdtt$ was used both to model fixed points of PCF
and as a proof principle.

\subsection{Contributions}
Here we extend our previous work in two ways. First we extend the
denotational semantics and adequacy proof to languages with recursive
types.  Secondly, we define a relation capturing extensionally equal
elements in the model.

More precisely, we consider the language FPC (simply typed lambda
calculus extended with general recursive types) with a call-by-name
operational semantics. Working internally in $\gdtt$ this language can
be given a denotational semantics in the synthetic style discussed
above.  In particular, function types of FPC are interpreted simply as
the function types of $\gdtt$. Base types are interpreted using the
lifting monad $L$ satisfying the isomorphism
(\ref{eq:lift:example}). In particular the unit type of FPC is
interpreted as $L1$ isomorphic to $1 + \laterbare L1$, so that
denotationally, a program of this type is either a value now, or a
delayed computation. Recursive types are modelled as guarded recursive
types satisfying the isomorphism
\[\den{\mu\alpha.\sigma} \cong \laterbare
  \den{\sigma\subst{\mu\alpha.\sigma}\alpha}
\]
(in the case of closed types). This means that the introduction rule
for recursive types (folding a term) can be interpreted as
$\purebare$.  To interpret unfolding of terms of recursive types we
construct, for every FPC type $\sigma$ a map
$\tick{}{\sigma} : \laterbare \den\sigma \to \sigma$, and interpret
unfolding as $\tick{}{\sigma\subst{\mu\alpha.\sigma}\alpha}$.  As a
consequence, folding followed by unfolding is interpreted as the map
$\delay{}{\sigma\subst{\mu\alpha.\sigma}\alpha}$ defined as
$\tick{}{\sigma\subst{\mu\alpha.\sigma}\alpha} \circ \purebare$. This
composition is not the identity, rather the denotational semantics
counts the number of fold-unfold reductions needed to evaluate a term
to a value.

Thus, to state a precise soundness theorem, the operational semantics
also needs to count the fold-unfold reductions. To do this, we define
a judgement $\many{M}k{N}$ to mean that $M$ reduces to $N$ in a sequence of
reductions containing exactly $k$ fold-unfold reductions, and an
equivalent big-step semantics $M \bigstep^kv$. One might hope to
formulate an adequacy theorem stating that for $M$ of type $1$,
$M \bigstep^k\fpcunit$ (where $\fpcunit$ is the introduction form for
$1$) if and only if $\den M = \delta{}{}^k\den
\fpcunit$. Unfortunately this is not true. For example, if
$M \bigstep^2 \fpcunit$ the type $M \bigstep^1 \fpcunit$ is empty, but
the identity type $\den M = \delta{}{}^1\den \fpcunit$ is equivalent
to $\laterbare 0$, a non-standard truth value different from $0$.  To
state an exact correspondence between the operational and denotational
semantics we use the \emph{guarded transitive closure of the
  small-step semantics} which \emph{synchronises} the steps of FPC
with those of $\gdtt$.  This is defined as $\manyk{M}{k+1}{N}$ if
$\zeromany{M}{M'}$, $M' \to^1 M''$ and
$\later{}{(\manyk{M''}{k}{N})}$, where $M' \to^1 M''$ is a fold-unfold
reduction in an evaluation context.

The adequacy theorem states that $\manyk Mk\fpcunit$ if and only if
$\den M = \delta{}{}^k\den \fpcunit$. We prove this working inside
$\gdtt$, and the proof shows an interesting aspect of guarded domain
theory: It uses a logical relation between syntax and semantics
defined by induction over the structure of types. The case of
recursive types requires a solution to a recursive type equation. In
the setting of classical domain theory, the existence of this solution
requires a separate argument~\cite{Pit96}, but here it is simply a
guarded recursive type.

%
%

The second contribution is a relation capturing extensionally equal
elements in the model. As mentioned above, the denotational semantics
distinguishes between computations computing the same value in a
different number of steps. In this paper we construct a relation on
the denotational semantics of each type relating elements
extensionally equal elements, i.e., elements that compute the same
value in a different number of steps. This relation is defined on the
\emph{global interpretation of types} $\denglob\sigma$ defined from
$\den\sigma$ by quantifying over the implicit clock variable (see
Section~\ref{sec:fpc:intro:sgdt} above). This is necessary, because,
as can be seen from the denotational semantics of guarded recursion,
any relation on $\den 1$ relating $\den\fpcunit$ to
$\delay{}{}^n\den\fpcunit$ for any $n$ will also necessarily relate
non-termination to $\den\fpcunit$. On the other hand, it is possible
to define such a relation on $\denglob \fpcunittype$ which is the
coinductive solution to
$\denglob \fpcunittype \cong 1 + \denglob \fpcunittype$. This is then
lifted to function types in the usual way for logical relations: Two
functions are related it they map related elements to related
elements, and to recursive types using a solution to a guarded
recursive type equation.  We prove a soundness result for this
relation stating that if the (global) denotation of two terms are
related, then the terms are contextually equivalent.


Finally we show that it is possible to execute the denotational
semantics. Of course, FPC is a non-total programming language, so to
run FPC programs in type theory, these must be given a time-out to
ensure termination. We demonstrate the technique in the case of
boolean typed programs and show that the denotation of a program
executes to true with a time-out of $n$ steps if and only if the
program evaluates to true in less than $n$ steps in the operational
semantics.

All constructions and proofs are carried out working informally in
$\gdtt$.  This work illustrates the strength of $\gdtt$ , and indeed
influenced the design of the type theory.

\subsection{Related work}
Escard\'{o} constructs a model of PCF using a category of ultrametric
spaces~\cite{Esc99}.  Since this category can be seen as a subcategory
of the topos of trees~\cite{BMSS12}, our previous work on PCF is a
synthetic version of Escard\'{o}'s model. Escard\'{o}'s model also
distinguishes between computations computing the same value in a
different number of steps, and captures extensional behaviour using a
logical relation similar to the one constructed here. Escard\'{o}
however, does not consider recursive types.  Although Escard\'{o}'s
model was useful for intuitions, the synthetic construction in type
theory presented here is very different, in particular the proof of
adequacy, which here is formulated in guarded dependent type theory.
%

Synthetic approaches to domain theory have been developed based on a
wide range of models dating back to~\cite{Hyl91, Ros86}.  Indeed, the
internal languages of these models can be used to construct models of
FPC and prove computational adequacy~\cite{Sim02}.  A more axiomatic
approach was developed in Reus's work~\cite{Reu96} where an
axiomatisation of domain theory is postulated a priori inside the
Extended Calculus of Constructions.

There has also been work on (non-synthetic) adaptations of domain
theory to type theory~\cite{BKV09,BBKV10,Dockins14}. However, due to
the mistmatch between set-theory and type theory ``\emph{some of the
  proofs and constructions are much more complex than they would
  classically and one does sometimes have to pay attention to which of
  two classically-equivalent forms of definition one works with}''
~\cite{BKV09}.  More recently Altenkirch et al.~\cite{ADK17} have
shown how to encode the free pointed $\omega$-cpo as a quotient
inductive-inductive types (QIIT). This looks like a more promising
direction for domain theory in type theory, but this has not yet been
developed to models of programming languages.


The lifting monad used in this paper is a \emph{guarded} recursive
variant of Capretta's delay monad~\cite{Cap05} considered by among
others~\cite{BKV09, BBKV10, Dan12,CUV15, ADK17, Veltri17}. The monad
$D(A)$ is coinductively generated by the constructors
$\texttt{now} : A \to D(A)$ and $\texttt{later} : D(A) \to D(A)$. As
reported by Danielsson~\cite{Dan12}, working with the partiality monad
requires convincing Agda of productivity of coinductive definitions
using workarounds. In this paper productivity is ensured by the type
system for guarded recursion.

In the delay monad, two computations of type $D(A)$ can be
distinguished by their number of steps. To address this issue,
Capretta also defines a weak bisimulation on this monad, similar to
the one defined in Definition~\ref{def:fpc:liftrel}, and proves the
combination of the delay monad with the weak bisimulation is a monad
using setoids. Chapman et al.\cite{CUV15, Veltri17} avoid using
setoids, but they crucially rely on proposition extensionality and the
axiom of countable choice. Altenkirch et al.~\cite{ADK17} show that under 
the assumption of countable choice, their free pointed $\omega$-cpo 
construction is equivalent to quotiented delay monad of Chapman et al. 
We work crucially with the non-quotiented delay monad when defining the 
denotational semantics, since the steps are necessary for guarded recursion.


This is an extended version of a conference publication~\cite{MP16}. A
number of proofs that were omitted from the previous version due to
space restrictions have been included in this version. There is also a
slight difference in approach: the conference version defined a
big-step operational semantics equivalent to the guarded transitive
closure of the small-step operational semantics of
Figure~\ref{fig:op:sem} below. This operational semantics synchronises
the steps of FPC with those of the meta-language, and capturing this
in a big-step semantics was quite tricky. Here, instead, we define a
simpler big-step operational semantics and prove this equivalent to
the ``global'' small-step semantics
(Lemma~\ref{lem:fpc:bigstep:many:manyk:soundness}).  The results on
executing the denotational semantics presented in
Section~\ref{sec:executing} are also new.

Since this work was carried out, the extensional type theory $\gdtt$
that we work in in this paper has been extended in two directions
towards intensionality and implementation. The first direction is
Guarded Cubical Type Theory~\cite{BBCGV16}, extending the fragment of
$\gdtt$ without universal quantification over clocks with
constructions from Cubical Type Theory~\cite{CCHM15}.  Guarded Cubical
Type Theory even has a prototype implementation.  The other direction
is Clocked Type Theory~\cite{clott}, a variant of the fragment of
$\gdtt$ without identity types in which delayed substitutions
(Section~\ref{sec:gdtt}) are encoded using a new notion of ticks on a
clock.  Clocked Type Theory has a strongly normalising reduction
semantics.  Since neither theory is complete, we stick to $\gdtt$ as
our type theory for this paper.

\paragraph{The paper is organized as follows.}
Section~\ref{sec:guarded-recursion-intro} gives a brief introduction
to the most important concepts of $\gdtt$. More advanced constructions of
the type theory are introduced as needed. Section~\ref{sec:FPC}
defines the encoding of FPC and its operational semantics in $\gdtt$. The
denotational semantics is defined and soundness is proved in
Section~\ref{sec:fpc:den}. Computational adequacy is proved in
Section~\ref{sec:fpc:adequacy}, and the relation capturing extensional
equivalence is defined in
Section~\ref{sec:extensional}. Section~\ref{sec:executing} shows how
to execute the denotational semantics of boolean programs.  We
conclude and discuss future work in Section~\ref{sec:conclusions}.

\paragraph{Acknowledgements.} We thank Nick Benton, Lars Birkedal,
Ale\v{s} Bizjak, and Alex Simpson for helpful discussions and
suggestions.

\section{Guarded recursion}
\label{sec:guarded-recursion-intro}
In this paper we work informally within a type theory with dependent
types, inductive types and guarded recursion.  Although inductive
types are not mentioned in~\cite{BGCMB16} the ones used here can be
safely added -- as they can be modelled in the topos of trees model --
and so the arguments of this paper can be formalised in 
Guarded Dependent Type Theory ($\gdtt$)~\cite{BGCMB16}.  We
start by recalling some core features of this theory, but postpone
delayed substitutions to Section~\ref{sec:gdtt} since these are not needed
for the moment.

When working in type theory, we use $\judgeq$ for judgemental equality
of types and terms and $\propeq$ for propositional equality (sometimes
$\propeq_A$ when we want to be explicit about the type). We also use
$\propeq$ for (external) set theoretical equality.

The core of guarded recursion consists of the type constructor $\laterbare$ 
and the fixed point operator $\fixcombinator : (\laterbare A \to A) \to A$ satisfying
\begin{equation}\label{eq:fixunfold}
\fixcombinator f = f(\purebare(\fixcombinator (f)))
\end{equation}
both introduced in Section~\ref{sec:fpc:intro:sgdt}. Elements of 
type $\laterbare A$ are intuitively elements of type $A$ available one time step from now. 
To illustrate the power of 
the fixed point operator, consider a type of guarded streams $\gstream{}$
satisfying 
\begin{equation} \label{eq:gstream}
\gstream{} \cong \NN \times \laterbare\gstream{}
\end{equation}
This is a guarded recursive type in the sense that the recursion variable 
appears under a $\laterbare$, and its elements are to be thought of as 
streams, whose head is immediately available and whose tails take
one time step to compute. The fixed point operator can be used to define
guarded streams by recursion. For example, the constant stream of a number
$n$ can be defined as $\fixcombinator (\lambda x. \pair nx)$, where the
type isomorphism (\ref{eq:gstream}) is left implicit. Note that the type of the fixed point
operator prevents us from defining elements like $\fixcombinator (\lambda x. x)$, which 
are not productive, in the sense that any element of the stream can be computed in finite
time. In fact, the type $\laterbare \gstream{} \to \gstream{}$ precisely captures productive
recursive stream definitions.

\rasmus{Use identity rather than isomorphism?}

The type constructor $\laterbare$ is an applicative functor in the
sense of~\cite{MP08}, which means that there is 
a ``later application''
$\app \co \laterbare (A \to B) \to \laterbare A \to \laterbare B$
written infix, satisfying
\begin{equation}
  \purebare (f) \circledast \purebare (t) \judgeq \purebare (f(t))
  \label{eq:fpc:next:laterappl}
\end{equation}
among other axioms (see also~\cite{BM13}).  In particular,
$\laterbare$ extends to a functor mapping $f\co A \to B$ to
$\lambda x\co \laterbare A \ld \purebare (f)\circledast x$.  Moreover,
the $\laterbare$ operator distributes over the identity type as
follows
\begin{equation}\laterbare (t \propeq_A u) \judgeq (\purebare t
  \propeq_{\laterbare A} \purebare u)\label{eq:later:id}\end{equation}

Guarded dependent type theory comes with universes in the style of
Tarski. In this paper, we will just use a single universe
$\U{}$. Readers familiar with~\cite{BGCMB16} should think of this as
$\U{\kappa}$, but since we work with a unique clock $\kappa$, we will
omit the subscript.  The universe comes with codes for type
operations, including $\codeop{+}\co \U{}\times \U{} \to \U{}$ for
binary sum types, codes for dependent sums and products, and
$\latercodebare \co \laterbare \U{} \to \U{}$ satisfying
\begin{equation}
  \El(\latercodebare (\purebare (A))) \judgeqty \laterold \El(A)
  \label{eq:latercode}
\end{equation}
where we use $\El(A)$ for the type corresponding to an element
$A\co \U{}$.  The type of $\latercodebare$ allows us to solve
recursive type equations using the fixed point combinator. For
example, if $A$ is small, i.e., has a code $\codeop{A}$ in $\U{}$, the
type equation (\ref{eq:lift:example}) can be solved by computing a
code of $LA$ as
\begin{equation}
  \codeop LA = \fixcombinator(\lambda X\co \laterbare \U{} \ld \codeop +(\codeop A,
  \latercodebare X)) 
  \label{eq:fix:lift:example}
\end{equation}
and then by taking the elements using $\El$. More precisely,
defining $LA$ as $\El(\codeop LA)$, $LA$ 
unfolds to
$\El(\codeop +(\codeop A,\latercodebare (\purebare (\codeop L A))))$
which is equal to $A + \El(\latercodebare (\purebare ( \codeop L A)))$
which is equal to $A + \laterbare LA$.  In this paper, we will only
apply the monad $L$ to small types $A$.

To ease presentation, we will usually not distinguish between types
and type operations on the one hand, and their codes on the other. We will still refer use the
notation $\latercodebare \co \laterbare \U{} \to \U{}$, but write $\laterbare$ for the composition
$\latercodebare \circ \purebare$. We generally leave $\El$
 implicit.

%

\subsection{The topos of trees model}
\label{sec:fpc:ToT}

The topos $\Sl$ of trees is the category of presheaves over $\omega$,
the first infinite ordinal. The category $\Sl$ models guarded
recursion~\cite{BMSS12} and provides useful intuitions, and so we
briefly recall it.

A closed type is modelled as an object of the topos of trees, i.e., as 
a family of sets $X(n)$ indexed by
natural numbers together with \emph{restriction maps}
$r_n^X\co X(n+1) \to X(n)$ as in the following diagram
\begin{equation}
  \label{eq:sobject}
  \begin{diagram}
    X(1) & X(2)\arrow{l}{} & X(3) \arrow{l}{} & X(4)\arrow{l}{} &
    \dots \arrow{l}{}
  \end{diagram}
\end{equation}

A term of type $Y$ in context $x : X$, for $X,Y$ closed types, 
is modelled as a morphism in $\Sl$, i.e., as a family of functions
$f_i : X(i) \to Y(i)$ obeying the \emph{naturality} condition
$f_i \circ r_{i}^{X} = r_{i}^{Y} \circ f_{i+1}$ as in the following
diagram
\begin{equation}
  \label{eq:nattransf}
  \begin{diagram}
    X(i) \arrow{d}{f_i}& X(i+1) \arrow{d}{f_{i+1}}\arrow{l}{r^{X}_i} \\
    Y(i) & Y(i+1)\arrow{l}{r^{Y}_i}
  \end{diagram}
\end{equation}

The $\laterbare$ type operator is modelled as an endofunctor in $\Sl$
such that $\laterbare X(1) = 1$, $\laterbare X(n+1) = X(n)$.
Intuitively, $X(n)$ is the $n$th approximation for computations of
type $X$, thus $X(n)$ describes the type $X$ as it looks if we have
$n$ computational steps to reason about it.

Using the proposition-as-types principle, types like $\laterbare^{3}0$
are non-standard truth values. Following the intuition that $\laterbare^{3}0(n)$
is the type $\laterbare^{3}0$ as it looks, if we have $n$ steps to reason about it,
$\laterbare^{3}0$ is the truth value of
propositions that appear true for $3$ computation steps, but then are
falsified after $4$.  In fact, in the model, $(\laterbare^{3} 0) (3)$
equals $1$, but $(\laterbare^{3} 0) (4)$ equals $0$ zero as depicted
by the following diagram
\begin{equation}
  \label{eq:laterthree}
  \begin{diagram}
    1 & 1\arrow{l}{} & 1 \arrow{l}{} & 0\arrow{l}{} & 0\arrow{l}{}&
    \dots \arrow{l}{}
  \end{diagram}
\end{equation}

\newcommand{\Set}{\textbf{Set}}
\newcommand{\homset}[3]{\text{Hom}_#1(#2,#3)}

\rasmus{I erased this: for an object $X$ we write $X^{\glob}$ for the set 
of global elements of $X$, since this is not used elsewhere}

The \emph{global elements} of a closed type $X$ is 
the set of morphisms from the constant object $1$ to $X$ in $\Sl$. 
This can be thought of as
the \emph{limit} of the sequence of (\ref{eq:sobject}) as a diagram in
$\Set$. 
This construction gives us the \emph{global view} of a type as it
allows us to observe \emph{all the computation at once}. For example,
the global elements of $\laterbare X$ correspond to those of 
$X$ simply by discarding the first component. 
%
Note that objects can have equal sets of global elements
without being isomorphic. In particular $0$ and $\laterbare^n 0$ are not
isomorphic.

For guarded recursive type equations, $X(n)$ describes the $n$th
unfolding of the type equation. For example, fixing an object $A$, the
unique solution to (\ref{eq:lift:example}) is
\[
  LA(n) = 1 + A(1) + \dots + A(n)
\]
with restriction maps defined using the restriction maps of $A$.
In particular, if $A$ is a constant presheaf, i.e., $A(n) = X$ for
some fixed $X$ and $r_n^A$ identities, then we can think of $LA(n)$ as
$\{0,\dots, n-1\}\times X + \{\bot\}$ with restriction map given by
$r_n(\bot) = \bot$, $r_n(n,x) = \bot$ and $r_n(i,x) = (i,x)$ for $i<n$. 
The set of global elements of
$LA$ is then isomorphic to $\N\times X + \{\bot\}$.  In particular, if
$X=1$, the set of global elements is $\bar\omega$, the natural numbers
extended with a point at infinity.

The global elements of $LA$, correspond to the elements of Capretta's
partiality monad~\cite{Cap05} $L^{\glob}$ defined as the coinductive
solution to the type equation
\begin{equation}\label{eq:partiality:monad}
  L^{\glob}A  \cong  A + L^{\glob}A 
\end{equation} 

Similarly, the type of $\gstream{}$ can be modelled as $\gstream{}(n) = \NN^n\times 1$. Note that
if these products associate to the right, we can even model (\ref{eq:gstream}) as an identity. 
The restriction maps of this type are projections, and the 
global elements of this type correspond to streams in the usual sense. 

\subsection{Universal quantification over clocks}
\label{sec:clock-variables}

The type of guarded streams $\gstream{}$ mentioned above, is not the usual 
coinductive type of streams. For example, a term $t : \gstream{}$ in context
$x : \gstream{}$ is a causal function of streams, i.e., one where the $n$ first elements
of the output depend only on the $n$ first elements of the input. This can be seen e.g. 
in the topos of trees model, where such a term is modelled by a family of maps 
$f_n :  \NN^n\times 1 \to  \NN^n\times 1$ commuting with projections. Causality
is crucial to the encoding of productivity in types mentioned above. 

On the other hand, a \emph{closed} term $t : \gstream{}$ is modelled by a global 
element of $\gstream{}$ and thus corresponds to a real stream of numbers. 
Likewise, if $t : \gstream{}$ only depends on a variable $x : \NN$, then $t$ denotes a map 
from the set of natural numbers to that of streams, because the context is modelled
as the constant topos of trees object $\NN$, with restriction maps being identities. More generally, say
a context is \emph{independent of time} if it is modelled as a constant object, i.e, one where all restriction 
maps are isomorphisms. The denotation of a term $t : \gstream{}$ in a context $\Gamma$ independent of time, 
corresponds to a map from $\Gamma(1)$ to the set of streams. 

The idea of independence of time can be captured syntactically using a notion of clocks, and universal 
quantification over these~\cite{AM13}. 
%
We now briefly recall 
this as implemented in \gdtt,  
referring to~\cite{BGCMB16} for details.

In $\gdtt$ all types and terms are typed in a clock context, i.e., a
finite set of names of clocks. For each clock $\kappa$, there is a
type constructor $\laterclock\kappa$, a fixed point combinator, and so
on. 
Each clock carries its own notion of time, and the idea of a context being independent of
time mentioned above, can be captured as a clock not appearing in a context. 

If $A$ is a type in a context where $\kappa$ does not appear, one can
form the type $\forall \kappa. A$, binding $\kappa$. This construction
behaves in many ways similarly to polymorphic quantification over
types in System F. There is an associated binding introduction form
$\alwaysterm{\kappa}{(-)}$ (applicable to terms where $\kappa$ does
not appear free in the context), and elimination form
$\alwaysapp{t}{\kappa'}$ having type $A\subst {\kappa'}{\kappa}$
whenever $t\co \alwaystype{\kappa}{A}$. 

Semantically, a closed type in the empty clock variable context is modelled by a set, and 
a type in a context of a single clock is modelled as an object in the topos of trees.
In the latter case, universal quantification over the single clock is modelled by taking the 
set of global elements. As we saw above, these sets correspond to coinductive types, and this 
also holds in the type theory: 
If $\gstream{}$ is the type of streams guarded on clock $\kappa$, i.e., satisfies 
$\gstream{} \cong \NN \times \laterclock\kappa\gstream{}$, then one can prove~\cite{AM13,Mog14} that
the type $\forall\kappa . \gstream{}$ behaves as a coinductive type of streams. Similarly, if 
$LA \cong A + \laterclock\kappa LA$, and $\kappa$ is not free in $A$, 
then $\forall\kappa . LA$ is a coinductive solution to 
$X \cong A + X$. This isomorphism arises 
as a composite of isomorphisms
\begin{align}
\nonumber \forall\kappa . LA & \cong \forall\kappa . (A + \laterclock\kappa LA) \\
\label{eq:forall:dist:plus} & \cong (\forall\kappa . A)  + (\forall\kappa . \laterclock\kappa LA) \\
\label{eq:forall:const} & \cong A + \forall\kappa. \laterclock\kappa LA \\
\label{eq:forall:later} & \cong A + \forall\kappa. LA
\end{align}
the components of which we recall below.
Using these encodings one can use guarded
recursion to program with coinductive types in such a way that typing guarantees 
productivity. We refer to~\cite{BM15} for a full model of guarded recursion with clocks,
in particular for how to model types with more than one free clock variable. 

The isomorphism (\ref{eq:forall:later}) arises from a general type isomorphism
$\forall\kappa . \laterclock\kappa A \cong \forall\kappa . A$ holding for all $A$.
The direction from right to left is induced by $\pure{}^\kappa : A \to\laterclock\kappa A$. 
For the direction from left to right, a form of elimination for $\laterclock\kappa$ is needed, but note that 
an unrestricted such of type $\laterclock\kappa A \to A$ in combination with fixed points 
makes the type system inconsistent. 
Instead \gdtt\ allows for a restricted elimination rule for
$\laterclock\kappa$: If $t$ is of type $\laterclock{\kappa}{A}$ in a context
where $\kappa$ does not appear free, then $\prev{\kappa}{t}$ has type
$\forall \kappa. A$.  Using $\prev{\kappa}{}$ we can define a term
$\force$:
\begin{equation}
  \begin{aligned}
    \force &: (\alwaystype{\kappa}{\laterclock{\kappa}{A}}) \to \alwaystype{\kappa}{A}\\
    \force &\eqdef \lambda x.\prev{\kappa}{\alwaysapp{x}{\kappa}}
  \end{aligned}\label{gdtt:force}
\end{equation}
The term $\force$ can be proved to be an isomorphism by the axioms 
\begin{equation} \label{eq:prev:beta:eta}
  \prev\kappa \pure{}^\kappa(t) \judgeq \Lambda\kappa . t \qquad \pure{}^\kappa(\alwaysapp{(\prev\kappa t)}{\kappa}) \judgeq t 
\end{equation}

If $\kappa$ is not free in $A$, the type  $\forall \kappa. A$ is isomorphic to
$A$, justifying the isomorphism (\ref{eq:forall:const}). 
The map $A \to \forall \kappa. A$ is simply $\lambda x \co A. \Lambda\kappa.x$. 
The other direction is given by application to a clock constant $\kappa_0$,
which we assume exists. These can be proved to be inverses of each other
using \emph{the clock irrelevance axiom}, which states that if $t : \forall\kappa. A$ and $\kappa$ does not 
appear free in $A$, then $\alwaysapp t{\kappa'} \judgeq \alwaysapp t{\kappa''}$ for all $\kappa'$ and $\kappa''$. 
Using $\force$ and the isomorphism $\forall \kappa. 0\cong 0$, one can prove that
$\forall \kappa. \laterclock\kappa^n 0$ is isomorphic to $0$, reflecting the fact that there are no global elements of 
$\laterbare^n 0$ in the model, as mentioned earlier. 
We refer to~\cite{BGCMB16} for details. 

The isomorphism (\ref{eq:forall:dist:plus}) is a special case of an isomorphism 
\begin{equation} \label{eq:forall:dist:sum}
\forall\kappa . (B + C) \cong (\forall\kappa . B) + (\forall\kappa . C)
\end{equation}
distributing $\forall\kappa$ over sums for all small types $B$ and $C$. To
describe this isomorphism, encode sum types as 
$B + C \eqdef \Sigma x : (1 + 1) . \copair BC(x)$ where $\copair BC$ is defined by cases by
$\copair BC(\inl(\star)) \judgeq B$ and $\copair BC(\inr(\star)) \judgeq C$. The result of applying the
left to right direction $d$ of the isomorphism to $x : \forall\kappa . (B + C)$ is defined by
cases of $\pi_1(\alwaysapp x{\kappa_0}) : 1 + 1$. If $\pi_1(\alwaysapp x{\kappa_0}) = \inl(\star)$, note that
for any $\kappa$, using the clock irrelevance axiom
\[\pi_1(\alwaysapp x{\kappa}) = \alwaysapp{(\Lambda\kappa . \pi_1(\alwaysapp x{\kappa}))}{\kappa} = 
\alwaysapp{(\Lambda\kappa . \pi_1(\alwaysapp x{\kappa}))}{\kappa_0} = \pi_1(\alwaysapp x{\kappa_0}) = \inl(\star)\]
and so $\Lambda\kappa . \pi_2(\alwaysapp x{\kappa})$ has type
\[\forall\kappa . \copair BC(\pi_1(\alwaysapp x{\kappa})) =  \forall\kappa . \copair BC(\inl(\star)) = \forall\kappa . C
\]
and so we can define in this case $d(x) = \Lambda\kappa . \inl(\pi_2(\alwaysapp x{\kappa}))$. The case of 
$\pi_1(\alwaysapp x{\kappa_0}) = \inr(\star)$ is similar.  
%
%
%
In fact, this construction generalises to an isomorphism 
\begin{equation}
  \forall \kappa. \Sigma (x : A). B  \cong \Sigma (x : A). \forall \kappa. B
  \label{eq:fpc:forallk:sigma:swap}
\end{equation}
valid whenever $\kappa$ is not free in $A$. 

Finally we note the following extensionality rule for quantification over clocks. 
\begin{equation}
\label{eq:forall:eta}
(t \propeq_{\forall\kappa . A} u) \judgeq \forall\kappa . (\alwaysapp t\kappa \propeq_A \alwaysapp u\kappa)
\end{equation}

In most of this paper we will work in a setting of a unique implicit clock $\kappa$, and simply write 
$\laterbare$ for $\laterclock\kappa$ to avoid cluttering all definitions and calculations with clocks. 

For the proof of computational adequacy we will need one more construction from $\gdtt$: The delayed substitutions. 
These will be recalled in Section~\ref{sec:gdtt}.

\section{FPC}
\label{sec:FPC}

This section defines the syntax, typing judgements and operational
semantics of FPC. These are inductive types in guarded type theory,
but, as mentioned earlier, we work informally in type theory, and in
particular remain agnostic with respect to choice of representation of
syntax with binding.

The typing judgements of FPC are defined in an entirely standard
way. The grammar for terms of FPC
\begin{align*}
  L,M,N ::= &\, \fpcunit \mid x \mid \fpcinl{M} \mid \fpcinr{M} 
            \mid \fpccase{L}{M}{N} \mid \fpcpair{M}{N} \\ 
            & \mid \fpcfst{M} \mid \fpcsnd{M} 
             \mid \lambda x : \tau .M \mid M N \mid \fpcfold{M} \mid \fpcunfold{N} 
\end{align*}
should be read as an inductive type of terms in the standard way.
Likewise the grammars for types and contexts and the typing judgements
defined in
Figure~\ref{fig:fpc:syntax}
should be read as defining inductive
types in type theory, allowing us to do proofs by induction over
e.g. typing judgements.

We denote by \FPCTypes, {\FPCTerms} and \FPCValues the types of
\emph{closed} FPC types and terms, and values of FPC and by
{\FPCOTerms} the type of all (also open) terms.  By a value we mean a
closed term matching the grammar
\begin{align*}
  v  ::= &\, \fpcunit  \mid \fpcinl{M} \mid
           \fpcinr{M} \mid \langle M, N \rangle
           \mid \lambda x : \tau . M \mid \fpcfold{M}
\end{align*}

\begin{figure}[tb]
\textbf{Well formed types}
  \begin{gather*}
    \Theta \in \text{Type Contexts} \eqdef \langle \rangle \mid \langle \Theta, \alpha \rangle \\
    \frac{}{\vdash \langle \rangle} \quad
    \frac{\vdash \Theta}{\vdash \Theta, \alpha} \alpha \not\in \Theta\\
    \frac{\vdash \Theta}{\Theta \vdash \Theta_i} 1 \le i \le \mid
    \Theta
    \mid \quad \frac{\vdash \Theta}{\Theta \vdash 1}\\
    \frac{\Theta, \alpha \vdash \tau}{\Theta \vdash \mu \alpha. \tau}
    \quad \frac{\Theta \vdash \tau_1 \quad \Theta \vdash
      \tau_2}{\Theta \vdash \tau_1 \texttt{op } \tau_2} \text{ for }
    \texttt{op} \in \{ +, \times, \to \}
  \end{gather*}
\textbf{Typing rules}
  \begin{gather*}
    \frac{x : \sigma \in \Gamma \qquad \cdot \vdash \Gamma }{\Gamma
      \vdash x : \sigma} \qquad \frac{}{\Gamma \vdash \fpcunit : \fpcunittype} \\
    \frac{\jud{\Gamma, x : \sigma}{M}{\tau}}{\jud{\Gamma}{(\lambda x : \sigma. M)}{\sigma \to \tau}} \quad
    \frac{\jud{\Gamma}{M}{\sigma \to \tau} \quad \jud{\Gamma}{N}{\sigma}}{\jud{\Gamma}{M N}{\tau}}\\
    \frac{\jud{\Gamma}{e}{\tau_1}}{\jud{\Gamma}{\fpcinl{e}}{\tau_1 +
        \tau_2}} \quad
    \frac{\jud{\Gamma}{e}{\tau_2}}{\jud{\Gamma}{\fpcinr{e}}{\tau_1 + \tau_2}}\\
    \frac{ \jud{\Gamma}{L}{\tau_1 + \tau_2} \quad \jud{\Gamma, x_1 :
        \tau_1}{M}{\sigma} \quad
      \jud{\Gamma, x_2 : \tau_2}{N}{\sigma}}{\jud{\Gamma}{\fpccase{L}{M}{N}}{\sigma}} \\
    \frac{\jud{\Gamma}{M}{\tau_1 \times
        \tau_2}}{\jud{\Gamma}{\fpcfst{M}}{\tau_1}} \qquad
    \frac{\jud{\Gamma}{M}{\tau_1 \times \tau_2}}{\jud{\Gamma}{\fpcsnd{e}}{\tau_2}}\quad
    \frac{\jud{\Gamma}{M}{\tau_1} \quad \jud{\Gamma}{N}{\tau_2}}{\jud{\Gamma}{\langle M,N\rangle}{\tau_1 \times \tau_2}}\\
    \frac{\jud{\Gamma}{M}{\mu\alpha.\tau}}{\jud{\Gamma}{\fpcunfold{M}}{\tau
        [ \mu \alpha. \tau / \alpha]}} \quad
    \frac{\jud{\Gamma}{M}{\tau [ \mu \alpha. \tau /
        \alpha]}}{\jud{\Gamma}{\fpcfold{M}}{\mu \alpha. \tau}}
  \end{gather*}
  \caption{Syntax of FPC}
  \label{fig:fpc:syntax}
\end{figure}

\subsection{Operational semantics}

\begin{figure}[tb]\textbf{Big-step semantics}
  \begin{gather*}
    \frac{}{v \bigstep^0 v} \\
    \frac{L \bigstep^k \fpcinl{L'} \quad  M[L'/x_1] \bigstep^m v}{\fpccase{L}{M}{N} \bigstep^{m+k} v}\quad \frac{L \bigstep^k \fpcinr{L'} \quad  N[L'/x_2] \bigstep^m v}{\fpccase{L}{M}{N} \bigstep^{m+k} v} \\
    \frac{L \bigstep^k \fpcpair{M}{N} \quad M \bigstep^l v}{\fpcfst{L}
      \bigstep^{k+l} v} \qquad \frac{L \bigstep^k \fpcpair{M}{N}
      \quad N \bigstep^l v}{\fpcsnd{L} \bigstep^{k+l} v}\\
    \frac{M \bigstep^k \lambda x. L \quad L[N/x] \bigstep^l v}{M N
      \bigstep^{k+l} v} \qquad
    \frac{M \bigstep^k \fpcfold{N} \quad N \bigstep^m v}{\fpcunfold{M}
      \bigstep^{k+m+1} v}
  \end{gather*}
\textbf{Small-step semantics}
    \begin{gather*}
    \inferrule*{}{(\lambda x: \sigma. M) (N) \to^0 M[N/x]} \qquad
    \inferrule*{}{\fpcunfold{(\fpcfold{M})} \to^1 M } \\
    \inferrule*{}{\fpccase{(\fpcinl{L})}{M}{N} \to^0 M [L/x_1]} \\
    \inferrule*{}{\fpccase{(\fpcinr{L})}{M}{N} \to^0 N [L/x_2]} \\
    \inferrule*{}{\fpcfst{\langle M, N \rangle} \to^0 M } \qquad
    \inferrule*{}{\fpcsnd{\langle M, N \rangle} \to^0 N } \\
    \qquad \inferrule*{M_1 \to^k M_2 \hspace{.5cm} k = 0,1}{E[M_1] \to^k E[M_2]} \gnl
    E  ::= \, [ \cdot ] \mid  E M \mid \fpccase{E}{M}{N}
    \mid \fpcfst{E} \mid \fpcsnd{E} \mid \fpcunfold{E} \gnl
        \frac{}{\many{M}{0}{M}} \qquad \frac{M \to^k M' \qquad \many{M'}{m}{N}}{\many{M}{k+m}{N}} 
  \end{gather*}
  \textbf{Guarded transitive closure of the small-step semantics}
   \begin{gather*}
     \inferrule*{M \to_*^0 N}{\manyk{M}{0}{N}} \qquad
     \inferrule*{\zeromany{M}{M'} \qquad M' \to^1 M'' \qquad \later{}{(\manyk{M''}{k}{N})}}{\manyk{M}{k+1}{N}}
   \end{gather*}
  \caption{Operational semantics for FPC.}
  \label{fig:op:sem}
\end{figure}
Figure~\ref{fig:op:sem} defines a big-step and a small-step operational semantics for FPC, as well as 
two transitive closures of the latter. All these definitions should be read as inductive types. 
Since the denotational semantics of FPC is intensional, counting reduction steps, it is necessary to also
count the steps in the operational semantics in order to state the
soundness and adequacy theorems precisely. More precisely, the
semantics counts the number of $\fpcunfoldbare$-$\fpcfoldbare$
reductions in the same fashion in which Escard{\'o} counted fix-point
reduction for PCF. 

The statement
\begin{equation} \label{eq:form:of:big:step} M \Downarrow^k v
\end{equation}
where $M$ is a term, $k$ a natural number, and $v$ a value, should be
read as '$M$ evaluates in $k$ steps to a value $v$. We can define more
standard big-step evaluation predicates as follows
\[
  M \Downarrow v \eqdef \Sigma k. M \bigstep^k v
\]


We note that the semantics is trivially deterministic.
\begin{lemma}
  \label{lem:fpc:small-step:det}
  The small-step semantics is deterministic: if $M\to^k N$ and
  $M\to^{k'} N'$, then $k=k'$ and $N= N'$.
\end{lemma}

Of the two transitive closures of the small-step semantics defined in Figure~\ref{fig:op:sem} the first is 
a standard one, equivalent to the big-step operational semantics. The second is a guarded version 
which synchronises the steps of FPC with those of the metalogic.  This is needed for the statement of the soundness
and adequacy theorems, and also allows for guarded recursion to be
used in the proofs of these.
The next lemma states the relationship between the big-step semantics and 
the two transitive closures of the small-step semantics
\begin{lemma}
  Let $M$ and $N$ be FPC terms,  $v$ a value and $k$ a natural number. Then
  \begin{enumerate}
  \item $M \bigstep^k v$ iff $\many{M}{k}{v}$
  \item $\many{M}{k}{N}$ iff $\forall \kappa. \manyk{M}{k}{N}$
  \end{enumerate}
  \label{lem:fpc:bigstep:many:manyk:soundness}
\end{lemma}

Note that in particular  $\many{M}{k}{N}$ implies $\manyk{M}{k}{N}$. The opposite implication does not
hold, as we shall see in the examples below. 

\begin{proof}
  The first statement is a essentially a textbook result on operational semantics, and we omit the proof.

  For the second statement the proof from left to right is by induction on $\many{M}{k}{N}$.
  The case of $M = N$ is trivial, so consider the case when $M \to^k M'$ and
  $\many{M'}{m}{N}$. 
  When $k = 0$, by definition
  $\many{M}{0}{M'}$, and by induction hypothesis we know that
  $\forall \kappa. \manyk{M'}{m}{N}$. Thus, $\manyk{M}{m}{N}$ holds for any $\kappa$, and so
  also $\forall \kappa. \manyk{M}{m}{N}$, since $\kappa$ is not free in the assumption $\many{M}{k}{N}$. 
  When $k = 1$ by induction hypothesis
  $\forall \kappa. \manyk{M'}{m}{N}$ and thus, for any $\kappa$, $M \to^1 M' \text{ and } \laterclock\kappa
  (\manyk{M'}{m}{N})$. As before, this allows us to conclude $\forall \kappa. \manyk{M}{m+1}{N}$.

  The right to left implication is proved by induction on
  $k$. When $k = 0$ the clock $\kappa$ is not free in $\manyk{M}{k}{N}$
  and so $\forall \kappa. \manyk{M}{k}{N}$ is isomorphic to $\manyk{M}{k}{N}$, 
  which implies $\many{M}{k}{N}$. When $k = k' + 1$ the assumption $\forall \kappa. \manyk{M}{k}{N}$
  implies that $\zeromany{M}{N'}$, $\toone{N'}{N''}$
  and $\forall\kappa . \laterclock\kappa (\manyk{N''}{k'}{N})$.  By the 
  type isomorphism (\ref{gdtt:force}) the latter implies
  $\forall \kappa. (\manyk{N''}{k'}{N})$, which by the induction hypothesis implies 
  $\many{N''}{k'}{N}$. Thus we conclude $\many{M}{k}{N}$.
\end{proof}

\subsection{Examples}
As an example of a recursive FPC type, one can encode the natural
numbers as
\begin{align*}
  \pcfnattype & \eqdef \mu \alpha. 1 + \alpha \\
  \pcfzero & \eqdef \fpcfold{(\fpcinl{(\fpcunit)})}\\
  \pcfsucc{M} & \eqdef \fpcfold{(\fpcinr{(M)})} 
\end{align*}
Using this definition we can define the term $\pcfifzbare$ of PCF.  If
$L$ is a term of type $\pcfnattype$ and $M$,$N$ are terms of type
$\sigma$ define $\pcfifzbare$ as
\[
  \pcfifzbare\;L\;M\;N \eqdef \fpccase{(\fpcunfold{L})}{M}{N}
\]
where $x_1,x_2$ are fresh.  It is easy to see that
$\manyk{\pcfifzbare\;\pcfzero\;M\;N}{k+1}{v}$ iff
$\laterbare (\manyk{M}{k}{v})$ and that
$\manyk{\pcfifzbare\;(\pcfsucc{L})\;M\;N}{k+1}{v}$ iff
$\laterbare (\manyk{N}{k}{v})$ for any $L$ term of type
$\pcfnattype$.  For example, $\manyk{\pcfifz{1}{0}{1}}{2}{42}$ is
$\laterbare 0$. On the other hand,  $\many{\pcfifz{1}{0}{1}}{2}{42}$ is 
equivalent to $0$, showing that $\manyk{}{}{}$ and $\many{}{}{}$ are not equivalent. 

Recursive types introduce divergent terms. For example, given a type
$A$,
the Turing fixed point combinator on $A$ can be encoded as follows:
\begin{align*}
  & B \eqdef \mu \alpha. (\alpha \to (A \to A) \to A) \\ 
  & \theta : B \to (A \to A) \to A \\ 
  & \theta \eqdef \lambda x \lambda y . y (\fpcunfold{x}\;x\;y ) \\
  & \pcffix{A}{} \eqdef \theta (\fpcfold{\theta})
\end{align*}
An easy induction shows that
$\manyk{(\pcffix{\sigma}{(\lambda x. x)}}{k}{v}) \propeq
\laterbare^{k} 0$, where $0$ is the empty type.

If $\many{M}{k}{v}$ with $v$ a value and $M$ a term, then
\begin{itemize}
\item $\manyk{M}{k}{v}$ is true
\item $\manyk{M}{n}{v}$ is logically equivalent to
  $\laterbare^{\texttt{min}(n,k)} 0 $ if $n \neq k$, where $0$ is the
  empty type
\end{itemize}
If, on the other hand, $M$ is divergent in the sense that for any $k$
there exists an $N$ such that $\many{M}{k}{N}$, then $\manyk{M}{n}{v}$
is equivalent to $\laterbare^{n}0$.

\section{Denotational Semantics}
\label{sec:fpc:den}
We now define the denotational semantics of FPC.  First we recall the
definition of the guarded recursive version of the \emph{lifting
  monad} on types from~\cite{PMB15}.  This is defined as the
\emph{unique} solution to the guarded recursive type equation
\[L A \cong A + \laterbare LA\] which exists because the recursive
variable is guarded by a $\laterbare$. Recall
(Section~\ref{sec:guarded-recursion-intro}) that guarded recursive
types are defined as fixed points of endomaps on the universe, so $LA$
is only defined for small types $A$.  We will only apply $L$ to small
types in this paper.
 
The isomorphism induces a map $\tick{}{LA} : \laterbare LA \to LA$ and
a map $\now{}{} \co A \to LA$. An element of $LA$ is either of the
form $\eta(a)$ or $\tick{}{} (r)$. We think of these cases as values
``now'' or computations that ``tick''.  Moreover, given $f : A \to B$
with $B$ a $\laterbare$-algebra (i.e., equipped with a map
$\tick{}B\co \laterbare B \to B$), we can lift $f$ to a homomorphism
of $\laterbare$-algebras $\hat f \co LA \to B$ as follows
\begin{equation}\begin{aligned}
    \hat f(\eta(a)) & \eqdef f(a) \\
    \hat f(\tick{}{}(r)) & \eqdef \tick{}{B}(\purebare (\hat f)
    \circledast r)
  \end{aligned}
  \label{def:fpc:unique:ext}
\end{equation}
Formally $\hat f$ is defined as a fixed point of a term of type
$\laterbare (LA \to B) \to LA \to B$. Recall that
$\lambda r. \purebare (\hat f) \circledast r$ is the application of
the functor $\laterbare$ to the map $\hat f$, thus $\hat f$ is an
algebra homomorphism.

Intuitively $LA$ is the type of computations possibly returning an
element of $A$, recording the number of steps used in the
computation. We can define the divergent computation as
$\bot \eqdef \fixcombinator(\theta)$ and a ``delay'' map
$\delay{}{LA}$ of type $LA \to LA$ for any $A$ as
$\delay{}{LA} \eqdef \tick{}{LA} \circ \purebare$. The latter can be
thought of as adding a step to a computation.  The lifting $L$ extends
to a functor. For a map $f : A \to B$ the action on morphisms can be
defined using the unique extension as
$L(f) \eqdef \widehat{\eta \circ f}$.

\subsection{Interpretation of types}
A type judgement $\Theta \vdash \tau$ is interpreted as a map of type
$\U{}^{\vert \Theta \vert} \to \U{}$, where $\vert \Theta \vert$ is
the cardinality of the set of variables in $\Theta$.  This
interpretation map is defined by a combination of induction and
\emph{guarded recursion} for the case of recursive types as in
Figure~\ref{fig:fpc:interp}.

\begin{figure}[tb]
  \begin{align*}
    \den{\Theta \vdash \alpha} (\rho) & \eqdef \rho(\alpha) \\
    \den{\Theta \vdash 1} (\rho) & \eqdef L 1 \\
    \den{\Theta \vdash \tau_1 \times \tau_2}(\rho) & \eqdef \den{\Theta
                                                     \vdash
                                                     \tau_1}(\rho) \times \den{\Theta \vdash \tau_2}(\rho) \\
    \den{\Theta \vdash \tau_1 + \tau_2} (\rho) & \eqdef L(\den{\Theta
                                                 \vdash
                                                 \tau_1}(\rho) + \den{\Theta \vdash \tau_2}(\rho))\\
    \den{\Theta \vdash \tau_1 \to \tau_2} (\rho) & \eqdef \den{\Theta
                                                   \vdash  \tau_1}(\rho) \to \den{\Theta \vdash \tau_2} (\rho)\\
    \den{\Theta \vdash \mu \alpha. \tau} (\rho) & \eqdef
                                                  \later{}{(\den{\Theta, \alpha \vdash \tau}(\rho, \den{\Theta
                                                  \vdash \mu \alpha. \tau} (\rho)))}
  \end{align*}
  \caption{Interpretation of FPC types}
  \label{fig:fpc:interp}
  \label{def:fpc:interp}
\end{figure}

More precisely, the case of recursive types is defined to be the fixed
point of a map from $\later{}{(\U{}^{\vert \Theta \vert} \to \U{})}$
to $\U{}^{\vert \Theta \vert} \to \U{}$ defined as follows:
\begin{equation}
  \lambda X. 
  \lambda \rho. \latercode{}  (\pure{}{(\lambda Y : \U{} . \den{\Theta, \alpha \vdash \tau}(\rho, Y))} \app (X \app \purebare(\rho)))
  \label{def:fpc:interp:rec:type}
\end{equation}
ensuring 
\begin{align*}
   \den{\Theta \vdash \mu \alpha. \tau}(\rho)
  & \judgeq
    \latercode{}  (\pure{}{(\lambda Y : \U{} . \den{\Theta, \alpha \vdash \tau}(\rho, Y))} \app (\pure{}(\den{\Theta \vdash \mu \alpha. \tau}) \app \purebare(\rho))) \\
  & \judgeq
    \latercode{}  (\pure{}{(\lambda Y : \U{} . \den{\Theta, \alpha \vdash \tau}(\rho, Y))} \app (\pure{}(\den{\Theta \vdash \mu \alpha. \tau}(\rho)))) \\
  & \judgeq
    \latercode{}  (\pure{}{(\den{\Theta, \alpha \vdash \tau}(\rho, \den{\Theta \vdash \mu \alpha. \tau}(\rho)))}) \\
  & \judgeq
    \laterbare(\den{\Theta, \alpha \vdash \tau}(\rho, \den{\Theta \vdash \mu \alpha. \tau}(\rho))) 
\end{align*}

The first equation is the application of rule (\ref{eq:fixunfold}) for
the guarded fix-point combinator, whereas the second equation is
derived by distributivity over the later application operator
described by rule (\ref{eq:fpc:next:laterappl}). Finally, the last
equation is derived by the fact that the elements of the code of the
later operator is the later operator on types (rule
(\ref{eq:latercode})).

We prove now the substitution lemma for types which states that
substitution behaves as expected, namely that substituting type
variables in the syntax with syntactic types corresponds to applying a
dependent type $\U{} \to \U{}$ to a type $\U{}$. This can be proved
using guarded recursion in the case of recursive types.
\begin{lemma}[Substitution Lemma for Types]
  \label{lem:fpc:subst:types}
  Let $\sigma$ be a well-formed type with variables in $\Theta$ and
  let $\rho$ be of type $\U{}^{\vert \Theta \vert}$. If
  $\Theta, \beta \vdash \tau$ then
  \[\den{\Theta \vdash \tau[\sigma/\beta]}(\rho) = \den{\Theta, \beta
      \vdash \tau} (\rho, \den{\Theta \vdash \sigma} (\rho))\]
\end{lemma}
\begin{proof}
  The proof is by induction on $\Theta, \beta \vdash \tau$. Most cases
  are straightforward, and we just show the case of
  $\Theta, \beta \vdash \foldedtype$.  The proof of this case is by
  \emph{guarded recursion}, and thus we assume that
  \begin{equation}
    \later{}{(\den{\Theta \vdash (\mu\alpha.\tau)[\sigma/\beta]}(\rho)
      = \den{\Theta, \beta \vdash \mu \alpha.\tau} (\rho, \den{\Theta
        \vdash \sigma} (\rho)))}
    \label{lem:fpc:subst:types:gr}
  \end{equation}
  Assuming (without loss of generality) that $\alpha$ is not $\beta$ 
  we get the following series of equalities
\begin{align*}
 & \den{\Theta \vdash   (\foldedtype)  [\sigma/\beta]}  (\rho) \\
& \qquad = \den{\Theta\vdash \mu \alpha. (\tau [\sigma/\beta])} (\rho) \\
 &\qquad = \later{}{(\den{\Theta, \alpha \vdash \tau[\sigma/\beta]}(\rho,
    \den{\Theta \vdash \mu \alpha. (\tau[\sigma/\beta])}(\rho)))} \\
 &\qquad = \later{}{(\den{\Theta, \alpha, \beta \vdash \tau}(\rho, \den{\Theta \vdash\mu
      \alpha. (\tau[\sigma/\beta])}(\rho), \den{\Theta,\alpha \vdash
      \sigma} (\rho, \den{\mu
      \alpha. (\tau[\sigma/\beta])}(\rho))))} \\
 &\qquad =     \later{}{(\den{\Theta, \alpha, \beta \vdash \tau}(\rho, \den{\Theta \vdash\mu
        \alpha. (\tau[\sigma/\beta])}(\rho), \den{\Theta \vdash
        \sigma} (\rho)))} 
\end{align*}
  The latter equals 
  \[\latercodebare (\purebare (\lambda X \lambda Y . \den{\Theta, \alpha, \beta \vdash
      \tau}(\rho, X, Y)) \app (\purebare (\den{\Theta \vdash\mu
        \alpha. (\tau[\sigma/\beta])} (\rho))) \app \purebare \den{\Theta \vdash
      \sigma} (\rho)) \]
  By (\ref{eq:later:id}), (\ref{lem:fpc:subst:types:gr}) implies
  \[
  \purebare(\den{\Theta \vdash \mu\alpha.(\tau[\sigma/\beta])}(\rho))
  = \purebare(\den{\Theta, \beta \vdash \mu \alpha.\tau} (\rho,
  \den{\Theta \vdash \sigma} (\rho)))
  \]
and so
\begin{align*}
 \den{\Theta \vdash \foldedtype [\sigma/\beta]} (\rho) 
 & = \later{}{(\den{\Theta, \alpha, \beta \vdash \tau}(\rho,
    \den{\Theta, \beta \vdash \mu \alpha. \tau}(\rho, \den{\Theta
      \vdash \sigma}(\rho)), \den{\Theta \vdash \sigma} (\rho)))} \\
 & = \later{}{(\den{\Theta, \beta, \alpha \vdash \tau}(\rho, \den{\Theta \vdash \sigma}
        (\rho),           \den{\Theta, \beta \vdash \mu
        \alpha. \tau}(\rho, \den{\Theta \vdash \sigma}(\rho))))} \\
 & = \den{\Theta, \beta \vdash \mu \alpha. \tau}(\rho,
    \den{\Theta \vdash \sigma}(\rho))
\end{align*}
\end{proof}

By direct use of the Substitution Lemma we can prove that the
interpretation of the recursive type equals the interpretation of the
unfolding of the recursive type itself, only one step
later. Intuitively, this means that we need to consume one
computational step to look at the data.
\begin{lemma} \label{lem:fpc:fix:eq} For all types $\tau$ and
  environments $\rho$ of type $\U{}^{\vert \Theta \vert}$,
  \[\den{\Theta \vdash \mu \alpha. \tau}(\rho) \propeq
    \later{}{\den{\Theta \vdash \tau[\mu \alpha. \tau/
        \alpha]}}(\rho)\]
\end{lemma}

The interpretation of every \emph{closed} type $\tau$ carries a
$\laterbare$-algebra structure, i.e., a map
$\tick{}{\tau} \co \laterbare \den{\tau} \to \den{\tau}$, defined by
guarded recursion and structural induction on $\tau$ as in Figure~\ref{fig:later:alg}. The case of recursive types is welltyped by Lemma~\ref{lem:fpc:fix:eq}, and 
can be formally constructed as a fixed point of a term of type
\[
  G\co \laterbare (\Pi \sigma\co \FPCTypes. (\later{}{\den{\sigma}}
  \to \den{\sigma})) \to \Pi \sigma. (\later{}{\den{\sigma}} \to
  \den{\sigma})
\]
as follows. Suppose
$F\co \laterbare (\Pi \sigma\co \FPCTypes. (\later{}{\den{\sigma}} \to
\den{\sigma}))$, and define $G(F)$ essentially as in
Figure~\ref{fig:later:alg} but with the clause
$G(F)_{\mu \alpha. \tau}$ for recursive types being defined as
\begin{equation}
  \label{sec:tick:welldef}
  \lambda x\co \later{}{\den{\foldedtype}}. (F_{\tau[\mu
    \alpha. \tau/\alpha]} \app x)
\end{equation}
Here $F_{\sigma}$ is defined as $F \app \purebare (\sigma)$ using a
generalisation of $\app$ to dependent products to be defined in
Section~\ref{sec:gdtt}.  Define $\tick{}{}$ as the fixed point of
$G$. Then
\begin{equation}
  \begin{aligned}
    \tick{}{\mu \alpha. \tau}(x) & \judgeq G(\purebare{}(\tick{}{}))_{\mu \alpha. \tau}(x) \\
    & \judgeq \purebare{}(\tick{}{})_{\tau[\mu \alpha. \tau/\alpha]}
    \app (x)
  \end{aligned}
\end{equation}

\begin{figure}[tb]
  \begin{align*}
    \tick{}{\fpcunittype} & \eqdef  \lambda x : \laterbare
                            \den{\fpcunittype}. \tick{}{L\den{1}} (x)\\
    \tick{}{\tau_1 \times \tau_2} & \eqdef  \lambda x : \laterbare
                                    \den{\tau_1 \times
                                    \tau_2}. \langle\tick{}{\tau_1}(\later{}(\pi_1)(x)),
                                    \tick{}{\tau_2} (\later{}(\pi_2)(x))\rangle \\ 
    \tick{}{\tau_1 + \tau_2} & \eqdef  \lambda x : \laterbare
                               \den{\tau_1 + \tau_2}. \tick{}{L\den{\tau_1 + \tau_2}}(x)                              \\ 
    \tick{}{\sigma \to \tau} & \eqdef \lambda f : \later{}{(\den{\sigma} \to
                               \den{\tau})}. \lambda x: \den{\sigma}. \tick{}{\tau}(f \app(\purebare (x))) \\
    \tick{}{\foldedtype} & \eqdef \lambda x\co \later{}{\den{\foldedtype}}. \purebare (\tick{}{\tau[\mu \alpha. \tau/\alpha]}) \app (x)
  \end{align*}
  \caption{Definition of
    $\tick{}{\sigma} : \laterbare\den{\sigma} \to \den{\sigma}$}
  \label{fig:later:alg}
\end{figure}

Using the $\tick{}{}$ we define the delay operation which,
intuitively, takes a computation and adds one step.
\[
  \delay{}{\sigma} \eqdef \tick{}{\sigma} \circ \purebare.
\]

\subsection{Interpretation of terms}

Figure~\ref{fig:fpc:term:interp} defines the interpretation of
judgements $\Gamma \vdash M : \sigma$ as functions from $\den{\Gamma}$
to $\den{\sigma}$ where
$\den{x_1 : \sigma_1 , \cdots, x_n : \sigma_n } \eqdef \den{\sigma_1}
\times \cdots \times \den{\sigma_n}$.  In the case of $\fpccasebare$,
the function $\widehat{f}$ is the extension of $f$ to a homomorphism
defined as in (\ref{def:fpc:unique:ext}) above, using the fact that
all types carry a $\laterbare$-algebra structure.  The interpretation
of $\fpcfold{}$ is welltyped because $\purebare (\den{M}(\gamma))$ has
type $\laterbare\den{\unfoldedtype}$ which by
Lemma~\ref{lem:fpc:fix:eq} is equal to $\den{\foldedtype}$.  In the
case of $\fpcunfoldbare$, since $\den{M}(\gamma)$ has type
$\den{\foldedtype}$, which by Lemma~\ref{lem:fpc:fix:eq} is equal to
$\laterbare\den{\unfoldedtype}$, the type of
$\tick{}{\unfoldedtype}(\den{M}(\gamma))$ is $\den{\unfoldedtype}$.

\begin{figure}[ht]
  \begin{align*}
    \den{\Gamma \vdash t : \sigma} & : \den{\Gamma} \to \den{\sigma}\\
    \den{\Gamma \vdash x} (\gamma) & \eqdef \gamma(x) \\
    \den{\Gamma \vdash \fpcunit}(\gamma) & \eqdef \now{}(\sunit)\\
    \den{\Gamma \vdash \langle M, N \rangle}(\gamma) & \eqdef
                                                       \gttpair{\den{M}(\gamma)}{\den{N}(\gamma)}\\
    \den{\Gamma \vdash \fpcfst{M}}(\gamma) & \eqdef  \gttfst (\den{M}(\gamma))\\
    \den{\Gamma \vdash \fpcsnd{M}}(\gamma) & \eqdef  \gttsnd (\den{M}(\gamma))\\
    \den{\Gamma \vdash \lambda x.M}(\gamma) & \eqdef  \lambda
                                              x. \den{M}(\gamma, x)\\
    \den{\Gamma \vdash M N}(\gamma) & \eqdef \den{M}(\gamma)
                                      \den{N}(\gamma)\\
    \den{\Gamma \vdash \fpcinl{E}}(\gamma) & \eqdef  \now{}{}(\sinl{\den{E}(\gamma)}) \\
    \den{\Gamma \vdash \fpcinr{E}}(\gamma) & \eqdef  \now{}{} (\sinr{\den{E}(\gamma)}) \\
    \den{\Gamma \vdash \fpccase{L}{M}{N} }(\gamma) & \eqdef \widehat f (\den{L}(\gamma))\\
                                   &  \hspace{-6ex} \begin{aligned} \text{ where }  & f (\sinl (x_1)) \eqdef \den{M}(\gamma, x_1) \\ 
                                     & f (\sinr (x_2)) \eqdef \den{N}(\gamma, x_2) \\
                                   \end{aligned} \\
    \den{\Gamma \vdash \fpcfold{M} }(\gamma) & \eqdef \purebare (\den{M}(\gamma))\\
    \den{\Gamma \vdash \fpcunfold{M} }(\gamma) & \eqdef
                                                 \tick{}{\unfoldedtype}
                                                 (\den{M}(\gamma))
  \end{align*}
  \caption{Interpretation of FPC terms}
  \label{fig:fpc:term:interp}
\end{figure}

\begin{lemma} \label{lem:fpc:unfold:fold:delay} If
  $\Gamma \vdash M : \tau[\mu \alpha.\tau/\alpha]$ then
  $\den{\fpcunfold{(\fpcfold{M})}} (\gamma) = \delay{}{\tau[\mu
    \alpha.\tau/\alpha]} \den{M}(\gamma)$.
\end{lemma}
\begin{proof}
  Straightforward by definition of the interpretation and by the type
  equality from Lemma~\ref{lem:fpc:fix:eq}.
\end{proof}

Next lemma proves substitution is well-behaved for terms. The proof is
standard textbook result from domain theory (e.g.~\cite{Winskel93, Streicher06}).
\begin{lemma}[Substitution Lemma]
  Let $\Gamma \equiv x_1: \sigma_1, \cdots, x_k: \sigma_k$ be a
  context such that $\Gamma \vdash M : \tau$, and let
  $\Delta \vdash N_i : \sigma_i$ be a term for each $i = 1, \ldots
  k$. If further $\gamma \in \den{\Delta}$, then
  \[\den{\Delta \vdash M [ \vec{N} / x ] : \tau} (\gamma) =
    \den{\Gamma \vdash M : \tau} \left(\den{\Delta \vdash \vec{N}:
        \vec{\sigma}} (\gamma)\right)\]
  \label{lem:fpc:subst:terms}
\end{lemma}
\begin{proof}

  By induction on the typing judgement $\Gamma \vdash M : \tau$.

  The cases for $\Gamma \vdash \fpcunit : \fpcunittype$,
  $\Gamma \vdash x : \tau$
  , $\Gamma \vdash M \ N : \tau$, $\jud{\Gamma}{\fpcfst{M}}{\tau_1}$,
  $\jud{\Gamma}{\fpcsnd{M}}{\tau_2}$,
  $\jud{\Gamma}{\langle M,N\rangle}{\tau_1 \times \tau_2}$ are
  standard.

  For the case $\jud{\Gamma}{\fpcinl{M}}{\tau_1 + \tau_2}$ we start
  from
  \[\den{\Delta \vdash (\fpcinl{M})[\vec{N}/\vec{x}]: \tau_1 +
    \tau_2}(\gamma)\]
  By substitution $(\fpcinl{M})[\vec{N}/\vec{x}]$ equals
  $\fpcinl{(M[\vec{N}/\vec{x}])}$. We also know that its denotation
  equals $\eta (\sinl \den{ (M[\vec{N}/\vec{x}])}(\gamma))$ by
  induction hypothesis this is equal to
  \[\eta (\sinl \den{\Gamma \vdash (M): \tau_1 + \tau_2}(\gamma,
  \den{\Delta \vdash \vec{N} : \vec{\sigma}}(\gamma)))\]
  which is now by definition what we wanted.  The case for
  $\jud{\Gamma}{\fpcinr{N}}{\tau_1 + \tau_2}$ is similar.

  Now the case for $\jud{\Gamma}{\fpccase{L}{M}{N}}{\sigma}$.  By
  definition we know that
  $\den{\Delta \vdash (\fpccase{L}{M}{N})[\vec{N}/\vec{x}]:
    \tau}(\gamma)$ is equal
  \[\den{\Delta \vdash \fpccase{L [\vec{N}/\vec{x}]}{M
      [\vec{N}/\vec{x}]}{N [\vec{N}/\vec{x}]}: \tau}(\gamma)\]
  which is by definition of the interpretation equal to
  \[\widehat{f} (\lambda
  x_1.\den{M[\vec{N}/\vec{x}]}(\gamma, x_1),\lambda x_2. \den{N [
    \vec{N} / \vec{x} ]}(\gamma, x_2)) (\den{L [ \vec{N} / \vec{x}
    ]}(\gamma))\]
  where $\widehat f$ is as in Figure~\ref{fig:fpc:term:interp}.
  By induction hypothesis we know that this is equal to
  \begin{align*}
    \widehat{f} (\lambda x_1.\den{M} (\gamma, x_1,& 
                                                      \den{\Delta \vdash \vec{N} : \vec{\sigma}}(\gamma)),
                                                      (\lambda x_2.\den{N} (\gamma, x_2, \den{\Delta \vdash \vec{N}
                                                      : \vec{\sigma}}(\gamma))))\\
                                                    &  (\den{L} (\gamma, \den{\Delta \vdash \vec{N} :
                                                      \vec{\sigma}}(\gamma)))\\
  \end{align*}
  which is equal by definition to
  \[\den{\Gamma \vdash \fpccase{L}{M}{N} : \tau} (\gamma, \den{\Delta
    \vdash \vec{N} : \vec{\sigma}}(\gamma))\]

  Now the fixed point cases. For the case
  $\jud{\Gamma}{\fpcunfold{M}}{\tau[\mu \alpha. \tau/\alpha]}$ we know
  that $\den{\Gamma \vdash (\fpcunfold{M})[\vec{N}/\vec{x}]}(\gamma)$
  is equal by definition of the substitution function to
  \[\den{\Gamma \vdash \fpcunfold{(M [\vec{N}/\vec{x}])}}(\gamma)\]
  which by definition of interpretation is
  $\tick{}{\unfoldedtype} (\den{\Gamma \vdash (M [\vec{N}/\vec{x}])}(\gamma))$. By
  induction hypothesis this is equal to 
  \[\tick{}{\unfoldedtype} (\den{\Gamma \vdash M}(\den{\Delta \vdash \vec{N}}(\gamma))\]
  which by definition is
  $\den{\Gamma \vdash \fpcunfold{(M)}} (\den{\Delta \vdash
    \vec{N}}(\gamma))$.
  For the case $\jud{\Gamma}{\fpcfold{M}}{\mu \alpha. \tau}$ we know
  that $\den{\Gamma \vdash (\fpcfold{M})[\vec{N}/\vec{x}]}(\gamma)$ is
  equal by defintion to
  $\den{\Gamma \vdash \fpcfold{(M [\vec{N}/\vec{x}])}}(\gamma)$ which
  is by definition of the interpretation equal to
  $\purebare (\den{\Gamma \vdash (M [\vec{N}/\vec{x}])}(\gamma))$.  By
  induction hypothesis we get
  $\sfold (\den{\Gamma \vdash M}(\den{\Theta \vdash \vec{N}}(\gamma))$
  which is by definition
  \[\den{\Gamma \vdash \fpcfold{(M)}} \den{\Theta \vdash
    \vec{N}}(\gamma)\]
\end{proof}

We now aim to show a soundness theorem for the interpretation of
FPC. We do this by first showing soundness of the single step
reduction as in the next lemma. As usual in denotational semantics,
this proves that the model is agnostic to operational reductions.

\begin{lemma} \label{lem:fpc:to:den} Let $M$ be a closed term of type
  $\tau$. If $M \to^k N$ then $\den{M}(*) = \delta^k \den{N} (*)$
\end{lemma}
\begin{proof}
The proof goes by induction on $M \to^k N$. The cases when $k = 0$
  follow straightforwardly from the structure of the denotational
  model.

  The case $\fpcunfold{(\fpcfold{M})} \to^1 M$ follows directly from
  Lemma~\ref{lem:fpc:unfold:fold:delay}.
  
  The case for $(\lambda x: \sigma. M) (N) \to^0 M[N/x]$ is
  straightforward from by Substitution
  Lemma~\ref{lem:fpc:subst:terms}.
  
  The case for $\fpccase{(\fpcinl{L})}{x.M}{x.N} \to^0 M [L/x]$ and
  the case for \[\fpccase{(\fpcinr{L})}{x.M}{x.N} \to^0 N [L/x]\] follow
  directly by definition.

  Also the elimination for the product, namely
  $\fpcfst{\langle M, N \rangle} \to^0 M$ and
  $\fpcsnd{\langle M, N \rangle} \to^0 N$ follow directly from the
  definition of the interpretation.

  Now we prove the inductive cases.  For the case $M_1 N \to^k M_2 N$
  we know that by definition $\den{M_1 N}(*) = \den{M_1}(*) \den{N}(*)$.
  By induction hypothesis we know that 
  $\den{M_1}(*) = \delay{}{\sigma \to \tau}[k](\den{M_2}(*))$, thus
  $\den{M_1}(*) \den{N}(*) = (\delay{}{\sigma \to
    \tau}[k](\den{M_2}(*))) \den{N}(*)$
  By definition of $\delay{}{}$ and $\tick{}{}$ this is equal to 
  $\delay{}{\tau}[k](\den{M_2}(*) \den{N}(*))$.
  
  Now the case for \[\fpccase{L}{M}{N} \to^k \fpccase{L'}{M}{N}\]  The
  induction hypothesis gives
  $\den{L} = \delay{}{\tau_1 + \tau_2} \circ \den{L'}$, and so
  Lemma~\ref{lem:fpc:case:hom} applies proving the case.

  The case for $\fpcfst{M} \to^k \fpcfst{M'}$ and for
  $\fpcsnd{M} \to^k \fpcsnd{M'}$ are similar to the previous case.
  
  Finally, the case for $\fpcunfold{M_1} \to^k \fpcunfold{M_2} $.  By
  definition we know that  \[\den{\fpcunfold{M_1}}(*) = \theta (\den{M_1}(*))\]  By
  induction hypothesis this is equal to
  $\theta(\delay{}{\foldedtype}[k](\den{M_2}(*)))$ which by
  Lemma~\ref{lem:fpc:unfold:hom} is equal to
  $\delay{}{\unfoldedtype}[k] (\theta (\den{M_2}(*)))$ thus
  concluding.
\end{proof}

The two most complicated cases of the proof of
Lemma~\ref{lem:fpc:to:den}, namely the $\fpcunfoldbare$-$\fpcfoldbare$
reductions and $\fpccasebare$, are captured in the following two
lemmas. In particular, the first of these states that the
interpretation of $\fpccasebare$ is a $\laterbare$-algebra
homomorphism. In other words, case analysing over a computation that
perform $n$ ticks and then produces a result $v$ is equal to a
computation that produces $n$ ticks and then performs case analysis
over a terminating computation producing a value $v$.

\begin{lemma}
  \label{lem:fpc:case:hom}
  \begin{enumerate}
  \item The interpretation of $\fpccasebare$ is a homomorphism of
    $\later{}{}{}$-algebras in the first variable, i.e.,
    \begin{align*}
      & \den{\ \lambda x \co \tau_1 + \tau_2 \ld \fpccase{x}{M}{N}}
        (\gamma) (\tick{}{}(r)) \\
      = & \tick{}{}(\purebare(\den{\ \lambda x \co \tau_1 + \tau_2 \ld \fpccase{x}{M}{N}}(\gamma)) \app r)
    \end{align*}
  \item If $\den{L}(\gamma) = \delta (\den{L'}(\gamma))$, then
    \[
      \den{\fpccase{L}{M}{N}}(\gamma) = \delta
      \den{\fpccase{L'}{M}{N}}(\gamma)
    \]
  \end{enumerate}
\end{lemma}
\begin{proof}
  For the proof of the first part, we use the notation $\widehat f$ as
  in Figure~\ref{fig:fpc:term:interp}. Since $\widehat f$ is a
  homomorphism of $\laterbare$-algebras we get
  \begin{align*}
    \den{\lambda x. \fpccase{x}{M}{N}} (\gamma) (\tick{}{\tau_1 +
    \tau_2} (r))
    & = \widehat{f}(\tick{}{\tau_1 + \tau_2 }(r)) \\
    & = \tick{}{\sigma} (\pure{}{(\widehat{f}}) \app r) \\
    & = \tick{}{\sigma} (\pure{}{\den{\lambda x. \fpccase{x}{M}{N}}
      (\gamma)} \app r)
  \end{align*}

  For the second part, note that $\widehat f$ is
  $\den{\lambda x. \fpccase{x}{M}{N}}(\gamma)$, so
  \begin{align*}
    \den{\fpccase{L}{M}{N}}(\gamma) 
    & = \widehat f (\den{L}(\gamma)) \\
    & = \widehat f
      (\delay{}{\tau_1 + \tau_2} (\den{L'}(\gamma))) \\
    & = \widehat f(\tick{}{\tau_1 +
      \tau_2}(\pure{}{(\den{L'}(\gamma))})) \\
    & = \tick{}{\sigma} (\pure{}{(\widehat f)} \app
      (\pure{}{(\den{L'}(\gamma))})) \\
    & = \tick{}{\sigma} (\pure{}{(\widehat f (\den{L'}(\gamma)))}) \\
    & = \delay{}{\sigma} (\den{\fpccase{L'}{M}{N}}(\gamma))
  \end{align*}
\end{proof}

We now prove the same for the interpretation of $\fpcunfoldbare$. The
key point here is to observe that the tick operation for a folded
recursive type , namely $\tick{}{\foldedtype}$, is precisely the tick
of the unfolded recursive type after one step of computation, namely
$\laterbare (\tick{}{\unfoldedtype})$.
\begin{lemma}\label{lem:fpc:unfold:hom}
  If $\foldedtype$ is a closed FPC type then
  \begin{enumerate}
  \item
    $\den{\ \lambda x \co \foldedtype \ld \fpcunfold{x}}
    (\tick{}{\foldedtype}(r)) = \tick{}{\unfoldedtype} (\purebare
    (\tick{}{\unfoldedtype}) \app r)$
  \item If
    $\den{M}(\gamma) = \delay{}{\foldedtype}(\den{M'}(\gamma))$, then
    \[\den{\fpcunfold{M}}(\gamma) =
      \delay{}{\unfoldedtype}(\den{\fpcunfold{M'}}(\gamma))\]
  \end{enumerate}
\end{lemma}
\begin{proof}
  The interpretation for
  $\den{\ \lambda x \co \foldedtype \ld
    \fpcunfold{x}}(\tick{}{\foldedtype}(r))$ yields
  $\tick{}{\unfoldedtype} (\tick{}{\foldedtype}(r))$.  This type
  checks as $r$ has type $\laterbare \den{\foldedtype}$, thus
  $(\tick{}{\foldedtype}(r))$ has type $\den{\foldedtype}$ which -- by
  Lemma~\ref{lem:fpc:fix:eq} -- is equal to
  $\laterbare \den{\unfoldedtype}$. Thus the term
  $\tick{}{\unfoldedtype} (\tick{}{\foldedtype}(r))$ has type
  $\den{\unfoldedtype}$.  Now by definition of $\tick{}{\foldedtype}$
  this is equal to
  $\tick{}{\unfoldedtype} (\purebare (\tick{}{\unfoldedtype}) \app
  (r))$ which is what we wanted.

    For the second statement, we compute  
\begin{align*}
\den{\fpcunfold{M}}(\gamma) 
& = \tick{}{\unfoldedtype}(\den{M}(\gamma)) \\
    & = \tick{}{\unfoldedtype}( \delay{}{\foldedtype}(\den{M'}(\gamma))) \\
    & = \tick{}{\unfoldedtype}(\tick{}{\foldedtype}(\purebare
  (\den{M'}(\gamma)))) \\
   & = \tick{}{\unfoldedtype} (\purebare (\tick{}{\unfoldedtype}) \app
  (\purebare (\den{M'}(\gamma)))) && \text{(statement 1)} \\
  & = \tick{}{\unfoldedtype} (\purebare (\tick{}{\unfoldedtype}
  (\den{M'}(\gamma)))) && \text{(rule (\ref{eq:fpc:next:laterappl}))} \\
  & = \tick{}{\unfoldedtype} (\purebare 
  (\den{\fpcunfold{M'}}(\gamma)))  \\
   & = \delay{}{\unfoldedtype}(\den{\fpcunfold{M'}}(\gamma))
\end{align*}

%
\end{proof}

We can now prove Lemma~\ref{lem:fpc:to:den}. As stated above, this is
soundness of the model w.r.t. the small-step operational
semantics. For the proof, it is crucial that the interpretation of
every term is an homomorphism of tick$\tickbare$-algebras.  This falls
out in many cases. For the cases of the interpretation of
$\fpcunfoldbare$ and the inductive case of $\fpccasebare$ we use the
lemmas we just proved above.
\begin{proofof}{Lemma~\ref{lem:fpc:to:den}}
  The proof is by induction on $M \to^k N$. Most of the cases are
  straightforward, some ($\beta$-reductions for function and sum
  types) using the substitution lemma
  (Lemma~\ref{lem:fpc:subst:terms}).  The case
  $\fpcunfold{(\fpcfold{M})} \to^1 M$ follows directly from
  Lemma~\ref{lem:fpc:unfold:fold:delay}.
  
%
%
%
  Now we prove the inductive cases.  For the case $M_1 N \to^k M_2 N$
  we know that by definition
  $\den{M_1 N}(*) = \den{M_1}(*) \den{N}(*)$.  By induction hypothesis
  we know that
  $\den{M_1}(*) = \delay{}{\sigma \to \tau}[k](\den{M_2}(*))$, thus
  $\den{M_1}(*) \den{N}(*) = (\delay{}{\sigma \to
    \tau}[k](\den{M_2}(*))) \den{N}(*)$ By definition of $\delay{}{}$
  and $\tick{}{}$ this is equal to
  $\delay{}{\tau}[k](\den{M_2}(*) \den{N}(*))$.
  
  In the case of \[\fpccase{L}{M}{N} \to^k \fpccase{L'}{M}{N}\] the
  induction hypothesis gives
  $\den{L}(*) = \delay{}{\tau_1 + \tau_2} \den{L'}(*)$, and so
  Lemma~\ref{lem:fpc:case:hom} applies proving the case.

  
  Finally, the case for $\fpcunfold{M} \to^k \fpcunfold{M'}$. If $k=0$
  the case follows trivially from the induction hypothesis. If $k=1$,
  the step from the induction hypothesis to the case is exactly the
  second statement of Lemma~\ref{lem:fpc:unfold:hom}.
\end{proofof}

We now state and prove soundness of our model w.r.t. the operational
semantics. We use the transitive closure over $\to^k$, namely
$\Rightarrow^k$, which is \emph{synchronised} with the $\laterbare$
operator in the type theory. Using $\Rightarrow$ (rather than
$\to_*$) is not essential to prove soundness,
but it is crucial to prove computational adequacy, which will be
presented in the next section. On the other hand, the explicit
step-indexing $k$ in $\Rightarrow^k$ is necessary to relate the number
of operational steps with the number of delays (or ticks) in the
denotational semantics.
\begin{proposition}[Soundness]
  \label{prop:fpc:soundness}
  Let $M$ be a closed term of type $\tau$. If $M \Rightarrow^k N $
  then $\den{M}(*) = \delta^k \den{N} (*)$.
\end{proposition}
\begin{proof}
  By induction on $k$. When $k = 0$ Lemma~\ref{lem:fpc:to:den} applies
  concluding the case. When $k = n + 1$ by definition we have
  $\zeromany{M}{M'}$, $M' \to^1 M''$ and
  $\later{}{(M'' \Rightarrow^{n} N)}$ .  By repeated application of
  Lemma~\ref{lem:fpc:to:den} we get $\den{M}(*) = \den{M'}(*)$ and
  $\den{M'}(*) = \delay{}{}(\den{M''}(*))$.  By induction hypothesis
  we get $\later{}{(\den{M''}(*) = \delta^{n} \den{N}(*))}$ which
  implies
  $\purebare(\den{M''}(*)) = \purebare(\delay{}{}[n] \den{N}(*)))$ and
  since $\delay{}{} = \tick{}{} \circ \purebare$, this implies
  $\delay{}{}(\den{M''}(*)) = \delay{}{}[k](\den{N}(*))$.  By putting
  together the equations we get finally
  $\den{M}(*) = \delta^k \den{N}(*)$.
\end{proof}

\section{Computational Adequacy}
\label{sec:fpc:adequacy}

Computational adequacy is the opposite implication of Proposition~\ref{prop:fpc:soundness}
in the case of terms of unit type. It is
proved by constructing a (proof relevant) logical relation
between syntax and semantics. The relation cannot be constructed just by 
induction on the structure of types, since in the case of recursive types,
the unfolding can be bigger than the recursive type. Instead, the relation is 
constructed by guarded recursion: we assume the relation exists \emph{later},
and from that assumption construct the relation \emph{now} by structural induction
on types. Thus the well-definedness of the logical relation is ensured by the
type system of $\gdtt$, more specifically by the rules for guarded recursion. 
This is in contrast to the classical proof in domain theory~\cite{Pit96}, where 
existence requires a separate argument. 

The logical relation uses a lifting of relations on values 
available now, to relations on values available later. To define this lifting, we need 
\emph{delayed substitutions}, an advanced feature of $\gdtt$. 

\subsection{Delayed substitutions}
\label{sec:gdtt}

In $\gdtt$, if $\wftype{}{\Gamma,x\co A}{B}$ is a well formed type and
$t$ has type $\later{} A$ in context $\Gamma$, one can form the type
$\later{}[\hrt{x\gets t}]B$. Intuitively, one time step from now, $t$ delivers
an element in $A$, and $\later{}[\hrt{x\gets t}]B$ is the type of elements 
that one time step from now delivers something in $B$ with $x$ substituted by 
the element delivered at that time by $t$.  
One motivation for this construction is to generalise
$\app$ (described in Section~\ref{sec:guarded-recursion-intro}) to a
dependent version: if $f\co \laterbare(\depprod xAB)$, then
$f\app t \co \later{}[\hrt{x\gets t}]B$.  The idea is that $t$ will
eventually reduce to a term of the form $\pure{}u$, and then
$\later{}[\hrt{x\gets t}]B$ will be equal to $\later{}B\subst ux$. But
if $t$ is open, we may not be able to do this reduction yet.

More generally, we define the notion of \emph{delayed substitution} as
follows. Suppose $\wfctx{}{\Gamma,\Gamma'}$ is a
wellformed context, and suppose $\Gamma'$ is on the form 
$\Gamma' = x_1\co A_1\dots x_n\co A_n$ with all $A_i$ independent, i.e., no $x_j$
appears in an $A_i$. A delayed substitution
$\xi \co {\Gamma} \rightarrowtriangle {\Gamma'}$ is
a vector of terms $\xi = \hrt{x_1\gets t_1, \dots, x_n\gets t_n}$ such
that $\hastype{}{\Gamma}{t_i}{A_i}$ for each $i$. \cite{BGCMB16} gives a more general
definition of delayed substitution allowing dependencies between the
$A_i$'s, but for this paper we just need the definition above.

If $\xi\co \Gamma \rightarrowtriangle \Gamma'$ is a delayed
substitution and $\wftype{}{\Gamma,\Gamma'}{B}$ is a wellformed type,
then the type $\later{}[\xi]B$ is wellformed in context $\Gamma$.  The
introduction form states $\pure{}[\xi]u\co \later{}[\xi]B$ if
$\hastype{}{\Gamma,\Gamma'}{u}{B}$.

In Figure~\ref{fig:fpc:gdtt:rules} we recall some rules
from~\cite{BGCMB16} needed below.
\begin{figure}[tb]
  \begin{align}
    \pure{}\xi\esubst{x}{\pure{}\xi\ld t}\ld u & \judgeq \pure{}\xi\ld (u[t/x]) \label{eq:fpc:delayed:next:beta} \\
    \pure{}\xi\esubst{x}{t}\ld x & \judgeq t  \label{eq:fpc:delayed:next:eta} \\
    \pure{}\xi\esubst{x}{t}\ld u & \judgeq \pure{}\xi\ld u  &  \label{eq:fpc:delayed:next:weakening} \\
    \pure{}[\xi\hrt{x \gets t, y \gets u}\xi']{v} & \judgeq \pure{}[\xi\hrt{y
                                                    \gets u,x \gets
                                                    t}\xi']{u} \label{eq:subst:comm}\\
    \pure{}[\xi]{\pure{}[\xi']{u}} & \judgeq
                                     \pure{}[\xi']{\pure{}[\xi]{u}} \label{eq:next:comm}\\
    (\pure{}[\xi]{t} \propeq_{\later{}[\xi]{A}} \pure{}[\xi]{s}) &\judgeqty \later{}[\xi]{(t \propeq_A s)}\label{eq:tyeq:later}\\
    \El(\latercode{} (\purebare\xi\ld A)) & \judgeqty \laterbare\xi\ld \El(A) \label{eq:El:later:subst}
  \end{align}
\caption{The notation $\xi\esubst{x}{t}$ means the extension of the
delayed substitution $\xi$ with $\esubst{x}{t}$. 
Rule
(\ref{eq:fpc:delayed:next:weakening}) requires $x$ not free in $u$. 
Rule (\ref{eq:next:comm}) requires that none of the variables in the
codomains of $\xi$ and $\xi'$ appear in the type of $u$, and that the
codomains of $\xi$ and $\xi'$ are independent.}
\label{fig:fpc:gdtt:rules}
\end{figure}
Of these, (\ref{eq:fpc:delayed:next:beta}) and (\ref{eq:fpc:delayed:next:eta})
can be considered $\beta$ and $\eta$ laws, and
(\ref{eq:fpc:delayed:next:weakening}) is a weakening principle. Rules
(\ref{eq:fpc:delayed:next:beta}), (\ref{eq:fpc:delayed:next:weakening}) and
(\ref{eq:subst:comm}) also have obvious versions for types, e.g.,
\begin{equation} \label{eq:fpc:delayed:later:beta}
  \later{}\xi\esubst{x}{\pure{}\xi\ld t}\ld B \judgeq \later{}\xi\ld
  (B[t/x])
\end{equation}

%

Rather than be taken as primitive, later application $\app$ can be
defined using delayed substitutions as
\begin{equation}\label{eq:app:def}
  g \app y \eqdef \pure{}[\hrt{f \gets g, x \gets y}]{f(x)}
\end{equation}
Note that if $g : \laterbare(A \to B)$ and $y : \laterbare A$, the type of $g\app y$
is $\laterbare{}[f \gets g, x \gets y]. B$ which reduces to $\laterbare B$ since $f$ and $x$ 
do not appear in $B$. 
With this definition, the rule
$\pure{}{(f (t))} \judgeq \pure{}{f} \circledast \pure{}{t}$ from
Section~\ref{sec:guarded-recursion-intro} generalises to
\begin{align}
  \purebare\xi\ld (f\,t) & \judgeq (\purebare\xi\ld f) \circledast
                           (\purebare\xi\ld t)
                           \label{eq:later:appl}
\end{align}
which follows from (\ref{eq:fpc:delayed:next:beta}). In fact, later application
generalises to the setting of delayed substitutions: if $g : \later{}[\xi] \Pi x : A . B$ and 
$y : \later{}[\xi] A$ define 
\begin{equation}\label{eq:app:gen:def}
  g \app y \eqdef \pure{}[\xi\hrt{f \gets g, x \gets y}]{f(x)} : \later{}[\xi\hrt{x \gets y}]B
\end{equation}
Note that in the special case where $y = \pure{}[\xi] u$ we get
\[
 g\app \pure{}[\xi] u \co \later{}[\xi]B[u/x]
\]

Rules (\ref{eq:fpc:delayed:next:eta}), (\ref{eq:fpc:delayed:next:weakening})
and (\ref{eq:next:comm}) imply
\begin{align*}
\pure{\kappa}[\xi \hrt{x \gets t}]{\pure{\kappa}{x}} 
& \equiv \pure{\kappa}{(\pure{\kappa}[\xi \hrt{x \gets t}]{x})} \\
& \equiv
  \pure{\kappa}{(t)} \\
& \equiv
  \pure{\kappa}[\xi \hrt{x \gets t}]{t} 
\end{align*}
which by (\ref{eq:tyeq:later}) gives an inhabitant of
\begin{equation}
  \later{}[\xi \hrt{x \gets t}]{(\pure{\kappa}{x} \propeq t)}
\label{eq:eta}
\end{equation}

\subsection{A logical relation between syntax and semantics}
\label{sec:logical:relation}

\rasmus{Rewrote a bit in this section. There were some things I did not understand.}

Our strategy to prove computational adequacy is by logical relation
argument. We construct a logical relation $\R{}$ as in
Figure~\ref{fig:fpc:adequacy:relation} between syntax and
semantics. This is done using first guarded recursion and then
induction on the FPC types.

%

Figure~\ref{fig:fpc:adequacy:relation} uses an operation lifting relations $\R{}$ from $A$ to
$B$ to relations $\laterR{}$ from $\laterbare A$ to $\laterbare B$ defined as
\begin{equation}
  t \laterR{} u \eqdef \later{}[\hrt{x \gets t,y \gets u}]{(x \R{} y)}
  \label{def:fpc:rel}
\end{equation}
As a consequence of (\ref{eq:fpc:delayed:later:beta}) the following
statement holds:
\begin{equation}\label{eq:fpc:laterR:next}
  (\pure{}[\xi]{t}) \laterR{}{} (\pure{}[\xi]{u}) \judgeqty \later{}[\xi]{(t \R{}{} u)}
\end{equation}

The lifting on relations is used, e.g., in the second case of $\R{\fpcunittype}$ where 
$x$ is assumed to have type $\laterbare L1$. In that case $\tick{}{\fpcunittype}(x)$
is a semantic computation that takes a step, and so should only be related to $M$,
if $M$ can also reduce in one step to an $M''$, that should be \emph{later} related to $x$.
Note that $x$ is not necessarily of the form $\purebare(y)$ for some $y$, but we can
still related $x$ to $\purebare(M'')$ using delayed substitutions as in the definition of 
$\laterR{\fpcunittype}$. 

Most of the definition of the logical relation is standard, e.g., in the case of function
types, where related functions are required to map related input to related output. 
The case of recursive type deserves some attention. On the right hand
side, we have $x$ of type $\den{\foldedtype}$, which means it is a
piece of data which later will be unfolded and therefore
available. More precisely, it has also type
$\laterbare\den{\unfoldedtype}$.  This semantic program is related to
a syntactic program $M$ if and only if the unfolding of $M$ reduces in
one computational step to an $M''$ which is later related to $x$.

The logical relation is an example of a guarded recursive definition.
To see this, note first that the lifting operation can be expressed on codes mapping
$\R{} \co A \to B \to \U{}$ to
\[
  \lambda x\co \later{}A, y\co \later{}B \ld
  \latercode{}(\pure{}[\hrt{x' \gets x,y' \gets y}](x'\R{}y'))
\]
and this operation factors as $F \circ \purebare$, 
for $F\co \laterbare (A \to B \to \U{})\to A \to B \to \U{}$ defined
as
\[
  \lambda S . \lambda x\co \later{}A, y\co \later{}B \ld
  \latercode{}(\pure{}[\hrt{x' \gets x,y' \gets y, \R{} \gets
    S}](x'\R{}y'))
\]
Using this, one can formally define the logical relation as a fixed
point of a function of type
\begin{align*}
  \laterbare & (\depprod{\tau}{\FPCTypes}{\den\tau \times \FPCTerms \to \U{}})
               \to 
               (\depprod{\tau}{\FPCTypes}{\den\tau \times \FPCTerms \to \U{}})
\end{align*}
similarly to the formal definition of $\tick{}{}$ in the equation
(\ref{sec:tick:welldef}).

\begin{figure}[tb]
  \begin{align*}
    \now{}{}(*) \R{\fpcunittype} M & \eqdef  \manyk{M}{0}{\fpcunit} \\
    \tick{}{\fpcunittype}(x) \R{\fpcunittype} M & \eqdef  
                                                        \Sigma
                                                        M', M'' \co \FPCTerms \ld M \to_*^0 M'\to^1 M''
                                                        \text{ and } x
                                                        \laterR{\fpcunittype}
                                                        \pure{}{(M'')} \\
    x \R{\tau_1 \times \tau_2} M & \eqdef  \pi_1(x)
                                   \R{\tau_1} \fpcfst{(M)} 
                                   \text{ and } \pi_2(x)
                                   \R{\tau_2} \fpcsnd{(M)}
    \\
    \now{}{}(\inl (x)) \R{\tau_1 + \tau_2} M & \eqdef
                                               \Sigma L. \manyk{M}{0}{\fpcinl{L}} \text{
                                               and } x
                                               \R{\tau_1} L \\
    \now{}{}(\inr (x)) \R{\tau_1 + \tau_2} M & \eqdef \Sigma L.
                                               \manyk{M}{0}{\fpcinr{L}} \text{
                                               and } x
                                               \R{\tau_2} L \\
    \tick{}{\tau_1 + \tau_2}(x) \R{\tau_1 + \tau_2} M & \eqdef
     \Sigma M',M''\co \FPCTerms \ld M \to_*^0 M' \to^1 M''   \text{ and } x \laterR{\tau_1 + \tau_2} \pure{}{(M'')}
    \\
    f \R{\tau \to \sigma} M & \eqdef \Pi x \co
                              \den{\tau},N\co\FPCTerms \ld x
                              \R{\tau} N \to f ( x ) \R{\sigma} (M N) \\
    x \R{\foldedtype} M & \eqdef \Sigma M' M''.
                          \fpcunfold{M} \to_*^0{M'}\toone{}{M''}
                          \text{ and } x \laterR{\unfoldedtype} \pure{}(M'')
  \end{align*}
  \caption{The logical relation
    $\mathcal{R}_\tau : \den{\tau} \times \FPCTerms \to \U{}$.}
  \label{fig:fpc:adequacy:relation}
\end{figure}

\subsection{Proof of computational adequacy}
Before proving computational adequacy we need to show some key
properties about the logical relation $\Rbare$. 
The first of these is that the relation respects the \emph{applicative}
structure  of the $\laterbare$ operator which is that we can
apply an argument that will be available later to a function that will
also be available later. 

\begin{lemma} \label{lem:fpc:r:app} If
  $f \laterR{\tau \to \sigma} \pure{}{(M)}$ and
  $r \laterR{\tau} \pure{}{(L)}$ then
  \[(f \circledast r) \laterR{\sigma} \pure{}{ (M L)}\]
\end{lemma}
\begin{proof}
  By definition $f \laterR{\tau \to \sigma} \pure{}{(M)}$ is type
  equal to \[\later{}[\hrt{x \gets f}]{(x \R{\tau \to \sigma} M)}\]
  which by definition is
  \[
  \later{}[\hrt{x \gets f}]{}(\Pi (y : \den{\tau}) (L:\FPCTerms) .  y
  \R{\tau} L \to x (y) \R{\sigma} M L)
  \]
  By applying the latter to $r$ and $\purebare L$ using the generalised later application of (\ref{eq:app:gen:def})
  we get an element of 
\begin{align*}
  & \later{}[\hrt{x \gets f,y \gets r, L \gets \purebare L}]{}( y \R{\tau} L \to x (y) \R{\sigma} M L) \\
  & \judgeq  \later{}[\hrt{x \gets f,y \gets r}]{}( y \R{\tau} L \to x (y) \R{\sigma} M L)
\end{align*}
 By further applying this to the hypothesis 
  $r \laterR{\tau} \pure{}{(L)}\judgeq \later{}[\hrt{y \gets r}]{}(y \R{\tau} L)$ we get
  \[\later{}[\hrt{x \gets f, y \gets r}]{}(x(y) \R{\sigma} M L)\]
  which is equivalent to
  $(f \circledast r) \laterR{\sigma} \pure{}{(M L)}$, thus concluding
  the case.
\end{proof}
Next we show that the relation is agnostic to $0$-step reduction in
the operational semantics. 
\begin{lemma} \label{lem:fpc:r:zeromany:right}\label{lem:fpc:r:zeromany:left}
  If ${M}\to^0{N}$ then $x \R{\sigma} M$ iff $x \R{\sigma} N$.
\end{lemma}
%
\begin{proof}
  We prove first the left to right implication by induction on $\sigma$, and show just a few cases.
%

  In the case of coproducts, we proceed by case
  analysis on $x$.  In the case of
  $x = \now{}{}(\sinl (y))$, by the assumption we have that
  $\zeromany{M}{\fpcinl{(N')}}$ and $y \R{\tau_{1}} N'$. If 
  $M = \fpcinl{(N')}$, then by $N$ must be of the form $\fpcinl{(N'')}$ for 
  some $N''$, such that ${N'}\to^0{N''}$. In this case, by induction 
  hypothesis $y \R{\tau_{1}} N''$ and so $x \R{\tau_1+\tau_2} N$. 
  If the reduction $\zeromany{M}{\fpcinl{(N')}}$ has positive length,
  by determinancy of the operational semantics 
  (Lemma~\ref{lem:fpc:small-step:det}) we get
  $\zeromany{N}{\fpcinl{N'}}$, and thus 
  $x \R{\tau_1 + \tau_2} N$. The case where
  $x = \now{}{}(\sinr (y))$ is similar.  When $x = \tick{}{\tau_1 + \tau_2} (y)$, by
  the assumption $x \R{\tau_1 + \tau_2} M$ there exist
  $M'$ and $M''$ such that $\zeromany{M}{M'}$ and $\toone{M'}{M''}$
  and $y \laterR{\fpcunittype} \purebare (M'')$. Again by determinancy of 
  the operational semantics, $N \to^0_* M'$ and thus we conclude 
  $x \R{\tau_1 + \tau_2} N$.

%
  Now we consider the case for recursive types. By assumption we know
  there exists $M'$ and $M''$ such that $\zeromany{\fpcunfold{M}}{M'}$
  and $\toone{M'}{M''}$ and
  $x \laterR{\unfoldedtype} \purebare (M'')$.  Since $\tozero{M}{N}$
  then also $\tozero{\fpcunfold{M}}{\fpcunfold{N}}$.  Therefore, from
  the assumption and the fact that the operational semantics is
  deterministic (Lemma~\ref{lem:fpc:small-step:det}) we get
  $\zeromany{\fpcunfold{N}}{M'}$.  By definition of the logical
  relation we get $x \R{\foldedtype} N$, which concludes the proof.
%

The proof of the right to left implication is also by induction on the structure of $\sigma$. 
Again we just show a few cases. 

In the case of the unit type, we proceed by case analysis on $x$.  When
  $x = \now{}{}(\sunit)$ we have that $\zeromany{N}{\fpcunit}$. Since
  $M \to^0 N$ we get $\zeromany{M}{\fpcunit}$ as required. When $x$ is
  $\tick{}{\fpcunittype} (x')$ by assumption $x \R{\fpcunittype} N$
  implies that there exists $N'$ and $N''$ such that
  $\zeromany{N}{N'}$ and $\toone{N'}{N''}$ and
  $x' \laterR{\fpcunittype} \purebare(N'')$. Since also $\zeromany{M}{N'}$
  this implies $x \R{\fpcunittype} M$.

  In the case of recursive types, by assumption we have that
  $x \R{\foldedtype} N$ and $\tozero{M}{N}$.  From the former we
  derive that there exists $M'$ and $M''$ such that
  $\zeromany{\fpcunfold{N}}{M'}$,$\toone{M'}{M''}$ and
  $x \R\unfoldedtype \purebare(M'')$.  Since $\tozero{M}{N}$ then also
  $\tozero{\fpcunfold{M}}{\fpcunfold{N}}$.  Therefore, we know that
  $\zeromany{\fpcunfold{M}}{M'}$, thus by definition of the logical
  relation we conclude.
\end{proof}

%

Now we show a key property of the logical relation. This states that
for programs that are related later after a $1$-step operational
reduction are related now. Note that in the interpretation of the unit
type we used the lifting monad. This was not strictly necessary to get
a ``tick'' algebra structure, but it is crucial to make the following
lemma to work.
\begin{lemma} \label{lem:fpc:r:tick} If $x \laterR{\tau} \pure{}{(M)}$
  and $\toone{M'}{M}$ then $\tick{}{\tau}(x) \R{\tau} M'$.
\end{lemma}
\begin{proof}
  The proof is by guarded recursion, so we assume
  that the lemma is ``later true'', i.e., that we have an inhabitant
  of the type obtained by applying $\laterbare$ to the statement of
  the lemma.  We proceed by induction on $\tau$.

  The cases for the unit type and for the coproduct are straightforward
  by definition. In the case for products, by assumption we have
  \[y \laterR{\tau_1 \times \tau_2} \purebare (M) \, .\]
  Unfolding definitions we get
  \[ \later{}[\hrt{x \gets y}]{(\pi_1(x) \R{\tau_1} (\fpcfst{M}))
    \text{ and }(\pi_2(x) \R{\tau_2} \fpcfst{(M)})}
  \]
  which implies
  \[
    (\pi_1(y)) \laterR{\tau_1}   \purebare (\fpcfst{M}) \qquad \text{and} \qquad  \pi_2(y) \laterR{\tau_2}
                                                   \purebare (\fpcsnd{M})
  \]
  Since $\toone{M'}{M}$ then also $\toone{\fpcfst{M'}}{\fpcfst{M}}$
  and $\toone{\fpcsnd{M'}}{\fpcsnd{M}}$, thus we can use the induction
  hypothesis on $\tau_1$ and $\tau_2$ and get
  \[
     \tick{}{\tau_1}(\pi_1(y)) \R{\tau_1}
      \fpcfst{M'} \qquad \text{and} \qquad \tick{}{\tau_2}(\pi_2(y)) \R{\tau_2}
                                        \fpcsnd{M'}
  \]
  by definition $\tick{}{\tau_1 \times \tau_2}$ commutes with $\pi_1$
  and $\pi_2$.  Thus, we obtain
  \[
     \pi_1(\tick{}{\tau_1 \times \tau_2}(y)) \R{\tau_1}
      \fpcfst{M'} \qquad \text{and} \qquad  \pi_2(\tick{}{\tau_1 \times \tau_2}(y)) \R{\tau_2}
                                        \fpcsnd{M'}
  \]
  which is by definition what we wanted.

  Now the case for the function space. Assume
  $f \laterR{\tau_1 \to \tau_2} \purebare (M)$ and $M' \to^1 M$.
  We must show that if $y : \den{\tau_1}$, $N : \FPCTerms$ and
  $y \R{\tau_1} N$ then
  $(\tick{}{\tau_1 \to \tau_2} (f)) (y) \R{\tau_2} (M N)$.
  So suppose $y \R{\tau_1} N$, and thus also
  $\later{}{(y \R{\tau_1} N)}$ which is equal to
  $\purebare (y) \laterR{\tau_1} \purebare (N)$.  By applying
  Lemma~\ref{lem:fpc:r:app} to this and
  $f \laterR{\tau_1 \to \tau_2} \purebare (M)$ we get
  \[
  f \app(\purebare (y)) \laterR{\tau_2} \purebare (M N)
  \]
  Since $M' \to^1 M$ also $M' N \to^1 M N$, and thus, by the induction
  hypothesis for $\tau_2$,
  $\tick{}{\tau_2}(f \app(\purebare(y))) \R{\tau_2} M' N$.
  Since by definition
  $\tick{}{\tau_1\to\tau_2}(f) (y) = \tick{}{\tau_2}(f
  \app\purebare(y))$, this proves the case.

  The interesting case is the one of $\foldedtype$.  Assume
  $x \laterR{\foldedtype} \purebare (M)$ and $\toone{M'}{M}$.  By
  definition of $\laterR{}$ this implies
  $\later{}[\hrt{y \gets x}]{(y \R{\foldedtype} M)}$ which by
  definition of $\R{\foldedtype}$ is
  \[
    \later{}[\hrt{y \gets x}]{}  \Sigma N'
                                  N''. \zeromany{\fpcunfold{M}}{N'} \text{ and } 
                                 \toone{N'}{N''}
                                  \text { and } (y \laterR{\unfoldedtype} \purebare (N''))
  \]
  Since zero-step reductions cannot eliminate outer
  $\fpcunfoldbare$'s, $N'$ must be on the form $\fpcunfold{N}$ for
  some $N$, such that $\zeromany{M}{N}$.  Thus, we can apply the
  guarded induction hypothesis to get
  \[
    \later{}[\hrt{y \gets x}]{}  ( \Sigma N. \zeromany{M}{N} \text{ and }
                                         (\tick{}{\unfoldedtype}(y) \R{\unfoldedtype} \fpcunfold{N})) 
  \]
  Since $\zeromany{\fpcunfold{M}}{\fpcunfold{N}}$, by
  Lemma~\ref{lem:fpc:r:zeromany:right} we get
  \[
  \later{}[\hrt{y \gets x}]{(\tick{}{\unfoldedtype}(y)
    \R{\unfoldedtype} \fpcunfold{M})}
  \]
  which by (\ref{eq:fpc:laterR:next}) is
  \[
  \pure{}[\hrt{y \gets x}]{(\tick{}{\unfoldedtype}(y))}
  \laterR{\unfoldedtype} \purebare (\fpcunfold{M})
  \]
  By (\ref{eq:app:def}) this implies
  \[
  \purebare (\tick{}{\unfoldedtype}) \app x \laterR{\unfoldedtype}
  \purebare (\fpcunfold{M})
  \]
  Since by assumption $\toone{M'}{M}$ also
  $\toone{\fpcunfold{M'}}{\fpcunfold{M}}$ thus, by definition of the
  logical relation
  \[
  \purebare(\tick{}{\unfoldedtype}) \app x \R{\foldedtype} M'
  \]
  By definition $\purebare(\tick{}{\unfoldedtype}) \app x$ is equal to
  $\tick{}{\foldedtype} (x)$ thus we can derive
  \[
  \tick{}{\foldedtype} (x) \R{\foldedtype} M'
  \]
  as we wanted.
\end{proof}

We can now finally state and prove the fundamental lemma stating that 
any term is related to its denotation in the logical relation of 
Figure~\ref{fig:fpc:adequacy:relation}. As we shall see below, this will imply computational
adequacy. 

\begin{lemma}[Fundamental Lemma]
  Suppose $\Gamma \vdash M : \tau$, for
  $\Gamma \equiv x_1 : \tau_1, \cdots, x_n : \tau_n$ and
  $\vdash N_i : \tau_i$, $\gamma_i : \den{\tau_i}$ and
  $\gamma_i \R{\den{\tau_i}} N_i$ for $i \in \{1 ,\dots, n\}$, then
  $\den{M}(\vec{\gamma}) \R{\tau} M [\vec{N}/\vec{x}]$
  \label{lem:fpc:fundamental}
\end{lemma}
\begin{proof}
The proof is by guarded recursion, and so we assume $\laterbare$ 
applied to the statement of the lemma. This implies that 
for all well-typed terms $M$ with context $\Gamma$ and type $\tau$ the
following holds:
  \[
  \later{}{(\den{M}(\vec{\gamma}) \R{\tau} M [\vec{N}/\vec{x}])}
  \]
  Then we proceed by induction on the typing derivation
  $\Gamma \vdash M : \tau$, showing only the interesting cases.  

  Consider first the case of $\Gamma \vdash \lambda x. M : \sigma \to \tau$.
  Assuming $\gamma_{n+1} \R{\sigma} M_{n+1}$, we must show 
  $\den{\lambda x.M} (\vec{\gamma}) (\gamma_{n+1}) \R\tau \den{M}(\vec{\gamma}, \gamma_{n+1})$. 
  Since
  \begin{align*}
    \den{\lambda x.M} (\vec{\gamma}) (\gamma_{n+1}) 
    & = \den{M}(\vec{\gamma}, \gamma_{n+1}) \\
    (\lambda x.M)[\vec{M}/\vec{x} ](M_{n+1}) & =  \lambda
                                                 x.(M[
                                                 \vec{M}/\vec{x}])(M_{n+1})
  \end{align*}
  and $\tozero{\lambda x.(M[ \vec{M}/\vec{x}])(M_{n+1})}{(M[\vec{M}/\vec{x}])[M_{n+1}/x]}$, by 
  Lemma~\ref{lem:fpc:r:zeromany:right} it suffices to prove
  \[\den{M} (\vec{\gamma}, \gamma_{n+1})
  \R{\tau}{M[\vec{M}/\vec{x}][M_{n+1}/x]}\]
  which follows from the induction hypothesis.
%
%
%

  For the case
  $\jud{\Gamma}{\fpcunfold{M}}{\tau[\mu\alpha.\tau/\alpha]}$ we must show that
  \[\den{\fpcunfold{M}}(\vec{\gamma}) \R{\unfoldedtype}
  (\fpcunfold{M})[\vec{N}/\vec{x}]\]
  By induction hypothesis we know that
  $\den{M}(\vec{\gamma}) \R{\foldedtype} (M [\vec{N}/\vec{x}])$
  which means that there exists $M'$ and $M''$ such that $\zeromany{\fpcunfold{(M [\vec{N}/\vec{x}])}}{M'}$ 
  and $\toone{M'}{M''}$ and
  $\den{M}(\vec{\gamma}) \laterR{\unfoldedtype} \purebare (M'')$.  By
  Lemma~\ref{lem:fpc:r:tick} then
  $\tick{}{\unfoldedtype}(\den{M}(\vec{\gamma})) \R{\unfoldedtype}
  M'$
  and since $\zeromany{\fpcunfold{(M [\vec{N}/\vec{x}])}}{M'}$ by
  repeated application of Lemma~\ref{lem:fpc:r:zeromany:right} we get
  \[\tick{}{\unfoldedtype}(\den{M}(\vec{\gamma})) \R{\unfoldedtype}
  \fpcunfold{(M [ \vec{N}/\vec{x}])}\]
  Since by definition $\den{\fpcunfold{M}}(\vec{\gamma}) = 
  \tick{}{\unfoldedtype}(\den{M}(\vec{\gamma}))$ this finishes the proof of the case. 

  For the case $\jud{\Gamma}{\fpcfold{M}}{\mu \alpha. \tau}$ we want
  to show that
  \[\den{\fpcfold{M}}(\vec{\gamma}) \R{\foldedtype}
  (\fpcfold{M})[ \vec{N}/\vec{x}]\]
  By definition of the
  logical relation we have to show that there exist $M'$ and $M''$
  such that
  \[\zeromany{\fpcunfold{(\fpcfold{(M [ \vec{N}/\vec{x}])})}}{M'}\]
  $\toone{M'}{M''}$ and that
  $\den{\fpcfold{M}}(\vec{\gamma}) \laterR{\unfoldedtype} \purebare(
  M'' )$.
  Setting $M''$ to be $(M [ \vec{N}/\vec{x}])$, we are left to show
  that
  \[ \den{\fpcfold{M}}(\vec{\gamma}) \laterR{\unfoldedtype} \purebare
  (M [ \vec{N}/\vec{x}]) \]
  which is equal by definition of the interpretation function to
  \[\purebare( \den{M}(\vec{\gamma} )) \laterR{\unfoldedtype}
  \purebare( (M [ \vec{N}/\vec{x}]) )\]
  The latter is equal by (\ref{eq:fpc:laterR:next}) to
  $\later{}{(\den{M}(\vec{\gamma}) \R{\unfoldedtype} (M [
    \vec{N}/\vec{x}]))}$ which is true by the guarded recursive
  hypothesis.
  
  For the case $\jud{\Gamma}{\fpcinl{M}}{\tau_1 + \tau_2}$ we have to
  prove that
  \[
  \den{\fpcinl{M}}(\vec{\gamma}) \R{\tau_1 + \tau_2} \fpcinl{M}[
  \vec{M}/\vec{x}]
  \]
  By definition of the interpretation function
  $\den{\fpcinl{M}}(\vec{\gamma})$ is equal to
  $\now{}{} (\sinl (\den{M}(\vec{\gamma})))$. By definition of the
  logical relation we have to prove that there exists $M'$ such that
  \[\manyk{(\fpcinl{M})[ \vec{M}/\vec{x}]}{0}{\fpcinl{M'}} \text{ and } \den{M}(\vec{\gamma}) \R{\tau_1} M'.\]  
  The former is trivially true
  with $M' = M [ \vec{M}/\vec{x}]$ and the latter is by induction
  hypothesis.  The case for
  $\jud{\Gamma}{\fpcinr{N}}{\tau_1 + \tau_2}$ is similar.
  
  For the case $\jud{\Gamma}{\fpccase{L}{M}{N}}{\sigma}$ we have to
  prove that
  \[
  \den{\fpccase{L}{M}{N}} (\vec{\gamma}) \R{\sigma}
  (\fpccase{L}{M}{N})[ \vec{M}/\vec{x}]
  \]
  For this it suffices to prove
  \begin{equation}
       \den{\lambda x. \fpccase{x}{M}{N}} (\vec{\gamma}) \R{\tau_1 +
        \tau_2
        \to \sigma}  (\lambda x. \fpccase{x}{M}{N}) [ \vec{M}/\vec{x}]
    \label{lem:fpc:fundamental:case:gr}
  \end{equation}
  and then applying this to
  $\den{L}(\vec{\gamma}) \R{\tau_1 + \tau_2} L[ \vec{M}/\vec{x}]$.  We
  prove (\ref{lem:fpc:fundamental:case:gr}) by guarded recursion thus
  assuming the statement is true later.

  Assume $y$ of type $\den{\tau_1 + \tau_2}$, $L$ a term, and $y \R{\tau_1 + \tau_2} L$.
  We proceed by case analysis on $y$ which is of type
  $\den{\tau_1 + \tau_2}$ which by definition is
  $L(\den{\tau_1} + \den{\tau_2})$.  In the case
  $y = \now{}{}(\sinl (z))$, where $z$ is of type
  $\den{\tau_1}$ we know by assumption that there exists $L'$ s.t.
  $\manyk{L}{0}{\fpcinl{(L')}}$ and $z \R{\tau_1} L'$. Since
  \[
  \den{\lambda x. \fpccase{x}{M}{N}}
  (\vec{\gamma})(\now{}{}(\sinl(z))) =\den{M}(\vec{\gamma},
  z)
  \] and
  \[
  \fpccase{L}{M [ \vec{M}/\vec{x}]}{N[ \vec{M}/\vec{x}]} \Rightarrow^0
  M[ \vec{M}/\vec{x}] [L'/x_1]
  \]
  by Lemma~\ref{lem:fpc:r:zeromany:left} we are left to prove
  \[
  \den{M}(\vec{\gamma}, \gamma) \R{\sigma} \ M[ \vec{M}/\vec{x}] [
  L'/x_1]
  \]
  which is true by induction hypothesis.
  The case $y = \now{}{}(\sinr (z))$ where $z$ is of
  type $\den{\tau_2}$ is similar.

  Now consider the case of 
  $y = \tick{}{\tau_1 + \tau_2} (z)$, where $z$ is of
  type $\later{}{\den{\tau_1 + \tau_2}}$.  By induction hypothesis we
  know that
  $\theta_{\tau_1 + \tau_2} (z) \R{\tau_1 + \tau_2} L $, thus
  there exist $L'$ and $L''$ of type $\FPCTerms$ such that
  $L \to_*^0 L' $, $L' \to^1 L''$ and
  $z \laterR{\tau_1 + \tau_2} \pure{}{( L'' )}$.

  Recall that we have assumed $\laterbare$ of (\ref{lem:fpc:fundamental:case:gr}), i.e.,
  \[
     \laterbare(\den{\lambda x. \fpccase{x}{M}{N}} (\vec{\gamma}) 
     \R{\tau_1 + \tau_2 \to \sigma}  
     (\lambda x. \fpccase{x}{M}{N}) [ \vec{M}/\vec{x}] ) 
  \]
  which is type equal to
  \[
     \pure{}{(\den{\lambda x. \fpccase{x}{M}{N}} (\vec{\gamma}))}
     \laterR{\tau_1 + \tau_2 \to \sigma} 
     \pure{}{((\lambda x. \fpccase{x}{M}{N}) [ \vec{M}/\vec{x}])}
  \]
  By Lemma~\ref{lem:fpc:r:app} we can apply this to the assumption $z \laterR{\tau_1 + \tau_2} \pure{}{( L'' )}$
  thus getting
  \[
     \purebare(\den{\lambda x. \fpccase{x}{M}{N}}(\vec{\gamma}))\circledast z 
     \laterR{\sigma} 
     \purebare(((\lambda x. \fpccase{x}{M}{N})[ \vec{M}/\vec{x}])(L'')) 
  \]  
  Since $\toone{L'}{L''}$ we can apply Lemma~\ref{lem:fpc:r:tick} and
  obtain
  \begin{align*}
     \theta_\sigma (\pure{}{( \den{\lambda x. \fpccase{x}{M}{N}}  (\vec{\gamma}))} \circledast z) 
     \R{\sigma} 
     \fpccase{L'}{M [ \vec{M}/\vec{x}]}{N [ \vec{M}/\vec{x}]} 
  \end{align*}
  By Lemma~\ref{lem:fpc:r:zeromany:left} with the fact that
  $L \to^0_* L'$ we get
  \begin{align*}
     \theta_\sigma (\pure{}{( \den{\lambda x. \fpccase{x}{M}{N}}
      (\vec{\gamma}))} \circledast z)
    \R{\sigma} 
    \fpccase{L}{M [ \vec{M}/\vec{x}]}{N [ \vec{M}/\vec{x}]} 
  \end{align*}
  And finally by simplifying the left-hand side using
  Lemma~\ref{lem:fpc:case:hom}:
  \begin{align*}
     \theta_\sigma (\pure{}{( \den{\lambda x. \fpccase{x}{M}{N}}
      (\vec{\gamma}))} \circledast z)  
    & =  \den{\lambda  x. \fpccase{x}{M}{N}} (\vec{\gamma}) (y)
  \end{align*}
  thus getting
  \[
      \den{\lambda x. \fpccase{x}{M}{N}} (\vec{\gamma})(y)
     \R{ \sigma}
     (\lambda x. \fpccase{x}{M}{N}) [ \vec{M}/\vec{x}]  (L)
  \]
  as we wanted.
\end{proof}
From the Fundamental lemma we can now prove computational adequacy.
\begin{theorem}[Intensional Computational Adequacy]\label{thm:fpc:adequacy}
  If $M\co \fpcunittype$ is a closed term then $\manyk{M}{k}{\fpcunit}$
  iff $\den{M}(*) \propeq \delaybare^k(\eta(\sunit))$.
\end{theorem}
\begin{proof}
  The left to right implication is soundness
  (Proposition~\ref{prop:fpc:soundness}). For the right to left implication note
  first that the Fundamental Lemma (Lemma~\ref{lem:fpc:fundamental})
  implies $\delta^k(\eta(\sunit)) \R{\fpcunittype} M$.  To complete the
  proof it suffices to show that
  $\delta^k_\fpcunittype (\eta(\sunit)) \R{\fpcunittype} M$ implies
  $\manyk{M}{k}{\fpcunit}$.

  This is proved by guarded recursion and induction on $k$: the case of $k = 0$ is
  immediate by definition of $\R{\fpcunittype}$. If $k = k' + 1$ first
  assume $\delta^k_\fpcunittype (\eta(\sunit)) \R{\fpcunittype} M$.  By
  definition of $\R{}$ there exist $M'$ and $M''$ such that
  $M \to^0_* M'$, $M' \to^1 M''$ and
  $\pure{}{(\delta^{k'}_{\fpcunittype}(\eta(\sunit)))} \laterR{\fpcunittype}
  \pure{}{(M'')}$
  which is type equal to
  $\later{}{(\delta_{\fpcunittype}^{k'}(\eta(\sunit)) \R{\fpcunittype}
    M'')}$.
  By the guarded recursion assumption we get
  $\later{}{(\manyk{M''}{k'}{\fpcunit})}$ which by definition implies
  $\manyk{M}{k}{\fpcunit}$.
\end{proof}

From Theorem~\ref{thm:fpc:adequacy} one can deduce that whenever two
terms have equal denotations they are contextually equivalent in a
very intensional way, as we now describe. By a context, we mean a term
$C[-]$ with a hole, and we say that $C[-]$ has type
$\Gamma,\tau \to (-, \fpcunittype)$ if $C[M]$ is a closed term of type
$\fpcunittype$, whenever $\hastype{}{\Gamma}{M}{\tau}$.

\begin{corollary}
  Suppose $\hastype{}{\Gamma}{M}{\tau}$ and $\den M \propeq \den N$.
  If $C[-]$ has type $\Gamma,\tau \to (-, \fpcunittype)$ and
  $ \manyk{C[M]}{k}{\fpcunit}$ also $\manyk{C[N]}{k}{\fpcunit}$.
\end{corollary}

As stated above, this is a very intensional result in the sense that
whenever two FPC-denotable programs are equal we can derive that, under
any context, they reduce to the same value with the same number of
computational steps. This means that our model distinguishes programs
whose input-output behaviour is the same, but the way in which the
result is computed is computationally different. More specifically, two different
algorithms implementing the same specification, but with a different
computational complexity, will be considered different in the model. 
We explain how to recover this extensionality via a logical relation
in the next section.

\section{Extensional Computational Adequacy}
\label{sec:extensional}

Our model of FPC is intensional in the sense that it distinguishes between
computations computing the same value in a different number of steps. 
In this section we define a logical relation which relates elements of the model
if they differ only by a finite number of computation steps. In particular, this
also means relating $\bot$ to $\bot$. 

Such a relation must be defined on the types of the form $\forall \kappa . \den \sigma$ 
rather than directly on the types $\den\sigma$. To see why, consider the case of 
$\sigma = \fpcunittype$, in which case $\den\sigma = L1$. 
Recall from Section~\ref{sec:fpc:ToT} that in the topos of trees model $L1$
is interpreted as the family of sets
\[L1(n) = \{\bot, 0, 1, \dots, n-1\}\]
which describes computations terminating in at most $n-1$ steps or using at
least $n$ steps (corresponding to $\bot$). It cannot distinguish between termination in
more than $n-1$ steps and real divergence.
Our relation should relate 
a terminating value $x$ in $L1(n)$ to any other terminating value, but not real divergence,
which is impossible, if divergence cannot be distinguished from slow termination.
Another, more semantic, way to phrase the problem is that termination as
described by the subsets $\{0, 1, \dots, n-1\}$ of $L1(n)$ for each $n$ does not 
form a subobject of $L1$. 

On the other hand, if $L1 \cong 1 + \laterclock\kappa 1$ then, as we saw in section 
the type 
\begin{equation*}
  L^{\glob}A  \eqdef \forall \kappa. LA
  \label{eq:fpc:co-inductive:lifting-monad}
\end{equation*}
is a coinductive solution to the type equation
\[
  L^{\glob}1 \cong 1 + L^{\glob}1
\]
Semantically $L^{\glob}1$ is modelled as 
the set $\N + \{\bot\}$, and termination is the subset of this corresponding to the left 
inclusion of $\N$. So on the global level we can, at least semantically, distinguish between
termination and non-termination. This is reflected syntactically in Lemma~\ref{lem:fpc:wb:den:gen}.

We refer to $\forall\kappa . \den \fpcunittype \judgeq L^{\glob}1$ as the \emph{global interpretation} of
the type $\fpcunittype$ because it captures the global behaviour (computable in \emph{any} number 
of steps) of terms of type $\fpcunittype$. We now extend this to the global interpretation of all types and terms 
and give the definition of the logical relation.
 
\subsection{Global interpretation of types and terms}

Recall that the developments above should be read as taking place in a context of an
implicit clock $\kappa$. 
To be consistent with the notation 
of the previous sections, $\kappa$ will remain implicit in the denotations
of types and terms, although one might choose to write e.g. $\den\sigma^\kappa$ 
to make the clock explicit. 

We define global interpretations of types and terms as
follows:
\begin{equation*}
  \begin{aligned}
    \denglobal{\sigma} & \eqdef\forall \kappa. \den{\sigma}\\
    \denglobal{M} &\eqdef \Lambda \kappa. \den{M}
  \end{aligned}
\end{equation*}
such that if $\hastype{}\Gamma M\tau$, then
\[
\denglob M\co \alwaystype{\kappa}{(\den\Gamma \to \den \tau)}
\] 
Note that $\denglobal{\sigma}$ is a wellformed type, because $\den\sigma$ is a wellformed
type in context $\sigma\co \FPCTypes$ and $\FPCTypes$ is an inductive
type formed without reference to clocks or guarded recursion, thus
$\kappa$ does not appear in $\FPCTypes$. By a similar argument $\denglobal{M}$
is welltyped.

Define for all $\sigma$ the \emph{delay} operator
$\delayglob{\sigma}\co \denglobal{\sigma} \to \denglobal{\sigma}$ as follows
\begin{equation}
  \delayglob{\sigma} (x) \eqdef \Lambda
  \kappa. \delay{}{\sigma} (x[\kappa])
  \label{def:delay:glob}
\end{equation}
Similarly for $LA$, $\delayglob{LA}(x) \eqdef \Lambda
\kappa. \delay{}{LA} (x[\kappa])$.

%

With these definitions we can lift the adequacy theorem to the global
points. To prove the denotational model
is computationally adequate w.r.t. the standard big-step operational
semantics $\bigstep^n$ we take the global view points of the the
denotational semantics in order to be able to remove the occurrences
of the $\laterbare$ operator. 
 \begin{corollary}[Computational adequacy]
\label{cor:fpc:gl:adequacy}
  If $M\co \fpcunittype$ is a closed term and $n$ is a natural number, then 
  $M \bigstep^{n} \fpcunit$ iff
  $\forall \kappa. \den{M}(*) \propeq \delaybare^n(\eta(\sunit))$.
\end{corollary}

\begin{proof}
Since $\forall\kappa.(-)$ is functorial, Theorem~\ref{thm:fpc:adequacy}
gives $\forall \kappa. \den{M}(\sunit) \propeq \delaybare^n(\eta(\sunit))$ 
iff $\forall \kappa. \manyk{M}{n}{\fpcunit}$, which by 
Lemma~\ref{lem:fpc:bigstep:many:manyk:soundness}
holds iff $M \bigstep^{n} \fpcunit$.
%
\end{proof}

We have now a semantics that implies the standard operational semantics.
However, we are still not able to prove that if two programs are equal
they are going to be contextually equivalent w.r.t. the input-output
behaviour. To achieve so, we need to lift the explicit step-indexing
as well. 
\subsection{A weak bisimulation relation for the lifting monad}
Before defining the logical relation on the interpretation of types,
we define a relational version of the guarded recursive lifting monad
$L$.  If applied to the identity relation on a type $A$ in which
$\kappa$ does not appear, we obtain a weak bisimulation relation
similar to the one defined by Capretta~\cite{Cap05} for the
coinductive partiality monad.

\begin{definition}
  For a relation $R : A \times B \to \U{}$ define the lifting
  $\liftrel: LA \times LB \to \U{}$ by guarded recursion 
  and case analysis on the elements of $LA$ and $LB$:
  \begin{equation*} 
    \begin{aligned}
      \now{\kappa}(x)\liftrel\now{\kappa} (y) & \eqdef x\;R\;y \\
      \now{\kappa}(x)\liftrel\tick{\kappa}{LB}(y) & \eqdef \Sigma
      n,y'. \tick{\kappa}{LB}(y) \propeq \delay{\kappa}{LB}[n](\now{\kappa}(y'))
      \text{ and } x\;R\;y'\\
      \tick{\kappa}{LA}(x)\liftrel\now{\kappa}(y) & \eqdef \Sigma
      n,x'. \tick{\kappa}{LA}(x) \propeq
      \delay{\kappa}{LA}[n](\now{\kappa}(x')) \text{ and } x'\;R\;y \\
      \tick{\kappa}{LA}(x)\liftrel\tick{\kappa}{LB}(y) & \eqdef
      x\laterliftrel y
    \end{aligned}
  \end{equation*}
  \label{def:fpc:liftrel}
\end{definition}
Intuitively, $LR$ relates two elements if they either both diverge, or both both converge
to elements related in $R$.
For example, $\bot$ as defined in Section~\ref{sec:fpc:den} 
is always related to itself which can be shown by guarded recursion 
as follows. Suppose $\laterbare (\bot \liftrel \bot)$. Since $\bot = 
\tickbare (\purebare (\bot))$, to prove $\bot \liftrel \bot$, we must prove 
$\purebare (\bot) \laterliftrel \purebare (\bot)$. But, this type is equal to 
the assumption $\laterbare (\bot \liftrel \bot)$ 
by (\ref{eq:fpc:laterR:next}).

By the intuition given for $LR$ below, it should be possible to add or remove ticks on either side without breaking relatedness in $LR$. The next lemma shows half of this. 
\begin{lemma}
  If $R : A \times B \to \U{}$, and $x \liftrel y$ then
  $x \liftrel \delay{}{LB}(y)$ and $\delay{}{LA}(x) \liftrel y$. 
  \label{lem:fpc:liftrel:delay:add:sx:add:dx}
\end{lemma}
\begin{proof}
  Assume $x \liftrel y$. We show $x \liftrel \delay{}{LB}(y)$.  
  The proof is by guarded recursion, hence we
  first assume:
  \begin{equation}
    \later{\kappa}{(\Pi x : LA, y : LB.
      x \liftrel y \Rightarrow x \liftrel  \delay{\kappa}{LB} (y))}.
    \label{lem:fpc:liftrel:delay:dx:gr}
  \end{equation}
  We proceed by case analysis on $x$ and $y$.  If 
  $x \propeq \now{}(x')$, then, since $x \liftrel y$, there exist $n$ and $y'$ 
  such that $y \propeq \delay{\kappa}{LB}[n](\now{\kappa}(y'))$ and $x'\;R\;y'$. 
  So then  $\delay{\kappa}{LB} (y) \propeq \delay{\kappa}{LB}[n+1](\now{\kappa}(y'))$,
  from which it follows that $x \liftrel \delay{}{LB}(y)$.
      
  
  For the case
  where $x \propeq \tick{\kappa}{LA} (x')$ and $y \propeq \now{\kappa} (v)$,
  it suffices to show that
  $\delay{\kappa}{LA}[n](\now{}(w)) \liftrel \now{}(v)$ implies
  $\delay{\kappa}{LA}[n](\now{}(w)) \liftrel \delay{\kappa}{LB}(\now{}(v))$.
  The case of $n\propeq 0$ was proved above. For $n \propeq m+1$ we know that if 
  $\delay{\kappa}{LA}[n](\now{}(w)) \liftrel \now{}(v)$ also 
  $\delay{\kappa}{LA}[m](\now{}(w)) \liftrel \now{}(v)$ holds by
  definition, and this implies  \[\later{}{(\delay{\kappa}{LA}[m](\now{}(w)) \liftrel
    \now{}(v))}\] But this type can be rewritten as follows
\begin{align*} 
 \later{}{(\delay{\kappa}{LA}[m](\now{}(w)) \liftrel \now{}(v))}
 & \judgeq \purebare(\delay{\kappa}{LA}[m](\now{}(w)) \laterliftrel \purebare(\now{}(v))) \\
 & \judgeq \tick{\kappa}{LA}(\purebare(\delay{\kappa}{LA}[m](\now{}(w)))) \liftrel \tick{\kappa}{LB}(\purebare(\now{}(v)))) \\
 & \judgeq \delay{\kappa}{LA}[n] (\now{}(w))\liftrel \delay{\kappa}{LB}(\now{}(v))
\end{align*}
proving the case. 
  
  Finally, the case when $x \propeq \tick{\kappa}{LA} (x')$ and
  $y \propeq \tick{}{LB} (y')$. The assumption in this case is $x' \laterliftrel y'$, 
   which means by (\ref{def:fpc:rel}),
\[\later{}[\hrt{x'' \gets x' , y'' \gets y'}]{x'' \liftrel
    y''}\]
  By the guarded recursion hypothesis (\ref{lem:fpc:liftrel:delay:dx:gr}) we
  get
  \[\later{\kappa}[\hrt{x'' \gets x' , y'' \gets y'}]{x'' \liftrel
    \delay{\kappa}{LB} (y'')}\] which can be rewritten to
  \begin{equation}
    \later{}[\hrt{x'' \gets x' , y'' \gets y'}]{x'' \liftrel
      \tick{}{LB}(\pure{}{(y'')})}
    \label{eq:fpc:liftrel:delay:dx:next}
  \end{equation}
  By (\ref{eq:eta}) there is an inhabitant of the type
  \[\later{\kappa}[\hrt{x'' \gets x' , y'' \gets y'}](\pure{}{(y'')} \propeq y')\]
  and thus (\ref{eq:fpc:liftrel:delay:dx:next}) implies
  $\later{\kappa}[\hrt{x'' \gets x'}]{x'' \liftrel
    \tick{\kappa}{LB}(y')}$, which, by (\ref{eq:fpc:laterR:next}) 
    and since $y \propeq \tick{\kappa}{LB} (y')$ equals  
    $x' \laterliftrel \pure{\kappa}{(y)}$. By definition, this is  
    \[\tick{\kappa}{LA} (x') \liftrel \tick{\kappa}{LB}
  (\pure{\kappa}{(y)})\]
  which since $x \propeq \tick{\kappa}{LA} (x')$ is $x \liftrel \delay{\kappa}{LB} (y)$.
\end{proof}

We can lift this result to $\Lglob$ as follows. Suppose $R : A \times B \to \U{}$
and $\kappa$ not in $A$ or $B$. Define $\liftrelglob : \Lglob A \times \Lglob B \to \U{}$ as
\[x \liftrelglob y \eqdef \forall \kappa. x[\kappa] \liftrel y[\kappa] \]

\begin{lemma}  
  Let $x : \Lglob A$ and $y : \Lglob B$. If $x \liftrelglob y$
  then $x \liftrelglob \delayglob{} (y)$ and $\delayglob{} (x)
  \liftrelglob y$.
\label{lem:fpc:liftrelglob:delay:add:sx:add:dx}
\end{lemma}
\begin{proof}
Follows directly from Lemma~\ref{lem:fpc:liftrel:delay:add:sx:add:dx}.
\end{proof}
One might expect that $\delay{}{LA}(x) \liftrel \delay{}{LB}(y)$
implies $x \liftrel y$. This is not true, it only implies $\laterbare(x \liftrel y)$.
In the case of $\Lglob$, however, we can use $\force$ to remove 
the $\laterbare$. 
\begin{lemma}
  For all $x : \Lglob A$ and $y : \Lglob B$ and for all $R: A \times B
  \to \U{}$, if $\delayglob{LA} (x) \liftrelglob \delayglob{LB} (y)$ then
  $x \liftrelglob y$.
  \label{lem:fpc:liftrelglob:delay:rm:sx:dx}
\end{lemma}
\begin{proof}
  Assume $\delayglob{LA} (x) \liftrelglob \delayglob{LB} (y)$. We can
  rewrite this type by unfolding definitions and (\ref{eq:fpc:laterR:next}) as follows.
\begin{align*}
 \delayglob{LA} (x) \liftrelglob \delayglob{LB} (y) & \judgeq
   \forall \kappa. (\delayglob{LA} (x))[\kappa] \liftrel
  (\delayglob{LB} (y))[\kappa] \\
  & \judgeq \forall \kappa. (\delay{\kappa}{LA} (x[\kappa])) \liftrel
  (\delay{\kappa}{LB} (y[\kappa])) \\
  & \judgeq \forall \kappa. (\pure\kappa(x[\kappa]) \laterliftrel \pure\kappa(y[\kappa])) \\
  & \judgeq \forall \kappa. \later{\kappa}{(x[\kappa] \liftrel (y[\kappa]))}
\end{align*}
  Using $\force$ this implies $\forall \kappa. (x[\kappa] \liftrel (y[\kappa]))$
  which is equal to 
  $x \liftrelglob y$.
\end{proof}

\begin{lemma}
  For all $x$ of type $\Lglob A$ and $y$ of type $\Lglob B$, if
  $\delayglob{LA} (x) \liftrelglob y$ then $x \liftrelglob y$.
  \label{lem:fpc:liftrelglob:delay:rm:sx}
\end{lemma}
\begin{proof}
  Assume $\delayglob{LA} (x) \liftrelglob y$.  Then by applying
  Lemma~\ref{lem:fpc:liftrelglob:delay:add:sx:add:dx} we get
  $\delayglob{LA} (x) \liftrelglob \delayglob{LB} (y)$ and by applying
  Lemma~\ref{lem:fpc:liftrelglob:delay:rm:sx:dx} we get $x \liftrelglob y$.
\end{proof}


With this machinery in place we can now define a relation on the
semantics that relates programs that produce the same value (or both
diverge) and that discards the information about the number of delays
used.

\subsection{Relating terms up to extensional equivalence}
Figure~\ref{fig:fpc:wbisim} defines for each FPC type $\tau$
the logical relation $\wbisim{\kappa}{\tau} : \den\tau \times \den\tau \to
\U{}$. The definition is by guarded recursion, and well-definedness
can be formalised using an argument similar to that used 
for well-definedness of $\tick{}{}$ in equation (\ref{sec:tick:welldef}).
The case of recursive types is well typed by Lemma~\ref{lem:fpc:fix:eq}.
The figure uses the following
lifting of relations to sum types. 
\begin{definition}
  Let $R : A \times B \to \U{}$ and $R' : A' \times B' \to \U{}$.
  Define $(R + R') : (A + A') \times (B + B') \to \U{}$ by case
  analysis as follows (omitting false cases)
  \begin{equation*}\label{def:fpc:liftrelplus}
    \begin{aligned}
      \sinl (x)\;(R+R')\;\sinl(y) & \eqdef x\;R\;y\\
      \sinr (x)\;(R+R')\;\sinr(y) & \eqdef x\;R'\;y
    \end{aligned}
  \end{equation*}
\end{definition}

\begin{figure}[tb]
    \begin{align*}
      x \wbisim{\kappa}{\fpcunittype} y & \eqdef x\;L (\propeq_{\fpcunittype})\;y\\
      x \wbisim{\kappa}{\tau_1 + \tau_2} y & \eqdef x \;
      L(\wbisim{\kappa}{\tau_1}
      + \wbisim{\kappa}{\tau_2})\;y\\
      x \wbisim{\kappa}{\tau_1 \times \tau_2} y & \eqdef \pi_1 (x)
      \wbisim{\kappa}{\tau_1} \pi_1 (y) \text{ and } \pi_2 (x)
      \wbisim{\kappa}{\tau_2} \pi_2
      (y)\\
      f \wbisim{\kappa}{\sigma \to \tau} g & \eqdef \Pi (x, y :
      \den{\sigma}). x \wbisim{\kappa}{\sigma} y \to f ( x )
      \wbisim{\kappa}{\tau} g (
      y )\\
      x \wbisim{\kappa}{\foldedtype} y & \eqdef x \laterWBisim{\kappa}{\unfoldedtype} y
    \end{align*}
    \caption{The logical relation $\wbisim{\kappa}{\tau}$ is a predicate over denotations of
      $\tau$ of type $\den{\tau} \times \den{\tau} \to \U{}$}
 \label{fig:fpc:wbisim}
\end{figure}

The logical relation can be generalised to open terms and the global interpretation
of terms as in the next two definitions.
\begin{definition}\label{def:fpc:wbisim:open}
  For $\Gamma \equiv x_1 : \sigma_1, \cdots, x_n: \sigma_n$ and for
  $f$, $g$ of type $\den{\Gamma} \to \den{\tau}$ define
  \[
  f \wbisim{\kappa}{\Gamma,\tau} g \eqdef \depprod{\vec{x},
  \vec{y}}{\den{\vec{\sigma}}} {\vec{x}\wbisim{\kappa}{\vec{\sigma}} \vec{y} \to f(\vec{x})
  \wbisim{\kappa}{\tau} g(\vec{y})}
  \]
  For $x, y$ of type $\forall\kappa . (\den\Gamma \to \den\tau)$ define
  \[
  x \wbglob{\Gamma, \tau} y \eqdef 
  \forall \kappa. x[\kappa] \wbisim{\kappa}{\Gamma, \tau} y[\kappa]
  \]
\end{definition}

Perhaps surprisingly, this relation is not reflexive. For example the function $f \co L1 \to L1$ defined by
$f(\eta(\sunit)) = \eta(\sunit)$ and $f(\tick{}{L1}(x)) = \bot$ does not satisfy $f \wbisim{\kappa}{1 \to 1} f$. 
On the other hand, the denotation of any term is always related to itself, as the following
proposition states.
\begin{proposition} 
  \label{prop:fpc:wbisim:refl}
  If $\hastype{}{\Gamma}{M}{\sigma}$, then
  $\den{M} \wbisim{}{\Gamma, \sigma} \den{M}$.
\end{proposition}

The rest of this section is devoted to the proof of Proposition~\ref{prop:fpc:wbisim:refl}
which is important for the proof of the extensional computational adequacy theorem. 
To prove the proposition we first establish some basic properties of the logical relation.
The first lemma states that delayed 
application $\app$ respects the logical relation. 
\begin{lemma}
  For all $f, g$ of type $\later{\kappa}{\den{\tau \to \sigma}}$ and
  $x,y$ of type $\later{\kappa}{\den{\tau}}$, if
  $f \laterWBisim{\kappa}{\tau \to \sigma} g$ and
  $x \laterWBisim{\kappa}{\tau} y$ then
  $(f \app x) \laterWBisim{\kappa}{\sigma} (g \app y)$.
  \label{lem:fpc:wbisim:later:appl}
\end{lemma}
\begin{proof}
  Assume $f \laterWBisim{\kappa}{\tau \to \sigma} g$ and
  $x \laterWBisim{\kappa}{\tau} y$.  By Definition~\ref{def:fpc:rel}
  $f \laterWBisim{\kappa}{\tau \to \sigma} g$ is
  $\later{\kappa}[\hrt{f' \gets f, g' \gets g}]{(f'
    \wbisim{\kappa}{\tau \to \sigma} g')}$
  which by unfolding the definition of $\wbisim{\kappa}{\tau\to \sigma}$ is
  \[\later{\kappa}[\hrt{f' \gets f, g' \gets g}]{}(\Pi (x,y :
  \den{\sigma}). x \wbisim{\kappa}{\tau} y \to f' (x)
  \wbisim{\kappa}{\sigma} g'(y))\]
  By applying this to $x$, $y$ and $x \laterWBisim{\kappa}{\tau} y$ using the 
  dependent version of $\app$ defined in (\ref{eq:app:gen:def}) we get
    \[\later{\kappa}[\hrt{f' \gets f, g' \gets g, a \gets x, b \gets
    y}]{(f'(a) \wbisim{\kappa}{\sigma} g'(b))}\]
  By (\ref{eq:fpc:laterR:next}) this is equal to 
  \[\pure{\kappa}[\hrt{f' \gets f, a \gets x}]{(f' (a))}
  \laterWBisim{\kappa}{\sigma} \pure{\kappa}[\hrt{g' \gets g, b \gets y}]{(g'(b))}\] 
  which by rule (\ref{eq:app:def}) is equal to
  \[
  (f \app x) \laterWBisim{\kappa}{\sigma} (g \app y)
  \]
\end{proof}

Next we show that $\tick{}{}$ respects the logical relation.

\begin{lemma} Let $x, y$ of type $\later{\kappa}{\den{\sigma}}$, if
  $(x \laterWBisim{\kappa}{\sigma} y) $ then
  $ \tick{\kappa}{\sigma} (x) \wbisim{\kappa}{\sigma}
  \tick{\kappa}{\sigma}(y)$
  \label{lem:fpc:wbisim:tick}
\end{lemma}
\begin{proof}
  We prove the statement by guarded recursion. Thus, we assume
  the statement holds ``later'' and we proceed by induction on
  $\sigma$.  All the cases for the types that are interpreted using
  the lifting -- namely the unit type and the sum type -- in
  Definition~\ref{def:fpc:liftrel} hold by definition of the lifting
  relation.

  First the case for the function types: Assume
  $\sigma = \tau_1 \to \tau_2$ and assume $f$ and $g$ of type
  $\later{\kappa}{\den{\tau_1 \to \tau_2}}$ such that
  $f \laterWBisim{\kappa}{\tau_1 \to \tau_2} g$. We must show that if
  $x,y : \den{\tau_1}^\kappa$ and $x \wbisim{\kappa}{\tau_1} y$
  then
  $(\tick{\kappa}{\tau_1 \to \tau_2} (f)) (x) \wbisim{\kappa}{\tau_2}
  (\tick{\kappa}{\tau_1 \to \tau_2} (g)) (y))$.

  So suppose $x \wbisim{\kappa}{\tau_1} y$, then also
  $\later{\kappa}{(x \wbisim{\kappa}{\tau_1} y)}$, which by
  (\ref{eq:fpc:laterR:next}) is equal to
  $\pure{\kappa}{(x)} \laterWBisim{\kappa}{\tau_1}
  \pure{\kappa}{(y)}$.
  By applying Lemma~\ref{lem:fpc:wbisim:later:appl} to this and
  $f \laterWBisim{\kappa}{\tau_1 \to \tau_2} g$ we get
  \[
  f \app(\pure{\kappa}{x}) \laterWBisim{\kappa}{\tau_2} g
  \app \pure{\kappa}{y}
  \]
  By induction hypothesis on $\tau_2$, we get
  $\tick{\kappa}{\tau_2}(f \app(\pure{\kappa}{x}))
  \wbisim{\kappa}{\tau_2} \tick{\kappa}{\tau_2}(g
  \app(\pure{\kappa}{y}))$.
  We conclude by observing that by definition of $\tick{\kappa}{}$,
  $\tick{\kappa}{\tau_1\to\tau_2} (f) (x) = \tick{\kappa}{\tau_2} (f
  \app\pure{\kappa}{(x)})$.

  The case of the product is straightforward. 
    
  For the case of recursive types, assume
  $\phi \laterWBisim{\kappa}{\foldedtype} \psi$. This is type equal to
  \[
  \later{\kappa}[\hrt{x \gets \phi, y \gets \psi}]{(x
    \wbisim{\kappa}{\foldedtype} y)}
  \]
  By definition this is equal to
  \[
  \later{\kappa}[\hrt{x \gets \phi, y \gets \psi}]{(x
    \laterWBisim{\kappa}{\unfoldedtype} y)}
  \]
  By the guarded recursion hypothesis we get
  \[
  \later{\kappa}[\hrt{x \gets \phi, y \gets \psi}]{(
    \tick{\kappa}{\unfoldedtype}(x) \wbisim{\kappa}{\unfoldedtype}
    \tick{\kappa}{\unfoldedtype} (y) )}
  \]
  By (\ref{eq:fpc:laterR:next}) this is equal to
  \[
  (\pure{}[\hrt{x \gets \phi}]{(\tick{\kappa}{\unfoldedtype}(x))})
  \laterWBisim{\kappa}{\unfoldedtype} (\pure{}[\hrt{y \gets
    \psi}]{(\tick{\kappa}{\unfoldedtype}(y))})
  \]
  This equals
  \[
  (\purebare(\tick{\kappa}{\unfoldedtype}) \app \phi
  \laterWBisim{\kappa}{\unfoldedtype}
  (\purebare(\tick{\kappa}{\unfoldedtype}) \app \psi
  \]
  By definition
  $\purebare(\tick{}{\unfoldedtype}) \app \phi$ is equal to
  $\tick{}{\foldedtype} (\phi)$ thus we can derive
  \[
  \tick{\kappa}{\foldedtype}(\phi) \laterWBisim{\kappa}{\unfoldedtype}
  \tick{\kappa}{\foldedtype}(\psi)
  \]
  which by definition of $\wbisim{\kappa}{\foldedtype}$ is
  \[
  \tick{\kappa}{\foldedtype}(\phi) \wbisim{\kappa}{\foldedtype}
  \tick{\kappa}{\foldedtype}(\psi)
  \]
\end{proof}

Next we generalise Lemma~\ref{lem:fpc:liftrel:delay:add:sx:add:dx} to 
hold for $\wbisim{\kappa}{\sigma}$ for all $\sigma$.

\begin{lemma}
  Let $\sigma$ be a closed FPC type and let $x$ and $y$ of type
  $\den{\sigma}$, if $x \wbisim{\kappa}{\sigma} y$ then
  $\delay{\kappa}{\sigma}(x) \wbisim{\kappa}{\sigma} y $ and
  $x \wbisim{\kappa}{\sigma} \delay{\kappa}{\sigma}(y)$.
  \label{lem:fpc:wbisim:delay:add:sx:add:dx}
\end{lemma}
\begin{proof}
  The proof is by guarded recursion and then by induction on the type
  $\sigma$.  Thus, assume this lemma holds ``later'', and proceed 
  by induction on $\sigma$. The cases of the unit type and coproduct 
  follow from Lemma~\ref{lem:fpc:liftrel:delay:add:sx:add:dx}
  and the case of products follows by induction from the fact that 
  $\delay{\kappa}{\tau_i}(\pi_i(x)) = \pi_i(\delay{\kappa}{\tau_1\times \tau_2}(x))$,
  for $i = 1,2$. The case of function types follows from the fact that 
  $\delay{}{\sigma \to \tau}(f)(x) = \delay{}{\tau}(f(x))$.

  For the case of recursive types assume
  $x \wbisim{\kappa}{\foldedtype} y$. Note that
  \begin{align*}
   x \wbisim{\kappa}{\foldedtype} y 
   & \judgeq x \laterWBisim{\kappa}{\unfoldedtype} y \\
   & \judgeq \later{}[\hrt{x' \gets x, y' \gets y}]{x' \wbisim{\kappa}{\unfoldedtype} y'}
  \end{align*}
  Using the dependent version of $\app$ as defined in (\ref{eq:app:gen:def})
  we can apply the guarded recursion assumption to conclude 
  $\later{}[\hrt{x' \gets x, y' \gets y}]{x'
    \wbisim{\kappa}{\unfoldedtype} \delay{}{\unfoldedtype}(y')}$.
  Note that the delay operator is the composition
  $\theta \circ \purebare$, thus $y'$ appears under $\purebare$. We
  can thus employ (\ref{eq:eta}) to derive that 
  $\later{}[\hrt{x' \gets x}]{x' \wbisim{\kappa}{\unfoldedtype}
    \tick{}{\unfoldedtype}(y)}$. From here we conclude by a simple computation:
  \begin{align*}
    \later{}[\hrt{x' \gets x}]{x' \wbisim{\kappa}{\unfoldedtype} \tick{}{\unfoldedtype}(y)} 
    & \judgeq x \laterWBisim{\kappa}{\unfoldedtype} \purebare(
  \tick{}{\unfoldedtype}(y)) \\
    & \judgeq x \laterWBisim{\kappa}{\unfoldedtype} \purebare (\tick{}{\unfoldedtype}) \app \purebare (y) \\
    & \judgeq x \laterWBisim{\kappa}{\unfoldedtype} \tick{}{\foldedtype} (\purebare (y)) \\
    & \judgeq x \wbisim{\kappa}{\foldedtype} \delay{}{\foldedtype}(y) 
  \end{align*}
%
%
\end{proof}

\begin{lemma}
  Let $\sigma$ be a closed FPC type and let $x, y$ of type
  $\denglob{\sigma}$. If $x \wbglob{\sigma} y$ then
  $x \wbglob{\sigma} \delayglob{\sigma} (y)$ and
  $\delayglob{\sigma} (x) \wbglob{\sigma} y$
  \label{lem:fpc:wbglob:delay:add:sx:add:dx}
\end{lemma}
\begin{proof}
  Direct from Lemma~\ref{lem:fpc:wbisim:delay:add:sx:add:dx}.
\end{proof}

\begin{proofof}{Proposition~\ref{prop:fpc:wbisim:refl}}
The proof is by induction on $M$ and we just show the interesting cases.
In all cases we will assume $\Gamma \equiv x_1: \sigma_1 .., x_n: \sigma_n $ and
that we are given $\vec{x}$ and $\vec{y}$ such that 
$\vec{x} \wbisim{\kappa}{\vec{\sigma}} \vec{y}$.   
  
  For case expressions, to prove that
  \[
  \den{\fpccase{L}{M}{N}}(\vec{x}) \wbisim{\kappa}{\tau}
  \den{\fpccase{L}{M}{N}}(\vec{y})
  \]
  it suffices to prove that
  \begin{equation}\den{\lambda x.\fpccase{x}{M}{N}}(\vec{x})
    \wbisim{\kappa}{\sigma
      \to \tau} \den{\lambda x.\fpccase{x}{M}{N}}(\vec{y})
    \label{lem:fpc:wbisim:refl:gr}
  \end{equation}
  Thus that for all $x$, $y$ s.t.
  $x \wbisim{\kappa}{\tau_1 + \tau_2} y$
  \[ \den{\lambda x.\fpccase{x}{M}{N}}(\vec{x})(x)
  \wbisim{\kappa}{\tau} \den{\lambda
    x.\fpccase{x}{M}{N}}(\vec{y})(y)\]
  holds. We prove (\ref{lem:fpc:wbisim:refl:gr}) by guarded
  recursion. Thus, we assume the statement holds ``later'' and we
  proceed by case analysis on $x$ and $y$.  When $x$ is
  $\now{\kappa}{}(x')$ and $y$ is $\now{\kappa}{}(y')$ 
  either $x'$ and $y'$ are both in the left component or they
  are both in the right component of the sum. The former case 
  $x' = \sinl(x'')$ and $y' = \sinl(y'')$
  reduces to 
  \[ \den{M}(\vec{x}, x'')
  \wbisim{\kappa}{\tau} \den{M}(\vec{y},y'')\]
  which follows from the induction hypothesis, and the latter case is similar. 
  
  Now consider the case of $x = \tick{\kappa}{\tau_1 + \tau_2} (x')$ and 
  $y=\now{\kappa}{}(v)$. Since by assumption
  $x \wbisim{\kappa}{\tau_1 + \tau_2} y$ there exists $n$ and $w$ such
  that $x = \delay{\kappa}{\tau_1 + \tau_2}[n] (\now{\kappa}{} (w))$
  and $w \wbisim{}{\tau_1 + \tau_2} v$. As before, $v$ and $w$ must be 
  in the same component of the coproduct, so assume
  $w = \sinl(w')$ and $v = \sinl(v')$ such that $w' \wbisim{}{\tau_1} v'$. 
  By induction hypothesis we know that
  $\den{M}(\vec{x}) \wbisim{\kappa}{\tau_1 \to \tau}
  \den{M}(\vec{y})$
  and thus that
  $\den{M}(\vec{x})(w') \wbisim{\kappa}{\tau} \den{M}(\vec{y})(v')$.
  By Lemma~\ref{lem:fpc:wbisim:delay:add:sx:add:dx} this implies
  $\delay{\kappa}{\tau}[n] (\den{M}(\vec{x})(w'))
  \wbisim{\kappa}{\tau} \den{M}(\vec{y})(v')$.
  Since
  \[\den{M}(\vec{x})(w') = \den{\lambda x.\fpccase{x}{M}{N}}(\vec{x})(\now{}{} (w)),\] 
  by Lemma~\ref{lem:fpc:case:hom} we get 
\begin{align*}
 \delay{\kappa}{\tau}[n] (\den{M}(\vec{x})(w')) 
 & = \delay{\kappa}{\tau}[n] (\den{\lambda x.\fpccase{x}{M}{N}})(\vec{x})(\now{}{} (w))) \\
 & = \den{\lambda x.\fpccase{x}{M}{N}}(\vec{x})(\delay{\kappa}{\tau_1 +
    \tau_2}[n] (\now{}{} (w)))
\end{align*}
and thus we conclude
  \begin{align*}
     \den{\lambda x.\fpccase{x}{M}{N}}(\vec{x})(\delay{\kappa}{\tau_1 + \tau_2}[n]
      (\now{}{} (w))) \wbisim{\kappa}{\tau_1 + \tau_2} 
    \den{\lambda x.\fpccase{x}{M}{N}})(\vec{x})(\now{}{} (v))
  \end{align*}
  which is what we wanted to show. 
  
  The last case is when $x$ is
  $\tick{\kappa}{\tau_1 + \tau_2} (x')$ and $y$ is
  $\tick{\kappa}{\tau_1 + \tau_2} (y')$. By guarded recursion we know
  that
  \[ \later{\kappa}{}(\den{\lambda x.\fpccase{x}{M}{N}}(\vec{x})
  \wbisim{\kappa}{\tau_1 + \tau_2 \to \tau} (\den{\lambda
    x.\fpccase{x}{M}{N}})(\vec{y}))\]
  By (\ref{eq:fpc:laterR:next}) we get
  \begin{align*}
     \purebare (\den{\lambda x.\fpccase{x}{M}{N}}(\vec{x}))
      \laterWBisim{\kappa}{\tau_1 + \tau_2 \to \tau} 
    \purebare(\den{\lambda x.\fpccase{x}{M}{N}}(\vec{y}))
  \end{align*}
  Since the assumption 
  $\tick{\kappa}{\tau_1 + \tau_2} (x') \wbisim{}{\tau_1 + \tau_2}
  \tick{\kappa}{\tau_1 + \tau_2} (y')$,
  means that 
  $x' \laterWBisim{}{\tau_1 + \tau_2} y'$, by
  Lemma~\ref{lem:fpc:wbisim:later:appl} this implies
  \begin{align*}
    \purebare (\den{\lambda x.\fpccase{x}{M}{N}})(\vec{x}) \app x'
      \laterWBisim{\kappa}{\tau} 
    \purebare(\den{\lambda x.\fpccase{x}{M}{N}})(\vec{y}) \app y'
  \end{align*}
  By Lemma~\ref{lem:fpc:wbisim:tick} this implies 
  \begin{align*}
    \tick{}{\tau}(\purebare (\den{\lambda
      x.\fpccase{x}{M}{N}})(\vec{x}) \app x')
      \wbisim{\kappa}{\tau} 
    \tick{}{\tau}(\purebare(\den{\lambda
      x.\fpccase{x}{M}{N}})(\vec{y}) \app y')
  \end{align*}
  By Lemma~\ref{lem:fpc:case:hom} we conclude that
  \begin{align*}
    \den{\lambda
      x.\fpccase{x}{M}{N}}(\vec{x}) (\tick{}{\tau_1 + \tau_2} (x'))
      \wbisim{\kappa}{\tau} 
     \den{\lambda
      x.\fpccase{x}{M}{N}})(\vec{y}) (\tick{}{\tau_1 + \tau_2} (y'))
  \end{align*}
  proving the case.
 
  Finally we prove the two cases for the recursive types. We first consider
  the case for $\fpcunfold{M}$ of type $\unfoldedtype$.  We have to
  show that
  \[
  \den{\fpcunfold{M}}(\vec{x}) \wbisim{\kappa}{\unfoldedtype}
  \den{\fpcunfold{M}}(\vec{y})
  \]
  By induction hypothesis we know that
  $\den{M}(\vec{x}) \wbisim{\kappa}{\foldedtype} \den{M}(\vec{y})$
  which by definition of $\wbisim{\kappa}{\foldedtype}$ is
  $\den{M}(\vec{x}) \laterWBisim{\kappa}{\unfoldedtype}
  \den{M}(\vec{y})$. By Lemma~\ref{lem:fpc:wbisim:tick} we get
  \[
  \tick{\kappa}{\unfoldedtype}(\den{M}(\vec{x}))
  \wbisim{\kappa}{\unfoldedtype}
  \tick{\kappa}{\unfoldedtype}(\den{M}(\vec{y}))
  \]
  and by definition of the interpretation function this is what we
  wanted. 
  
  Now the case for $\fpcfold{M}$ of type $\foldedtype$.  By
  induction hypothesis we know that
  $\den{M}(\vec{x}) \wbisim{\kappa}{\unfoldedtype} \den{M}(\vec{y})$
  which implies 
  $\laterbare (\den{M}(\vec{x}) \wbisim{\kappa}{\unfoldedtype}
  \den{M}(\vec{y}))$
  which is equal to
  \[\purebare (\den{M}(\vec{x})) \laterWBisim{\kappa}{\unfoldedtype}
  \purebare(\den{M}(\vec{y})).\]
  By definition of $\wbisim{\kappa}{\foldedtype}$ this is precisely
  $\purebare (\den{M}(\vec{x})) \wbisim{\kappa}{\foldedtype}
  \purebare(\den{M}(\vec{y}))$
  which by definition of the interpretation function is
  \[\den{\fpcfold{M}}(\vec{x}) \wbisim{\kappa}{\foldedtype}
  \den{\fpcfold{M}}(\vec{y})\]
\end{proofof}

\subsection{Extensional computational adequacy}

Contextual equivalence of FPC is defined in the standard way by
observing convergence at unit type. We first define the language of
contexts. These are FPC programs with a hole $[-]$ defined inductively
as in the next definition.
\begin{definition}[Contexts]
  \begin{align*}
    \ctx  & := [-] \mid \lambda x. \ctx \mid \ctx N \mid M \ctx \\
          & \mid \fpcinl{\ctx} \mid \fpcinr{\ctx} \mid \fpcpair{\ctx}{M}
            \mid \fpcpair{M}{\ctx} \mid \fpcfst{\ctx} \mid \fpcsnd{\ctx} \\
          & \mid \fpccase{\ctx}{M}{N} \\ 
          & \mid \fpccase{L}{\ctx}{N} \mid \fpccase{L}{M}{\ctx}\\
          & \mid \fpcunfold{\ctx} \mid \fpcfold{\ctx} 
  \end{align*}
\end{definition}

Intuitively, a context is a term that takes a term and returns a new term. 

We define the ``fill hole'' function $ \cdot [\cdot]: \ctx \times \FPCOTerms \to \FPCOTerms$ 
by induction on the context in the standard way. Note that 
this may capture free variables in the term being substituted.

We say that a context $C$ has type $(\Gamma, \sigma) \to (\Delta, \tau)$ if 
$\hastype{}{\Delta}{C[M]}{\tau}$ whenever $\hastype{}{\Gamma}{M}{\sigma}$. This can be captured by
a typing relation on contexts as defined in Figure~\ref{fig:fpc:context:typing}.
Next we define contextual equivalence using the big-step semantics
$\bigstep$. This states that two program are contextually equivalent
if no context can distinguish them. Using $\bigstep$ (instead of
$\bigstep^k$) ensures that we capture the standard notion of
contextual equivalence, thus that two programs producing the same
value will be equivalent no matter how many steps they take to
terminate. 




\begin{figure}[tb]
  \begin{mathpar}
    \inferrule* { \; }{\ctxhastype{-}{\Gamma}{\tau}{\Gamma}{\tau} }
    \and \inferrule* {\ctxhastype{C}{\Gamma}{\tau}{(\Delta, x:
        \sigma')}{\sigma} } {\ctxhastype{(\lambda
        x.C)}{\Gamma}{\tau}{\Delta}{\sigma' \to \sigma} } \and
    \inferrule* { \ctxhastype{C}{\Gamma}{\tau}{\Delta}{\tau' \to
        \sigma} \qquad \hastype{}{\Delta}{N}{\tau'} }{ \ctxhastype{C
        N}{\Gamma}{\tau}{\Delta}{\sigma} } \and \inferrule* {
      \ctxhastype{C}{\Gamma}{\sigma}{\Delta}{\tau'} \qquad
      \hastype{}{\Delta}{M}{\tau' \to \sigma} }{ \ctxhastype{M
        C}{\Gamma}{\sigma}{\Delta}{\sigma}} \and \inferrule*
    {\ctxhastype{C}{\Gamma}{\sigma}{\Delta}{\foldedtype}}{
      \ctxhastype{\fpcunfold{C}}{\Gamma}{\sigma}{\Delta}{\unfoldedtype}
    } \and \inferrule*
    {\ctxhastype{C}{\Gamma}{\sigma}{\Delta}{\unfoldedtype}}{
      \ctxhastype{\fpcfold{C}}{\Gamma}{\sigma}{\Delta}{\foldedtype} }
    \and \inferrule* {\ctxhastype{C}{\Gamma}{\tau}{\Delta}{\tau_1
        \times
        \tau_2}}{\ctxhastype{\fpcfst{C}}{\Gamma}{\tau}{\Delta}{\tau_1}}
    \qquad \inferrule* {\ctxhastype{C}{\Gamma}{\tau}{\Delta}{\tau_1
        \times
        \tau_2}}{\ctxhastype{\fpcsnd{C}}{\Gamma}{\tau}{\Delta}{\tau_2}}
    \and \inferrule* {\ctxhastype{C}{\Gamma}{\tau}{\Delta}{\tau_1}
      \qquad
      \hastype{}{\Delta}{N}{\tau_2}}{\ctxhastype{\fpcpair{C}{N}}{\Gamma}{\tau}{\Delta}{{\tau_1}\times{\tau_2}}}
      \and
    \inferrule* {\ctxhastype{C}{\Gamma}{\tau}{\Delta}{\tau_2} \qquad
      \hastype{}{\Delta}{M}{\tau_1}}{
      \ctxhastype{\fpcpair{M}{C}}{\Gamma}{\tau}{\Delta}{{\tau_1}\times{\tau_2}}}
    \and \inferrule* {\ctxhastype{C}{\Gamma}{\tau}{\Delta}{\tau_1 +
        \tau_2} \qquad \hastype{}{\Delta, x_1 : \tau_1}{M}{\sigma}
      \qquad \hastype{}{\Delta, x_2 :
        \tau_2}{N}{\sigma}}{\ctxhastype{\fpccase{C}{M}{N}}{\Gamma}{\tau}{\Delta}{\sigma}}
    \and \inferrule* { \hastype{}{\Delta}{L}{\tau_1 + \tau_2} \qquad
      \ctxhastype{C}{\Gamma}{\tau}{(\Delta, x_1 : \tau_1)}{\sigma}
      \qquad \hastype{}{\Delta, x_2 :
        \tau_2}{N}{\sigma}}{\ctxhastype{\fpccase{L}{C}{N}}{\Gamma}{\tau}{\Delta}{\sigma}}
    \and \inferrule*{ \hastype{}{\Delta}{L}{\tau_1 + \tau_2} \qquad
      \hastype{}{\Delta, x_1 : \tau_1}{M}{\sigma}\qquad
      \ctxhastype{C}{\Gamma}{\tau}{(\Delta, x_2 : \tau_2)}{\sigma}
    }{\ctxhastype{\fpccase{L}{M}{C}}{\Gamma}{\tau}{\Delta}{\sigma}}
    \and \inferrule*{\ctxhastype{C}{\Gamma}{\tau}{\Delta}{\tau_1 }}{
      \ctxhastype{\fpcinl{C}}{\Gamma}{\tau}{\Delta}{\tau_1 + \tau_2}}
    \qquad \inferrule*{\ctxhastype{C}{\Gamma}{\tau}{\Delta}{\tau_2 }}{
      \ctxhastype{\fpcinr{C}}{\Gamma}{\tau}{\Delta}{\tau_1 + \tau_2}}
  \end{mathpar}
  \caption{Typing judgment for contexts}
  \label{fig:fpc:context:typing}
\end{figure}

\begin{definition}
  \label{def:fpc:ctxeq}
  Let $\hastype{}\Gamma {M,N}\tau$. We say that $M,N$ are contextually
  equivalent, written $M \ctxeq N$, if for all contexts $C$ of type
  $(\Gamma,\tau) \to (-, \fpcunittype)$
  \[ C[M] \bigstep \fpcunit \iff C[N] \bigstep \fpcunit \]
\end{definition}

\newcommand{\wbisimbare}{\wbisim{}{}}
Finally we can state the main theorem of this section. Using the
global view of the logical relation $\wbisimbare$ we can prove if the
denotations of two programs are related then they are contextual
equivalent in the extensional sense. 
\begin{theorem}[Extensional Computational Adequacy]  \label{thm:ext:adequacy}
  If $\Gamma \vdash M,N : \tau$ and
  $\denglob{M} \wbglob{\Gamma, \tau} \denglob{N}$ then $M \ctxeq N$. 
\end{theorem}

To prove this theorem, we need the following lemma stating that contexts
preserve the logical relation. 

\begin{lemma} \label{lem:fpc:wbisim:ctx:gen}
  Let $\hastype{}{\Gamma}{M}{\tau}$ and $\hastype{}{\Gamma}{N}{\tau}$ and suppose
  $\den{M} \wbisim{\kappa}{\Gamma,\tau} \den{N}$. If $C$ is a context such that 
  $C : \Gamma, \tau \to \Delta,\sigma$ then 
  $\den{C[M]} \wbisim{\kappa}{\Delta, \sigma} \den{C[N]}$
\end{lemma}

\begin{proof}
The proof is by induction on $C$ and most cases can be proved either 
very similarly to corresponding cases of Proposition~\ref{prop:fpc:wbisim:refl},
or by direct application of Proposition~\ref{prop:fpc:wbisim:refl}. We show how to 
do the latter in two cases. 

  For a context
  $\fpcunfold{C}$ of type
  $ (\Gamma, \sigma) \to (\Delta, \unfoldedtype)$ we have by induction
  that $C$ has type $(\Gamma, \sigma) \to (\Delta, \foldedtype)$ and
  thus induction hypothesis we know that
  $\den{C[M]}(\vec{x}) \wbisim{\kappa}{\foldedtype}
  \den{C[N]}(\vec{y})$.
  By Proposition~\ref{prop:fpc:wbisim:refl} we know that
  \[\den{\lambda x. \fpcunfold{x}} \wbisim{\kappa}{(\foldedtype) \to
    (\unfoldedtype)} \den{\lambda x. \fpcunfold{x}}\]
  By applying this latter fact to the induction hypothesis we obtain
  \[\den{\fpcunfold{C[M]}}(\vec{x}) \wbisim{\kappa}{\unfoldedtype}
  \den{\fpcunfold{C[N]}}(\vec{y})\]
  which is what we wanted.

  When the context binds a variable one has to be a bit more careful. 
  For example, for a context
  of the form $\fpccase{L}{C}{N'}$ of type
  $(\Gamma, \tau) \to (\Delta, \sigma)$ we have by induction that $C$
  has type $(\Gamma, \tau) \to ((\Delta, x_1 \co \tau_1), \sigma)$ and thus by
  induction hypothesis we know by applying the context parameters that
  $\den{C[M]}(\vec{x}) \wbisim{\kappa}{\tau_1, \sigma}
  \den{C[N]}(\vec{y})$.
  From this we also know that
  \begin{equation} \label{eq:bisim:lambda:C[M]}
\den{\lambda x_1. C[M]}(\vec{x}) \wbisim{\kappa}{\tau_1 \to
    \sigma} \den{\lambda x_1.C[N]}(\vec{y}).
\end{equation}
  By Proposition~\ref{prop:fpc:wbisim:refl} we know that
  \[\den{\lambda
    x. \fpccase{L}{x(x_1)}{N'}}(\vec{x})\wbisim{\kappa}{(\tau_1 \to
    \sigma) \to \sigma} \den{\lambda
    x. \fpccase{L}{x(x_1)}{N'}}(\vec{y}). \]
  By applying this to (\ref{eq:bisim:lambda:C[M]}) we
  conclude.
%
%
\end{proof}

As a direct consequence we get the following lemma. 

\begin{lemma}
  If $\hastype{}\Gamma{M,N}{\tau}$ and
  $\denglob{M} \wbglob{\Gamma,\tau} \denglob{N}$ then for all contexts
  $C$ of type $(\Gamma,\tau) \to (-, \fpcunittype)$,
  $\denglob{C[ M ]} \wbglob{(-,\fpcunittype)} \denglob{C[ N ]} $
  \label{lem:fpc:wbglob:ctx}
\end{lemma}

The next lemma states that if two computations of unit type are
related then the first converges iff the second converges.  Note that
this lemma needs to be stated using the fact that the two computations
are \emph{globally related}.
\begin{lemma} \label{lem:fpc:wb:den:gen} For all $x, y$ of type
  $\denglob{\fpcunittype}$, if $x \wbglob{(-,\fpcunittype)} y$ then
  \[\Sigma n. x \propeq (\delayglob{\fpcunittype})^{n} (\eta(\sunit))
  \Leftrightarrow \Sigma m. y \propeq (\delayglob{\fpcunittype})^m
  (\eta(\sunit))\]
\end{lemma}
\begin{proof}  We show the left to right implication, so
  suppose $x \propeq (\delayglob{\fpcunittype})^{n} (\eta(\sunit))$.
  The proof proceeds by induction on $n$. If $n\propeq 0$ then since
  by assumption
  $\forall \kappa. \alwaysapp{x}{\kappa}\wb{\fpcunittype}
  \alwaysapp{y}{\kappa}$,
  by definition of $\wb{\fpcunittype}$, for all $\kappa$, there exists
  an $m$ such that
  $\alwaysapp{y}{\kappa} \propeq \delay{}{\fpcunittype}^m
  (\eta(\sunit))$.
  By type isomorphism (\ref{eq:fpc:forallk:sigma:swap}), since $m$ is
  a natural number, this implies there exists $m$ such that for all
  $\kappa$,
  $\alwaysapp{y}{\kappa} \propeq \delay{}{\fpcunittype}^m
  (\eta(\sunit))$
  which implies $y \propeq \alwaysterm\kappa{\alwaysapp y\kappa} 
  \propeq (\delayglob{\fpcunittype})^m(\eta(\sunit))$.

  In the inductive case $n \propeq n' + 1$, since by
  Lemma~\ref{lem:fpc:liftrelglob:delay:rm:sx}
  $(\delayglob{\fpcunittype})^{n'} (\denglob{v}) \wbglob{\fpcunittype}
  y$,
  the induction hypothesis implies
  $\Sigma m. y \propeq (\delayglob{\fpcunittype})^m (\eta(\sunit))$.
\end{proof}

\begin{proof}[Proof of Theorem~\ref{thm:ext:adequacy}]
  Suppose $\denglob{M} \wbglob{\Gamma, \tau} \denglob{N}$ and that $C$
  has type $(\Gamma,\tau) \to (-, \fpcunittype)$. We show that if
  $C[M] \Downarrow \fpcunit$ also $C[N] \Downarrow \fpcunit$. So
  suppose $C[M] \Downarrow \fpcunit$.  By definition this means
  $\Sigma n. C[M] \Downarrow^n \fpcunit$.  By 
  Corollary~\ref{cor:fpc:gl:adequacy} we get
  $\Sigma n .\forall \kappa. \den{C[M]} \propeq (\delay{}{\fpcunittype})^{n}
  (\eta(\sunit))$
  which is equivalent to
  $\Sigma n . \denglob{C[M]} \propeq (\delayglob{\fpcunittype})^{n}
  (\eta(\sunit))$.
  From the assumption and Lemma~\ref{lem:fpc:wbglob:ctx} we know that 
  $\denglob{C[ M ]} \wbglob{\fpcunittype} \denglob{C[ N ]}$, so by
  Lemma~\ref{lem:fpc:wb:den:gen} there exists an $m$ such that
  $\denglob{C[N]} \propeq (\delayglob{\fpcunittype})^m (\eta(\sunit))$.
  By applying the Corollary~\ref{cor:fpc:gl:adequacy}
  once again we get $C[N] \Downarrow \fpcunit$ as desired.
\end{proof}
\section{Executing the denotational semantics}
\label{sec:executing}

In this final section we sketch an additional benefit of the denotational
semantics described in this paper: The denotational 
semantics can be executed. More precisely, given a closed
FPC term of base type and a number $n$, the denotational semantics
can be executed up to $n$ steps. This will terminate if and only if the 
big-step operational semantics terminates in $n$ steps or less. 
The time-out $n$ is necessary since FPC programs can
diverge and programs in type theory must terminate. We emphasize that 
at the moment there is no full implementation of \gdtt\, and so the practical
implications of this section are speculative. 

We illustrate the execution 
of the denotational semantics in the case of programs computing booleans, i.e.,
closed term of type $\fpcunittype + \fpcunittype$. The global interpretation
of such a term has type $\denglob{\fpcunittype + \fpcunittype} = \forall\kappa . L(L1 + L1)$. 
We first define a term
\[
\runstep : (\forall\kappa . L(L1 + L1)) \to (1 + 1) + (\forall\kappa . L(L1 + L1))
\]
running the denotation of the term for one step. We define $\runstep\, x$ by cases of 
$\alwaysapp x{\kappa_0} : L(L1 + L1)\subst{\kappa_0}\kappa$ 
where $\kappa_0$ is the clock constant. If $\alwaysapp x{\kappa_0} = \eta(\inl(y))$ for some $y$,
then $\runstep\, x = \inl(\inl(\star))$, and likewise 
if $\alwaysapp x{\kappa_0} = \eta(\inr(y))$ for some $y$,
then $\runstep\, x = \inl(\inr(\star))$. In case $\alwaysapp x{\kappa_0}$ is of the form
$\tick{}{}(y)$, then, as we saw in the construction of the isomorphism (\ref{eq:forall:dist:sum}) 
in Section~\ref{sec:clock-variables}, there is a term $z_\kappa$
such that $\alwaysapp x{\kappa} = \tick{}{}(z_\kappa)$. (Precisely, 
$z_\kappa = \pi_2(\alwaysapp x{\kappa})$ using the encoding
of binary sums as dependent sums over $1+1$.) In that case we define 
$\runstep x = \inr(\prev\kappa z_{\kappa})$. 

Using $\runstep$ we can define a function 
\begin{align*}
 \exec & :\NN \to (\forall\kappa . L(L1 + L1)) \to (1 + 1) + (\forall\kappa . L(L1 + L1))
\end{align*}
such that $\exec\,n$ iterates $\runstep$ until it gets a result, or for at most $n+1$ times. Precisely,
we define $\exec\, 0 \, x  = \runstep \, x$ and $\exec\, (n+1) \, x = \runstep \,x$ if $\runstep\, x$ is in the left component
and $\exec\, (n+1) \, x = \exec\, n \, y$ if $\runstep \, x = \inr(y)$. 

We now show that executing the denotational semantics using $\exec\, n$ corresponds to executing the operational
semantics for up to $n$ steps.  

\begin{proposition} \label{prop:exec}
 Let $M$ be a closed term of FPC of type $\fpcunittype + \fpcunittype$,
 and let $n$ be a natural number. Then 
 $\exec\, n\, \denglob{M} = \inl(\inl(\star))$ iff there exists an $N$ such that 
 $M\Downarrow^k \fpcinl(N)$ for some $k\leq n$. 
\end{proposition}

To prove Proposition~\ref{prop:exec} we need following two lemmas. 

%

\begin{lemma} \label{lem:exec:term:delay}
 If $\exec\, n\, x = \inl(\inl(\star))$ then there exists a $k \leq n$ and a $y$
 such that $x = (\delayglob{\fpcunittype + \fpcunittype})^k(\Lambda \kappa . \eta( \inl(\alwaysapp y\kappa)))$.
\end{lemma}

\begin{proof}
 The proof is by induction on $n$ and case analysis of $\alwaysapp x{\kappa_0}$. If $\alwaysapp x{\kappa_0} 
 = \eta(\inl(y))$ for some $y$, then, as above, also $\alwaysapp x{\kappa} = \eta(\inl(z_\kappa))$ for 
 some $z_\kappa$ and so $x = \Lambda\kappa . \eta(\inl(z_\kappa))$ proving the lemma. 

 If $\alwaysapp x{\kappa_0} = \eta(\inr(y))$, then also $\exec\, n\, x = \inl(\inr(\star))$. Comparing this with the assumption
 we get $\inl(\inr(\star)) = \inl(\inl(\star))$. Recall~\cite[Section~2.12]{hottbook} that $\inl(\inr(\star)) = \inl(\inl(\star))$
 is equivalent to $\inr(\star) = \inl(\star)$ which is equivalent to the empty type, so 
 from this we conclude $0$ and thus anything is provable. 

 Suppose finally that  $\alwaysapp x{\kappa_0} = \tick{}{}(y)$. Then $\alwaysapp x{\kappa} = \tick{}{}(z_\kappa)$, 
 and $\runstep x = \inr(\prev\kappa z_{\kappa})$. In this case $n$ must be greater than $0$, i.e., $n = m+1$, and 
 $\exec\, (m+1) \, x = \exec\, m \, (\prev\kappa z_{\kappa})$. In this case, by induction hypothesis, 
 $\prev\kappa z_{\kappa} = (\delayglob{\fpcunittype + \fpcunittype})^k(\Lambda \kappa . \eta( \inl(\alwaysapp y\kappa)))$
 for some $y$ and $k\leq n$. So then,
\begin{align*}
 x & = \Lambda \kappa . (\alwaysapp x\kappa) \\
 & = \Lambda \kappa . (\tick{}{} (z_\kappa)) \\
 & = \Lambda \kappa . (\tick{}{\fpcunittype + \fpcunittype} \pure{}^\kappa(\alwaysapp{(\prev\kappa z_\kappa)}\kappa) \\
 & = \Lambda \kappa . (\delay{}{{\fpcunittype + \fpcunittype}}{} (\alwaysapp{(\delayglob{\fpcunittype + \fpcunittype})^k(\Lambda \kappa . \eta( \inl(\alwaysapp y\kappa)))}\kappa) \\
 & = (\delayglob{\fpcunittype + \fpcunittype}{})^{k+1}(\Lambda \kappa . \eta( \inl(\alwaysapp y\kappa))
\end{align*}
\end{proof}

\begin{lemma} \label{lem:delay:1+1:term}
 Let $M$ be a closed term of FPC of type $\fpcunittype + \fpcunittype$. If 
 $\denglob M = (\delayglob{\fpcunittype + \fpcunittype})^k(\Lambda \kappa . \eta( \inl(\alwaysapp y\kappa)))$ 
 then there exists an $N$ such that $M \Downarrow^k \fpcinl(N)$. 
\end{lemma}

\begin{proof}
We prove by induction on $k$ that if 
\[
\forall\kappa . \delay{}{\fpcunittype + \fpcunittype}^k(\eta( \inl(\alwaysapp y\kappa)))
 \R{\tau_1 + \tau_2} M
\]
then there exists an $N$ such that $M \Downarrow^k \fpcinl(N)$. The lemma then follows from the Fundamental Lemma
(Lemma~\ref{lem:fpc:fundamental}). In the case of $k= 0$, by definition the assumption implies 
\[
\forall\kappa . \Sigma N . M \Downarrow^0 \fpcinl(N)
\]
which by an application to the clock constant $\kappa_0$ implies 
\[
\Sigma N . M \Downarrow^0 \fpcinl(N)
\]
as desired. 
If $k = l+1$, the assumption
$\forall\kappa . \tick{}{}(\pure{}^\kappa(\delay{}{\fpcunittype + \fpcunittype}^{l}(\eta( \inl(\alwaysapp y\kappa))))) \R{\tau_1 + \tau_2} M$ reduces to 
 \[
   \forall\kappa . (\Sigma M',M''\co \FPCTerms \ld M \to_*^0 M' \to^1 M''   
   \text{ and } \pure{}^\kappa(\delay{}{\fpcunittype + \fpcunittype}^l(\eta( \inl(\alwaysapp y\kappa))))
 \laterR{\tau_1 + \tau_2} \pure{}{(M'')})
 \]
 This implies 
 \[
   \Sigma M',M''\co \FPCTerms \ld M \to_*^0 M' \to^1 M''   
   \text{ and } \forall\kappa . \laterclock\kappa(\delay{}{\fpcunittype + \fpcunittype}^l(\eta( \inl(\alwaysapp y\kappa)))
 \R{\tau_1 + \tau_2} M'')
 \]
 which, using $\force$ implies
 \[
   \Sigma M',M''\co \FPCTerms \ld M \to_*^0 M' \to^1 M''   
   \text{ and } \forall\kappa . \delay{}{\fpcunittype + \fpcunittype}^l(\eta( \inl(\alwaysapp y\kappa)))
 \R{\tau_1 + \tau_2} M''
 \]
 Now the induction hypothesis applies to give an $N$ such that $M'' \Downarrow^l \fpcinl(N)$, which
 by Lemma~\ref{lem:fpc:bigstep:many:manyk:soundness}  implies $\many{M''}{l}{v}$ and thus $\many{M}{k}{v}$
 which implies $M \Downarrow^k \fpcinl(N)$ again by Lemma~\ref{lem:fpc:bigstep:many:manyk:soundness}. 
\end{proof}

\begin{proofof}{Proposition~\ref{prop:exec}}
The left to right implication follows from Lemmas~\ref{lem:exec:term:delay} and~\ref{lem:delay:1+1:term}. 
If $M\Downarrow^k \fpcinl(N)$ for some $k\leq n$, then 
$\denglob M = (\delayglob{\fpcunittype + \fpcunittype})^k(\Lambda\kappa . \eta(\inl (\den N)))$.
We prove that this implies that 
$\exec\, n\, \denglob M  = \inl(\inl(\star))$ 
by induction on $k$. The case of $k=0$ 
follows directly by definition of $\exec$. If $k = l + 1$ also $n= m+1$ for some $m$. Observe now that for any 
$x : \denglob{\fpcunittype + \fpcunittype}$ 
\begin{align*}
 \runstep \, \delayglob{\fpcunittype + \fpcunittype}(x) & = 
 \runstep \, \Lambda\kappa.(\tick{}{}(\pure{}^\kappa(\alwaysapp x\kappa))) \\
 & = \inr(\prev\kappa (\pure{}^\kappa(\alwaysapp x\kappa))) \\ 
 & = \inr(\Lambda\kappa . \alwaysapp x\kappa) \\
 & = \inr(x)
\end{align*}
and so in particular
\[
\runstep \,\denglob M = \inr({(\delayglob{\fpcunittype + \fpcunittype})^l(\Lambda\kappa . \eta(\inl (\den N)))})
\]
so that 
\begin{align*}
 \exec\, (n+1)\, \denglob M & = \exec\, n\, (\delayglob{\fpcunittype + \fpcunittype})^l(\Lambda\kappa . \eta(\inl (\den N)))
\end{align*}
which equals $\inl(\inl(\star))$ by the induction hypothesis.
\end{proofof}

\section{Conclusions and Future Work}
\label{sec:conclusions}

We have shown that programming languages with recursive types can be
given sound and computationally adequate denotational semantics in
guarded dependent type theory. The semantics is intensional, in the
sense that it can distinguish between computations computing the same
result in different number of steps, but we have shown how to capture
extensional equivalence in the model by constructing a logical
relation on the interpretation of types.

This work can be seen as a first step towards a formalisation of
domain theory in type theory. Other, more direct formalisations have
been carried out in Coq, e.g. \cite{BKV09, BBKV10,Dockins14} but we
believe that the synthetic viewpoint offers a more abstract and
simpler presentation of the theory.  Moreover, we hope that the
success of guarded recursion for operational reasoning, mentioned in
the introduction, can be carried over to denotational models of more
advanced programming language features as, for example, to general
references, for which, at the present day, no denotational model
exists. 

Future work also includes implementation of $\gdtt$\ in a proof
assistant, allowing for the theory of this paper to be machine
verified. Currently, initial experiments are being carried out in this
direction~\cite{BBCGV16}.

Finally, we have not yet investigate the possible applications of the
weak bisimulation introduced in Section~\ref{sec:extensional}. 

\printbibliography[heading=bibintoc]

\label{lastpage}
\end{document}